\documentclass[amsthm]{elsart}
\usepackage{yjsco}

\usepackage{amssymb, amsbsy, amsfonts}
\usepackage{amsmath, amsthm}
\usepackage{latexsym,float,color}
\usepackage{makeidx}         % allows index generation
\usepackage{multicol}

\let\set\mathbb
\def\vect#1{\mathbf{#1}}
\def\lcm{\operatorname{lcm}}
\def\mspan{\operatorname{span}}

\def\ldeg{\operatorname{ldeg}}
\def\ord{\operatorname{ord}}
\def\ford{\operatorname{ford}}
\def\per{\operatorname{per}}

\def\rE{$R$}
\def\piE{$\Pi$}
\def\pisiK{$\Pi\Sigma$}
\def\rpisiSE{$R\Pi\Sigma^*$}
\def\pisiSE{$\Pi\Sigma^*$}
\def\sigmaSE{$\Sigma^*$}

\def\KK{\set K}
\def\NN{\set N}
\def\ZZ{\set Z}
\def\GG{\set G}
\def\HH{\set H}
\def\FF{\set F}
\def\EE{\set E}
\def\QQ{\set Q}
\def\AA{\set A}

\newenvironment{theindex}
  {\renewcommand\item{\par\hangindent 40pt}
   \newcommand\subitem{\item\hspace*{20pt}}
   
   \newcommand\indexspace{\par \vskip 10pt plus 5pt minus 3pt\relax}}
  {}

\newcommand{\ExternalProof}[2]{\begin{proofstep}\noindent\textbf{#1. }#2\hfill\ensuremath{
\Box}\end{proofstep}}

\newcommand{\fct}[3]{#1\colon #2 \to #3}
\newcommand{\dfield}[2]{({#1},{#2})}
\newcommand{\constF}[2]{{\rm const}{(#1,#2)}}
\newcommand{\const}[2]{{\rm const}{\:\!#1}}
\newcommand{\sconstF}[3]{{\rm sconst}_{#1}{(#2,#3)}}
\newcommand{\sconst}[3]{{\rm sconst}_{#1}{#2}}
\newcommand{\gsconstF}[2]{{\rm sconst}{(#1,#2)}}
\newcommand{\gsconst}[2]{{\rm sconst}\:\!#1}
\newcommand{\sigmaFac}[2]{{#1}_{(#2)}}
\newcommand{\dgroup}[3]{{#1}^{#2}_{#3}}

\newcommand{\lr}[1]{\langle #1\rangle}
\newcommand{\ltr}[1]{[#1,\tfrac{1}{#1}]}

\newdimen\listablecorrection
\newtheorem{theorem}{Theorem}[section]
\newtheorem{proofstep}[theorem]{Proof}
\newtheorem{proposition}[theorem]{Proposition}
\newtheorem{corollary}[theorem]{Corollary}
\newtheorem{lemma}[theorem]{Lemma}
\newtheorem{remark}[theorem]{Remark}
\newtheorem{definition}[theorem]{Definition}
\newtheorem{example}[theorem]{Example}

\makeindex

\begin{document}

\begin{frontmatter}

\title{A Difference Ring Theory for\\ Symbolic Summation}

\thanks{Supported by the Austrian Science Fund (FWF) grant SFB F50 (F5009-N15) and the European
Commission through contract PITN-GA-2010-264564 ({LHCPhenoNet}).}

\author{Carsten Schneider}
\address{Research Institute for Symbolic Computation (RISC)\\ 
Johannes Kepler University\\
Altenbergerstra{\ss}e 69, 4040 Linz, Austria}
\ead{Carsten.Schneider@risc.jku.at}

\begin{abstract}
A summation framework is developed that enhances Karr's difference field approach. It covers not only indefinite nested sums and products in terms of transcendental extensions, but it can treat, e.g., nested products defined over roots of unity. The theory of the so-called \rpisiSE-extensions is supplemented by algorithms that support the construction of such 
difference rings automatically and that assist in the task to tackle symbolic summation problems. Algorithms are presented that solve parameterized telescoping equations, and more generally parameterized first-order difference equations, in the given difference ring. As a consequence, one obtains algorithms for the summation paradigms of telescoping and Zeilberger's creative telescoping.
With this difference ring theory one gets a rigorous summation machinery that has been applied to numerous challenging problems coming, e.g., from combinatorics and particle physics.

\end{abstract}


\end{frontmatter}

\section{Introduction}

In his pioneering work~\cite{Karr:81,Karr:85} M.~Karr introduced a very general class of difference fields, the so-called \pisiK-fields, in which expressions in terms of indefinite nested sums and products can be represented. 
In particular, he 
developed an algorithm that decides constructively if for a given expression
$f(k)$ represented in a \pisiK-field $\FF$ there is an expression $g(k)$
represented in the field $\FF$ such that the telescoping equation (anti-difference)
\begin{equation}\label{Equ:TeleProblem}
f(k)=g(k+1)-g(k)
\end{equation}
holds. If such a solution exists, one obtains for an appropriately chosen $a\in\NN$ the identity
\begin{equation}\label{Equ:SumTele}
\sum_{k=a}^b f(k)=g(b+1)-g(a).
 \end{equation}
\noindent His algorithms can be viewed as the discrete version of Risch's integration algorithm; see~\cite{Risch:69,Bron:97}.
In the last years the \pisiK-field theory has been pushed forward. It is now possible to obtain sum representations, i.e., right hand sides in~\eqref{Equ:SumTele} with certain optimality criteria such as
minimal nesting depth~\cite{Schneider:08c,Schneider:10b}, minimal number of
generators in the summands~\cite{Schneider:04a} or minimal degrees in the
denominators~\cite{Schneider:07d}. For the simplification of products see~\cite{Schneider:05c,Petkov:10}. We emphasize that exactly such refined representations give rise to more efficient telescoping algorithms worked out in~\cite{Schneider:10a,Schneider:14}.

A striking application is that Karr's algorithm and all the enhanced versions can be used to solve the parameterized telescoping problem~\cite{Schneider:00,Schneider:10c}:
for given indefinite nested product-sum expressions $f_1(k),\dots,f_n(k)$ represented in $\FF$, find constants $c_1,\dots,c_n$, free of $k$ and not all zero,
and find $g(k)$ represented in $\FF$  such that
\begin{equation}\label{Equ:ProblemPT}
g(k+1)-g(k)=c_1\,f_1(k)+\dots+c_n\,f_n(k)
\end{equation} holds. In particular, this problem covers Zeilberger's creative telescoping paradigm~\cite{Zeilberger:91} for a bivariate function $F(m,k)$ by setting $f_i(k)=F(m+i-1,k)$ with $i\in\{1,\dots,n\}$ and representing these $f_i(k)$ in $\FF$. Namely, if one finds such a solution,
one ends up at the recurrence
$$g(m,b+1)-g(m,a)=c_1\,\sum_{k=a}^bf(m,k) +\dots+c_n\,\sum_{k=a}^bf(m+n-1,k).$$
In a nutshell, one cannot only treat indefinite summation but also
definite summation problems. In this regard, also recurrence solvers have been
developed where the coefficients of the recurrence and the inhomogeneous part
can be elements from a \pisiK-field~\cite{Bron:00,Schneider:05a,ABPS:14}. All these algorithms generalize and enhance substantially the ($q$--)hyper\-geometric and holonomic
toolbox~\cite{Abramov:71,Gosper:78,Zeilberger:90a,Zeilberger:91,Petkov:92,Paule:95,AequalB,PauleRiese:97,Bauer:99,Chyzak:00,KP:11,Koutschan:13} in order to rewrite
definite sums to indefinite nested sums. For details on these aspects we refer
to~\cite{Schneider:13b}.

Besides all these sophisticated developments, e.g., within the summation package \texttt{Sigma}~\cite{Schneider:07a}, there is one critical gap which concerns all the developed tools in the setting of difference fields: Algebraic products, like
\begin{equation}\label{Equ:Equ:AlgebraicObjects}
(-1)^k=\prod_{i=1}^k(-1),\quad(-1)^{\binom{k+1}{2}}=\prod_{i=1}^k\prod_{j=1}^{i}(-1),\quad(-1)^{\binom{k+2}{3}}=\prod_{i=1}^k\prod_{j=1}^{i}\prod_{k=1}^j(-1),\dots
\end{equation}
cannot be expressed in \pisiK-fields, which are built by a tower of transcendental field extensions. Even worse, the objects given in~\eqref{Equ:Equ:AlgebraicObjects} introduce zero-divisors,
like
\begin{equation}\label{Equ:ZeroDivisorRel}
(1-(-1)^k)(1+(-1)^k)=0
\end{equation}
which cannot be treated in a field or in an integral domain. In applications these
objects occur rather frequently as standalone objects or in nested sums~\cite{ABS:11,ABS:13}. It is thus a fundamental
challenge to include such objects in an enhanced summation theory.

With the elegant theory of~\cite{Singer:97,Singer:08} one can handle such objects by several
copies of the underlying difference field, i.e., by implementing the concept of interlacing in an algebraic way. First steps to combine these techniques with \pisiK-fields have been made in~\cite{Erocal:11}. 

Within the package \texttt{Sigma} a different approach~\cite{Schneider:01} has been implemented. Summation objects like $(-1)^k$ and sums over such objects are introduced by a tower of generators subject to the relations such
as~\eqref{Equ:ZeroDivisorRel}. In
this way one obtains a direct translation between the summation objects
and the generators of the corresponding difference rings. This enhancement has been applied non-trivially, e.g., to combinatorial problems~\cite{Schneider:04c,PSW:11}, number theory~\cite{Schneider:07b,Schneider:09a} or to problems from particle physics~\cite{BKSF:12}; for the most recent evaluations of Feynman integrals~\cite{Physics1,Physics2,Physics3} up to 300 generators were used to model the summation objects in difference rings. But so far, this successful and very efficient machinery of \texttt{Sigma} was built, at least partially, on heuristic considerations.

In this article we shall develop the underlying difference ring theory and supplement it with the missing algorithmic building blocks in order to obtain a rigorous summation machinery. More precisely, we will enhance
the difference field theory of~\cite{Karr:81,Karr:85} to a difference ring theory by introducing
besides \piE-extensions (for transcendental product extensions) and
\sigmaSE-extensions (for transcendental sum extensions) also \rE-extensions
which enables one to represent objects such
as~\eqref{Equ:Equ:AlgebraicObjects}. 
An important ingredient of this theory is the exploration of the so-called semi-constants (resp.\ semi-invariants) and the formulation of the symbolic summation problems within these notions. 
In particular, we obtain algorithms that can solve certain classes of parameterized first-order linear difference equations. As special instances we obtain algorithms for the parameterized telescoping problem, in particular for the summation
paradigms of telescoping and creative telescoping. 
In addition, we
provide an algorithmic toolbox that supports the construction of the so-called  simple \rpisiSE-extensions automatically. As a special case we demonstrate, how d'Alembertian solutions~\cite{Abramov:94} of a recurrence, a subclass of Liouvillian solutions~\cite{Singer:99,Petkov:2013}, can be represented in such \rpisiSE-extensions. In particular, 
we will illustrate the underlying problems and their solutions by discovering the following identities
\begin{align}\label{SumId}
 \sum_{k=1}^b (-1)^{\binom{k+1}{2}} k^2 
\sum_{j=1}^k &\frac{(-1)^j}{j}=
\frac{1}{2}
\sum_{j=1}^b \frac{(-1)^{\binom{j+1}{2}}}{j}-\frac{1}{4} (-1)^{\binom{b+1}{2}} \big(-1+(-1)^b+2 b\big)\nonumber\\[-0.3cm]
&
+(-1)^{\binom{b+1}{2}} \frac{1}{2} \big(b (b+2)+(-1)^b \big(b^2-1\big)
\big)\sum_{j=1}^b \frac{(-1)^j}{j},\\
\label{Equ:ProductId}
\prod_{k=1}^b\frac{-(\iota^k)}{1+k}=&
\big(-\frac{\iota}{2}-\frac{1}{2}
\big)\frac{-(-1)^b+\iota}{b (b+1)}\Big(\prod_{j=1}^{b-1}\frac{\iota^j}{j}\Big);
\end{align}
here the imaginary unit is denoted by $\iota$, i.e., $\iota^2=-1$. 

The outline is as follows. In Section~\ref{Sec:MainResults} we will introduce the basic
notations of difference rings (resp.\ fields) and define \rpisiSE-ring extensions. Furthermore, we
will work out the underlying problems in the setting of difference rings and
motivate the different challenges that will
be treated in this article. In addition, we give an overview of the main results and show how they can be applied for symbolic summation. In the remaining sections these results will be worked out in details. In Section~\ref{Sec:SingleExt} we present the crucial properties of single nested \rpisiSE-extensions. Special emphasis will be put on the properties of the underlying ring. In Section~\ref{Sec:NestedExt} we will consider a tower of such extensions
and explore the set of semi-constants. In Section~\ref{Sec:Period} we present algorithms that calculate the order, period and factorial order of the generators of \rE-extensions. Finally, in Section~\ref{Sec:M} and Section~\ref{Sec:PFLDE} we elaborate algorithms that are needed to construct \rpisiSE-extensions and that solve as a special case the (parameterized) telescoping problem. A conclusion is given in Section~\ref{Sec:Conclusion}.

\section{Basic definitions, the outline of the problems, and the main results}\label{Sec:MainResults}

In this article all rings are commutative with 1 and all rings (resp.\
fields) have characteristic $0$; in particular, they contain
the rational numbers $\QQ$ as a subring (resp.\ subfield). A ring (resp.\ field) is called computable if there are algorithms available that can perform the standard operations (including zero recognition and deciding constructively if an element is invertible). The multiplicative group
of units (invertible elements) of a ring $\AA$ is denoted by $\AA^*$. The ideal generated by $S\subseteq\AA$ is denoted by $\langle S\rangle$\index{$\langle S\rangle$}. If $\AA$ is a subring (resp. subfield/multiplicative subgroup)
of $\tilde{\AA}$ we also write $\AA\leq\tilde{\AA}$. The non-negative integers are denoted by $\NN=\{0,1,2,\dots\}$. 

In this section we will present a general framework in which our symbolic summation problems can be formulated and tackled in the setting of difference rings.
Here an indefinite nested product-sum expression $f(k)$ (like  in~\eqref{Equ:TeleProblem} or~\eqref{Equ:ProblemPT})
is described in a ring (resp.\ field) $\AA$ and the shift behaviour of such an expression is reflected by a ring automorphism (resp.\ field automorphism) $\fct{\sigma}{\AA}{\AA}$, i.e., $\sigma^i(f)$ with $i\in\ZZ$ represents the expression $f(k+i)$. 
In the following we call such a ring $\AA$ (resp.\ field) equipped with a ring automorphism (resp. field
automorphism) $\sigma$ a difference ring\index{difference!ring/field} (resp. difference field)~\cite{Cohn:65,Levin:08} and denote it by
$\dfield{\AA}{\sigma}$.
We remark that any difference field is also a difference ring. Conversely, any difference ring $\dfield{\AA}{\sigma}$ with $\AA$ being a field is automatically
a difference field. A difference ring (resp.\ field) $\dfield{\AA}{\sigma}$ is called computable if both, $\AA$ and the function $\sigma$ are computable; note that in such rings one can decide if an element is a constant, i.e., if $\sigma(c)=c$.
The set of constants is also denoted by\index{constant field/ring}
$\constF{\AA}{\sigma}=\{c\in\AA|\,\sigma(c)=c\}$, and
if it is clear from the context, we also write
$\const{\AA}{\sigma}=\constF{\AA}{\sigma}$. It is easy to check that
$\const{\AA}{\sigma}$ is a subring (resp. a subfield) of $\AA$ which
contains as subring (resp.\ subfield) the rational numbers $\QQ$. 
Throughout this article we will take care that $\const{\AA}{\sigma}$ is always a field (and not just a ring), called the constant field and denoted by $\KK$.

In the first subsection we introduce the class of difference rings in which we will model indefinite nested sums and products. They will be introduced by a tower of ring extensions, the so-called \rpisiSE-ring extensions. 

In Subsection~\ref{Subsec:AlgProblems}  we will focus on two tasks:\\[0.05cm]
\noindent(1) Introduce techniques that enable one to test if the given tower of extensions is an \rpisiSE-extension; even more, derive tactics that enable one to represent sums and products automatically in \rpisiSE-extensions.\\
(2) Work out the underlying subproblems in order to solve two central problems of symbolic summation: telescoping  (compare~\eqref{Equ:TeleProblem}) and parameterized telescoping (compare~\eqref{Equ:ProblemPT}). In their simplest form they can be specified as follows.

\medskip

\noindent\fbox{\begin{minipage}{12.8cm}
\noindent\index{Problem!T}\textbf{Problem~T for $\dfield{\AA}{\sigma}$.} \textit{Given} a difference ring $\dfield{\AA}{\sigma}$ and given $f\in\AA$.
\textit{Find}, if possible, a $g\in\AA$ such that the telescoping (T) equation holds: 
\begin{equation}\label{Equ:TeleDR}
\sigma(g)-g=f.
\end{equation}
\noindent\index{Problem!PT}\textbf{Problem~PT for $\dfield{\AA}{\sigma}$.} \textit{Given} a difference ring $\dfield{\AA}{\sigma}$ with constant field $\KK$ and given $f_1,\dots,f_n\in\AA$.
\textit{Find}, if possible, $c_1,\dots,c_n\in\KK$ (not all $c_i$ being zero)
and a $g\in\AA$ such that the parameterized telescoping (PT) holds: 
\begin{equation}\label{Equ:ParaTeleDR}
\sigma(g)-g=c_1\,f_1+\dots+c_n\,f_n.
\end{equation}
\end{minipage}}

\medskip

In Subsection~\ref{Subsec:MainResults} we will present the main results of theoretical and algorithmic nature to handle these problems, and in Subsection~\ref{Subsec:dAlembert} we demonstrate how the new summation theory can be used to represent d'Alembertian solutions in \rpisiSE-extensions.

\subsection{The definition of \rpisiSE-extensions}\label{Subsec:RPiSi}

A difference ring $\dfield{\tilde{\AA}}{\tilde{\sigma}}$ is a difference ring extension\index{difference!ring/field extension} of a difference ring $\dfield{\AA}{\sigma}$ if $\AA\leq\tilde{\AA}$ and
$\tilde{\sigma}|_{\AA}=\sigma$, i.e., $\AA$ is a subring of $\tilde{\AA}$ and $\tilde{\sigma}(a)=\sigma(a)$ for all $a\in\AA$. The definition of difference field extensions is the same by replacing the word ring with field.
In short (for the ring and field version) we also write
$\dfield{\AA}{\sigma}\leq\dfield{\tilde{A}}{\tilde{\sigma}}$\index{$\dfield{\AA}{\sigma}\leq\dfield{\tilde{A}}{\tilde{\sigma}}$}. If it is clear
from the context, we do not distinguish anymore between $\sigma$ and
$\tilde{\sigma}$.

\noindent For the construction of \rpisiSE-extensions, we start with the following basic properties.

\begin{lemma}
Let $\AA$ be a ring with $\alpha\in\AA^*$ and $\beta\in\AA$ equipped with a ring
automorphism $\fct{\sigma}{\AA}{\AA}$. Let
$\AA[t]$ be a polynomial ring and $\AA\ltr{t}$ be a ring of Laurent
polynomials.
\begin{enumerate}
\item There is a unique
automorphism $\fct{\sigma'}{\AA[t]}{\AA[t]}$ with $\sigma'|_{\AA}=\sigma$ and
$\sigma'(t)=\alpha\,t+\beta$.
\item There is a unique
automorphism $\fct{\sigma''}{\AA\ltr{t}}{\AA\ltr{t}}$ with
$\sigma''|_{\AA}=\sigma$ and
$\sigma''(t)=\alpha\,t$ (where $\sigma''(\frac{1}{t})=\alpha^{-1}\,\frac{1}{t}$). In particular, if $\beta=0$,
$\sigma''|_{\AA[t]}=\sigma'$.
\item If $\AA$ is field and $\AA(t)$ is a rational function field, there is a
unique field automorphism
$\fct{\sigma'''}{\AA(t)}{\AA(t)}$ with $\sigma'''|_{\AA}=\sigma$ and
$\sigma'''(t)=\alpha\,t+\beta$. In particular, $\sigma'''|_{\AA[t]}=\sigma'$;
moreover, $\sigma'''|_{\AA[t,1/t]}=\sigma''$ if $\beta=0$.
\end{enumerate}
\end{lemma}

\noindent In summary, let $\dfield{\AA}{\sigma}$ be a difference ring and $t$ be transcendental over $\AA$. Then we obtain the uniquely determined difference ring extension 
$\dfield{\AA[t]}{\sigma}$ of $\dfield{\AA}{\sigma}$
with $\sigma(t)=\alpha\,t+\beta$ where $\alpha\in\AA^*$ and $\beta\in\AA$. In
particular, we get the uniquely determined difference ring
extension $\dfield{\AA\ltr{t}}{\sigma}$ of $\dfield{\AA}{\sigma}$ with $\sigma(t)=\alpha\,t$. Thus for
$\beta=0$, we have the chain of extensions
$\dfield{\AA}{\sigma}\leq\dfield{\AA[t]}{\sigma}\leq\dfield{\AA\ltr{t}}{\sigma}
.$
Moreover, if $\AA$ is a
field,
we obtain the uniquely determined difference field extension 
$\dfield{\AA(t)}{\sigma}$ of $\dfield{\AA}{\sigma}$ with
$\sigma(t)=\alpha\,t+\beta$. 
Following the notions of~\cite{Bron:00} each of the extensions, i.e.,
$\dfield{\AA}{\sigma}\leq\dfield{\AA[t]}{\sigma}$,
$\dfield{\AA}{\sigma}\leq\dfield{\AA\ltr{t}}{\sigma}$ or
$\dfield{\AA}{\sigma}\leq\dfield{\AA(t)}{\sigma}$ are called unimonomial\index{extension!unimonomial}
extensions (of polynomial, Laurent polynomial or of rational function type, respectively).

\begin{example}\label{Exp:FactorialUni}
(0) Take the difference field $\dfield{\QQ}{\sigma}$ with $\sigma(c)=c$ for all $c\in\QQ$.\\ 
(1) Take the unimonomial field extension $\dfield{\QQ(k)}{\sigma}$ of $\dfield{\QQ}{\sigma}$ with $\sigma(k)=k+1$: $\QQ(k)$ is a rational function field and $\sigma$ is extended from $\QQ$ to $\QQ(k)$ with $\sigma(k)=k+1$.\\
(2) Take the unimonomial ring extension $\dfield{\QQ(k)\ltr{t}}{\sigma}$ of $\dfield{\QQ(k)}{\sigma}$ with $\sigma(t)=(k+1)\,t$: $\QQ(k)\ltr{t}$ is a ring of Laurent polynomials with coefficients from $\QQ(k)$ and the automorphism is extended from $\QQ(k)$ to $\QQ(k)\ltr{t}$ with $\sigma(t)=(k+1)\,t$. 
\end{example}

\noindent Finally, we consider those extensions where the constants remain unchanged.

\begin{definition}\label{Def:PiSigmaSingle}
Let $\dfield{\AA}{\sigma}$ be a difference ring.\index{extension!\piE}\index{extension!\sigmaSE}\index{extension!\piE}
\begin{itemize}
\item A unimonomial ring extension $\dfield{\AA[t]}{\sigma}$ of
$\dfield{\AA}{\sigma}$ with $\sigma(t)-t\in\AA$
and $\const{\AA[t]}{\sigma}=\const{\AA}{\sigma}$ is called \sigmaSE-ring extension (in short \sigmaSE-extension). 
\item If $\AA$ is a field, a unimonomial field extension $\dfield{\AA(t)}{\sigma}$ of
$\dfield{\AA}{\sigma}$ with $\sigma(t)-t\in\AA$
and $\const{\AA(t)}{\sigma}=\const{\AA}{\sigma}$ is called \sigmaSE-field\footnote{We restrict Karr's $\Sigma$-field extensions to \sigmaSE-field extensions being slightly less general but covering all sums treated explicitly in Karr's work~\cite{Karr:81}.} extension. 
\item A unimonomial ring extension $\dfield{\AA\ltr{t}}{\sigma}$ of
$\dfield{\AA}{\sigma}$ with $\frac{\sigma(t)}{t}\in\AA^*$
and $\const{\AA\ltr{t}}{\sigma}=\const{\AA}{\sigma}$ is called \piE-ring extension (in short \piE-extension). 
\item If $\AA$ is a field, a unimonomial field extension $\dfield{\AA(t)}{\sigma}$ of
$\dfield{\AA}{\sigma}$ with $\frac{\sigma(t)}{t}\in\AA^*=\AA(t)\setminus\{0\}$
and $\const{\AA(t)}{\sigma}=\const{\AA}{\sigma}$ is called \piE-field extension. 
\end{itemize}
The generators of a \sigmaSE-extension (in the ring or field version) and a
\piE-extension (in the ring or field version) are called \sigmaSE-monomial and
\piE-monomial, respectively.
\end{definition}

\begin{remark}\label{Remark:StableConstants}
Keeping the constants unchanged is a central property to tackle the (parameterized) telescoping problem. E.g., if the constants are extended, there do not exist bounds on the degrees as utilized in Subsection~\ref{Subsection:DegreeBounds}. Additionally, introducing no extra constants is \textit{the} essential property to embed the derived difference rings into the ring of sequences; this fact has been worked out, e.g., in~\cite{Schneider:10c} which is related to~\cite{Singer:08}.
\end{remark}

\begin{example}[Cont.\ Ex.~\ref{Exp:FactorialUni}]\label{Exp:FactorialPi}
For $\dfield{\QQ}{\sigma}\leq\dfield{\QQ(k)}{\sigma}\leq\dfield{\QQ(k)\ltr{t}}{\sigma}$ from Example~\ref{Exp:FactorialUni} we have that 
$\const{\QQ(k)\ltr{t}}{\sigma}=\const{\QQ(k)}{\sigma}=\const{\QQ}{\sigma}=\QQ$, which can be checked easily. 
Thus $\dfield{\QQ(k)}{\sigma}$ is a \sigmaSE-field extension of $\dfield{\QQ}{\sigma}$ and $\dfield{\QQ(k)\ltr{t}}{\sigma}$ is a \piE-extension of $\dfield{\QQ(k)}{\sigma}$. The generator $k$ is a \sigmaSE-monomial and the generator $t$ is a \piE-monomial.
\end{example}

\noindent For more complicated extensions it is rather demanding to check if the constants remain unchanged. In this regard, we refer to the field-algorithms given in~\cite{Karr:81} or to our enhanced ring-algorithms given below which can perform these checks automatically.

\medskip

For further considerations we introduce the order function $\fct{\ord}{\AA}{\NN}$ with\index{function!order}\index{$\ord(f)$}
\begin{equation}\label{equ:OrderDef}
\ord(h)=\begin{cases}
0&\text{ if }\nexists n>0\text{ s.t.\ } h^n=1\\
\min\{n>0|\,h^n=1\}&\text{ otherwise}.
\end{cases}
\end{equation}
The third type of extensions is concerned with algebraic objects like~\eqref{Equ:Equ:AlgebraicObjects}. 
Let $\lambda\in\NN$ with $\lambda>1$, take a root of
unity $\alpha\in\AA^*$ with $\alpha^{\lambda}=1$ and construct the unimonomial extension
$\dfield{\AA[y]}{\sigma}$ of $\dfield{\AA}{\sigma}$ with $\sigma(y)=\alpha\,y$.
Now
take the ideal $I:=\lr{y^{\lambda}-1}$ and consider the quotient ring $\EE=\AA[y]/I$.
Since $I$ is closed under $\sigma$, i.e., $I$ is a reflexive difference ideal~\cite[page~71]{Cohn:65}, one can
verify that $\fct{\sigma}{\EE}{\EE}$ with
$\sigma(f+I)=\sigma(f)+I$
forms a ring automorphism. In other words, $\dfield{\EE}{\sigma}$ is a
difference ring. Moreover, there is the natural embedding of $\AA$ into $\EE$
with
$a\to a+I$.
By identifying $a$ with $a+I$, $\dfield{\EE}{\sigma}$ is a difference ring extension of $\dfield{\AA}{\sigma}$.

\begin{lemma}\label{Lemma:RExistence}
Let $\dfield{\AA}{\sigma}$ be a difference ring and $\alpha\in\AA^*$ with
$\alpha^{\lambda}=1$ for some ${\lambda}>1$. Then there is (up to a difference ring isomorphism)
a unique difference ring extension $\dfield{\AA[x]}{\sigma}$ of
$\dfield{\AA}{\sigma}$ with $x\notin\AA$ subject to the relations $x^{\lambda}=1$ and
$\sigma(x)=\alpha\,x$.
\end{lemma}
\begin{proof}
Consider the difference ring extension $\dfield{\EE}{\sigma}$ of
$\dfield{\AA}{\sigma}$ constructed above. Define $x:=y+I$. Then
$\sigma(x)=\alpha\,x$ and $x^{\lambda}=y^{\lambda}+I=1+I=1$. Further, 
$\EE=\{\sum_{i=0}^{{\lambda}-1}a_ix^i|a_i\in\AA\}$. Thus we obtain a difference ring
extension as claimed in the lemma. Now suppose that there is another difference
ring extension $\dfield{\AA[x']}{\sigma'}$ of $\dfield{\AA}{\sigma}$ with $x'\notin\AA$ subject to
the relations $\sigma'(x')=\alpha\,x'$ and $x'^{\lambda}=1$. Then by the first
isomorphism theorem, there is the ring isomorphism $\fct{\tau}{\EE}{\AA[x']}$
with $\tau(\sum_{i=0}^{{\lambda}-1}f_i\,x_i)=\sum_{i=0}^{{\lambda}-1}f_i\,x'_i$. Since
$\tau(\sigma(x))=\tau(\alpha\,x)=\tau(\alpha)\,\tau(x)=\alpha\,x'=\sigma'(x')$, it follows that $\tau(\sigma(f))=\sigma'(\tau(f))$ for all $f\in\AA[x]$. Summarizing, $\tau$ is a difference ring isomorphism.
\end{proof}

\noindent The extension $\dfield{\AA[x]}{\sigma}$
of $\dfield{\AA}{\sigma}$ in Lemma~\ref{Lemma:RExistence} is called
algebraic extension of order ${\lambda}$.\index{extension!algebraic}

\begin{example}\label{Exp:Qkxy}
(0) Take the \sigmaSE-ext.\ $\dfield{\QQ(k)}{\sigma}$ of $\dfield{\QQ}{\sigma}$ with $\sigma(k)=k+1$ from Ex.~\ref{Exp:FactorialPi}.\\
(1) Take the algebraic extension $\dfield{\QQ(k)[x]}{\sigma}$ of $\dfield{\QQ(k)}{\sigma}$ with $\sigma(x)=-x$ of order $2$: $\QQ(k)[x]$ is an algebraic ring extension of $\QQ(k)$ subject to the relation $x^2=1$ and $\sigma$ is extended from $\QQ(k)$ to $\QQ(k)[x]$ with $\sigma(x)=-x$. Note that $x$ represents the expression $X(k)=(-1)^k$ with $X(k+1)=-X(k)$.\\
(2) Take the algebraic extension $\dfield{\QQ(k)[x][y]}{\sigma}$ of $\dfield{\QQ(k)[x]}{\sigma}$ with $\sigma(y)=-x\,y$ of order~$2$: $\QQ(k)[x][y]$ is a ring extension of $\QQ(k)[x]$ with $y^2=1$ and $\sigma$ is extended from $\QQ(k)[x]$ to $\QQ(k)[x][y]$ with $\sigma(y)=-x\,y$. Note that $y$ represents the expression $Y(k)=(-1)^{\binom{k+1}{2}}=\prod_{j=1}^k(-1)^j$ with $Y(k+1)=-(-1)^k\,Y(k)$.
\end{example}

\noindent As for unimonomial extensions, we restrict now to those algebraic extensions where the constants remain unchanged. For the underlying motivation we refer to Remark~\ref{Remark:StableConstants}.

\begin{definition}\label{Def:RSingle}
Let ${\lambda}\in\NN\setminus\{0,1\}$. An algebraic extension $\dfield{\AA[x]}{\sigma}$
of $\dfield{\AA}{\sigma}$ order ${\lambda}$ with $\const{\AA[x]}{\sigma}=\const{\AA}{\sigma}$ is called root of unity extension (in short \rE-extension)\index{extension!\rE} of order ${\lambda}$.
The generator $x$ is called \rE-monomial.
\end{definition}

\begin{example}[Cont.\ Ex.~\ref{Exp:Qkxy}]\label{Exp:QkxyIsRExt}
For $\dfield{\QQ}{\sigma}\leq\dfield{\QQ(k)}{\sigma}\leq\dfield{\QQ(k)[x]}{\sigma}\leq\dfield{\QQ(k)[x][y]}{\sigma}$ from Example~\ref{Exp:Qkxy} we have that 
$\const{\QQ(k)[x][y]}{\sigma}=\const{\QQ(k)[x]}{\sigma}=\const{\QQ(k)}{\sigma}=\QQ$, which can be checked algorithmically; see Example~\ref{Exp:QkxyCheckRExt} below.
Thus $\dfield{\QQ(k)[x]}{\sigma}$ is an \rE-extension of $\dfield{\QQ(k)}{\sigma}$ and $\dfield{\QQ(k)[x][y]}{\sigma}$ is an \rE-extension of $\dfield{\QQ(k)[x]}{\sigma}$.
\end{example}

To this end, we define a tower of such extensions. First, we introduce the following notion.
Let $\dfield{\AA}{\sigma}\leq\dfield{\EE}{\sigma}$ with $t\in\EE$. In the following $\AA\lr{t}$\index{$\AA\lr{t}$} denotes
the polynomial ring $\AA[t]$ if $\dfield{\AA[t]}{\sigma}$ is a
\sigmaSE-extension of $\dfield{\AA}{\sigma}$. 
$\AA\lr{t}$ denotes the ring of Laurent polynomials $\AA\ltr{t}$ if
$\dfield{\AA\ltr{t}}{\sigma}$ is a \piE-extension of $\dfield{\AA}{\sigma}$.
Finally, $\AA\lr{t}$ denotes the ring $\AA[t]$ with $t\notin\AA$ subject to the relation $t^{\lambda}=1$ if
$\dfield{\AA[t]}{\sigma}$ is an \rE-extension of $\dfield{\AA}{\sigma}$ of
order ${\lambda}$. 

\begin{definition}
\index{extension!(nested) \piE,\sigmaSE,\rE,\rE\piE,\\ \rE\sigmaSE,\pisiSE,\rpisiSE}
\index{\pisiSE-field}
\index{monomial!\piE,\sigmaSE,\rE,\rE\piE,\rE\sigmaSE,\pisiSE,\rpisiSE}
A difference ring extension $\dfield{\AA\lr{t}}{\sigma}$ of
$\dfield{\AA}{\sigma}$ is called \rpisiSE-extension if it is an \rE-extension,
\piE-extension or \sigmaSE-extension. Analogously, it is called
\rE\sigmaSE-extension, \rE\piE-extension or \pisiSE-extension if it is one of the corresponding extensions. More generally,
$\dfield{\GG\lr{t_1}\lr{t_2}\dots\lr{t_e}}{\sigma}$ is a (nested) \rpisiSE-extension
(resp.\ \rE\piE, \rE\sigmaSE, \hbox{\pisiSE-,} \rE-, \hbox{\piE-}, \sigmaSE-extension) of
$\dfield{\GG}{\sigma}$ if it is a tower of such extensions.\\
Similarly, if $\AA$ is a field, $\dfield{\AA(t)}{\sigma}$ is called a \pisiSE-field extension if it is either a \piE-field extension or a \sigmaSE-field extension. $\dfield{\GG(t_1)\dots(t_e)}{\sigma}$ is called a \pisiSE-field extension (resp.\ \piE-field extension, \sigmaSE-field extension) of $\dfield{\GG}{\sigma}$ if it is a tower of such extensions. In particular, if $\const{\GG}{\sigma}=\GG$, $\dfield{\GG(t_1)(t_2)\dots(t_e)}{\sigma}$ is called a \pisiSE-field over $\GG$.\\
In both, the ring and field version, $t_i$ is called \rpisiSE-monomial (resp.\ \rE\piE-, \rE\sigmaSE-, \pisiSE-monomial) if it is a generator of a \rpisiSE-extension (resp.\ \rE\piE-, \rE\sigmaSE-, \pisiSE-extension).
\end{definition}

\begin{example}[Cont.\ Ex.~\ref{Exp:QkxyIsRExt}]
(1) $\dfield{\QQ(k)}{\sigma}$ is a \pisiSE-field over $\QQ$.\\
(2) $\dfield{\QQ(k)\lr{x}\lr{y}}{\sigma}$ 
is an \rE-extension of $\dfield{\QQ(k)}{\sigma}$.
\end{example}

\noindent The generators with their sequential arrangement, incorporating the recursive definition of the automorphism, are always given explicitly. In particular, any reordering of the generators must respect the recursive nature induced by the automorphism.

\subsection{A characterization of \rpisiSE-extensions and their algorithmic construction}\label{Subsec:AlgProblems}

For the construction of \rpisiSE-extensions we rely on the following result; for the proofs of part~1, part~2 and part~3 we refer to Proof~\ref{Thm:SigmaChar},  Proof~\ref{Thm:PiChar} and Proof~\ref{Thm:RChar}, respectively.

\begin{theorem}\label{Thm:RPSCharacterization}
Let $\dfield{\AA}{\sigma}$ be a difference ring. Then
the following holds.
\begin{enumerate}
\item Let $\dfield{\AA[t]}{\sigma}$ be a unimonomial ring extension 
of $\dfield{\AA}{\sigma}$ with $\sigma(t)=t+\beta$ where $\beta\in\AA$ such that $\const{\AA}{\sigma}$ is a field. Then this is a \sigmaSE-extension (i.e.,
$\const{\AA[t]}{\sigma}=\const{\AA}{\sigma}$) iff there does not exist a
$g\in\AA$ with $\sigma(g)=g+\beta$.

\item Let $\dfield{\AA\ltr{t}}{\sigma}$ be a unimonomial ring
extension of
$\dfield{\AA}{\sigma}$ with $\sigma(t)=\alpha\,t$ where $\alpha\in\AA^*$. Then
this is a
\piE-extension (i.e., $\const{\AA\ltr{t}}{\sigma}=\const{\AA}{\sigma}$) iff
there are no $g\in\AA\setminus\{0\}$ and $m\in\ZZ\setminus\{0\}$
with
$\sigma(g)=\alpha^m\,g$. If it is a \piE-extension, $\ord(\alpha)=0$.

\item Let $\dfield{\AA[t]}{\sigma}$ be an algebraic ring
extension of $\dfield{\AA}{\sigma}$ of order $\lambda>1$ with $\sigma(t)=\alpha\,t$
where $\alpha\in\AA^*$. Then this is an \rE-extension (i.e.,
$\const{\AA[t]}{\sigma}=\const{\AA}{\sigma}$) iff there are no
$g\in\AA\setminus\{0\}$ and $m\in\{1,\dots,\lambda-1\}$ with
$\sigma(g)=\alpha^m\,g$. If it is an \rE-extension, then $\alpha$ is primitive, i.e., $\ord(\alpha)=\lambda$.
\end{enumerate}
\end{theorem}

\noindent For Karr's celebrated field version~\cite{Karr:81,Karr:85} of this result we refer to  Theorems~\ref{Thm:SigmaFieldTheorem} and~\ref{Thm:PiCharField} below, that can be nicely embedded in the general difference ring framework. We emphasize that Theorem~\ref{Thm:RPSCharacterization}  facilitates algorithmic tactics to build difference ring extensions and to verify simultaneously if they form \rpisiSE-extensions. Here we consider two cases.

\subsubsection{Testing and constructing \rE\piE-extensions}

Let $\dfield{\AA}{\sigma}$ be a difference ring and let $\alpha\in\AA$. Then we want to decide if we can construct an \rE\piE-extension $\dfield{\AA\lr{t}}{\sigma}$ of $\dfield{\AA}{\sigma}$ with $\sigma(t)=\alpha\,t$. First, we have to check if 
$\alpha\in\AA^*$. E.g., for the class of difference rings $\dfield{\AA}{\sigma}$, built by simple \rpisiSE-extensions introduced in  Definition~\ref{Def:SimpleRPS} below, this task will be straightforward. Next, we need the order of $\alpha$, i.e., we have to solve the following Problem~O with $G:=\AA^*$.
\medskip

\noindent\fbox{\begin{minipage}{12.8cm}
\noindent\index{Problem!O}\textbf{Problem~O in $G$.} \textit{Given} a group
$G$ and $\alpha\in G$. \textit{Find} $\ord(\alpha)$.
\end{minipage}}

\medskip

\noindent Given $\lambda=\ord(\alpha)$, we can decide which case has to be treated. 
If $\lambda=0$, only the construction of a \piE-extension might be possible due to  Theorem~\ref{Thm:RPSCharacterization}. Thus we construct the unimonomial extension $\dfield{\AA\ltr{t}}{\sigma}$ of $\dfield{\AA}{\sigma}$ with $\sigma(t)=\alpha\,t$. Otherwise, if $\lambda>0$, we construct the algebraic extension $\dfield{\AA[t]}{\sigma}$ of $\dfield{\AA}{\sigma}$ with $\sigma(t)=\alpha\,t$ of order $\lambda$. Finally, we check if our construction is indeed a \piE-extension or \rE-extension, i.e., if the constants remain unchanged. Using  Theorem~\ref{Thm:RPSCharacterization} this test can be accomplished by solving
\medskip

\noindent\fbox{\begin{minipage}{12.8cm}
\noindent\index{Problem!MT}\textbf{Problem~MT in $\dfield{\AA}{\sigma}$.} \textit{Given} a difference ring
$\dfield{\AA}{\sigma}$ and $\alpha\in\AA^*$ with $\lambda=\ord(\alpha)$. \textit{Decide} if
there are a $g\in\AA\setminus\{0\}$ and an $m\in\ZZ\setminus\{0\}$ for the case $\lambda=0$ (resp.\ $m\in\{1,\dots,\lambda-1\}$ for the case $\lambda>0$) such that the multiplicative version of the telescoping equation (MT) holds: 
\begin{equation}\label{Equ:MultTele}
\sigma(g)=\alpha^m\,g.
\end{equation}
\end{minipage}}

\medskip

\noindent More generally, if we are given a tower of algebraic and unimonomial extensions, which model indefinite nested products, Problem MT can be used to check if the construction constitutes a nested \rE\piE-extension.

\begin{example}[Cont.\ Ex.~\ref{Exp:QkxyIsRExt}]\label{Exp:QkxyCheckRExt}
We will verify that $\dfield{\QQ(k)[x][y]}{\sigma}$ is an \rE-extension of $\dfield{\QQ(k)}{\sigma}$. (1) Take $\alpha=-1$ with $\lambda=\ord(\alpha)=2$. We solve Problem~MP by the algorithms presented below: there are no $g\in\QQ(k)^*$ and $m\in\{1\}$ with $\sigma(g)=(-1)^m\,g$. Hence by Theorem\footnote{Note: Theorem~\ref{Thm:RPSCharacterization}.(3) is a shortcut for ``part~3 of Theorem~\ref{Thm:RPSCharacterization}''. The same convention will be applied for other references.}~\ref{Thm:RPSCharacterization}.(3) $\dfield{\QQ(k)[x]}{\sigma}$ is an \rE-extension of $\dfield{\QQ(k)}{\sigma}$. \\
(2) Now we solve Problem~O for $\alpha=-x$ and get $\lambda=\ord(-x)=2$; see Example~\ref{Exp:GetOrder}.(2). In addition, solving Problem~MP for $\alpha$ shows that there is no $g\in\QQ(k)[x]\setminus\{0\}$ with $\sigma(g)=-x\,g$. Thus by Theorem~\ref{Thm:RPSCharacterization}.(3) $\dfield{\QQ(k)[x][y]}{\sigma}$ forms an \rE-extension of $\dfield{\QQ(k)[x]}{\sigma}$. 
\end{example}

\begin{example}\label{Exp:QkxtCheckRPiExt}
We construct a ring in which the objects in~\eqref{Equ:ProductId} can be represented. \\
(0) Take the \pisiSE-field $\dfield{\KK(k)}{\sigma}$ over $\KK=\QQ(\iota)$ with $\sigma(k)=k+1$.\\
(1) Take $\alpha=\iota$. Then solving Problem~O provides $\lambda=\ord(\alpha)=4$. In particular solving the corresponding Problem~MP proves that there are no $g\in\KK(k)^*$ and $m\in\{1,2,3\}$ with~\eqref{Equ:MultTele}. Hence by Theorem~\ref{Thm:RPSCharacterization}.(3) we can construct the \rE-extension $\dfield{\KK(k)[x]}{\sigma}$ of $\dfield{\KK(k)}{\sigma}$ with $\sigma(x)=\iota\,x$. Note that the \rE-monomial $x$ represents $\iota^k$.\\
(2) Take $\alpha=x\,k$. Solving Problem~O yields $\lambda=\ord(\alpha)=0$ and solving Problem~MP shows that there are no $g\in\KK(k)[x]\setminus\{0\}$ and $m\in\ZZ\setminus\{0\}$ with~\eqref{Equ:MultTele}.
With Theorem~\ref{Thm:RPSCharacterization}.(2) we can construct the \piE-extension $\dfield{\KK(k)[x]}{\sigma}\leq\dfield{\KK(k)[x]\lr{t}}{\sigma}$ with $\sigma(t)=x\,k\,t$; here the \piE-monomial $t$ represents $\prod_{j=1}^{k-1}j\iota^j$. 
\end{example}

\subsubsection{Testing and constructing \sigmaSE-extensions}

In order to verify if a unimonomial extension as given in Theorem~\ref{Thm:RPSCharacterization}.(1) is a \sigmaSE-extension, it suffices to solve Problem~T with $f=\beta$ and to check if there is not a telescoping solution. We illustrate this feature by actually constructing a difference ring in which the summand
\begin{equation}\label{Equ:MainSummand}
f(k)=(-1)^{\binom{k+1}{2}} k^2 
\sum_{j=1}^k \frac{(-1)^j}{j}
\end{equation}
given on the left hand side of~\eqref{SumId} and the additional sum 
\begin{equation}\label{Equ:ExtraSum}
\sum_{j=1}^k \frac{(-1)^{\binom{j+1}{2}}}{j}
\end{equation}
occurring on the right hand side of~\eqref{SumId} 
can be represented. In particular, we demonstrate how identity~\eqref{SumId} can be discovered in this difference ring. 

\begin{example}[Cont.\ Ex.~\ref{Exp:QkxyIsRExt}]\label{ExpQkxysS}
(0) Take the difference ring $\dfield{\AA}{\sigma}$ with $\AA=\QQ(k)[x][y]$.\\
(1) Take $f=\sigma(\frac{x}{k})=\frac{-x}{k+1}$. Then solving Problem~T shows that there is no $g\in\AA$ with $\sigma(g)-g=\frac{-x}{k+1}$. Hence we can construct the \sigmaSE-extension $\dfield{\AA[s]}{\sigma}$ of $\dfield{\AA}{\sigma}$ with $\sigma(s)=s+\frac{-x}{k+1}$; note that the \sigmaSE-monomial $s$ represents $\sum_{j=1}^k \frac{(-1)^j}{j}$.\\
(2) Take $f=\sigma(\frac{y}{k})=\frac{-x\,y}{k+1}$. Then solving Problem~T shows that there is no $g\in\AA[s]$ with $\sigma(g)-g=\frac{-x\,y}{k+1}$. Hence we can construct the \sigmaSE-extension $\dfield{\AA[s][S]}{\sigma}$ of $\dfield{\AA[s]}{\sigma}$ with $\sigma(S)=S+\frac{-x\,y}{k+1}$; note that the \sigmaSE-monomial $S$ represents the sum~\eqref{Equ:ExtraSum}.\\
(3) Take $f=y\,k^2\,s$ which represents~\eqref{Equ:MainSummand}. Solving Problem~T produces the solution
\begin{equation}\label{Equ:DFTeleSol}
g=s y \big(\tfrac{1}{2} (k-1)
   (k+1) x-\tfrac{1}{2} (k-2)
   k\big)+y
   \big(\tfrac{1}{4} (1-2
   k)-\tfrac{1}{4}x\big)+\tfrac
   {1}{2}S;
\end{equation}
for further details see Example~\ref{Exp:TeleSigmaBound}.
Hence this yields the solution of the telescoping equation~\eqref{Equ:TeleProblem} for our summand~\eqref{Equ:MainSummand} by replacing the \rE\sigmaSE-monomials $x,y,s,S$ with the corresponding summation objects. Taking $a=1$ in~\eqref{Equ:SumTele} and performing the evaluation $c:=g(1)=0\in\QQ$ gives the identity~\eqref{SumId}.\\ 
(4) Note that we succeeded in representing
the sum $F(k)=\sum_{i=1}^kf(i)$ with $f$ from~\eqref{Equ:MainSummand} in the difference ring  in $\AA[s][S]$ with
$\sigma(g)-c=\sigma(g)$. Namely, replacing the variables in $\sigma(g)$ with the corresponding summation objects yields the right hand side of~\eqref{SumId}. This is of particular interest if there are further sums defined over $F(k)$ which one wants to represent in a \sigmaSE-extension over $\dfield{\AA[s][S]}{\sigma}$.
\end{example}

\noindent We remark that for the derivation of the identity~\eqref{SumId} it is crucial to introduce the extra sum~\eqref{Equ:ExtraSum}. Here this was accomplished manually. But, using algorithms from~\cite{Schneider:08c,Schneider:14} in combination with the results of this article, this sum can be determined automatically.

\subsubsection{The underlying problems for \rpisiSE-extensions}

As in the difference field approach~\cite{Karr:81,Schneider:05a,Schneider:08c,Schneider:14},
Problem~T and more generally Problem~PT will be solved by reducing them from $\dfield{\AA}{\sigma}$ to smaller difference rings (i.e., rings built by less \rpisiSE-monomials). Likewise, this reduction technique can be applied in order to solve a special case of Problem~MT that will cover all the cases needed for our difference ring constructions. However, in order to carry out these reductions, one has to tackle generalized problems within the recursion steps.

\medskip

For Problem~MT the following generalization is needed.
Let $\dfield{\AA}{\sigma}$ be a difference ring, let $W\subseteq\AA$ and let
$\vect{f}=(f_1,\dots,f_n)\in(\AA^*)^n$. Then we define the set~\cite{Karr:81}\index{$M(\vect{f},\AA)$}
$$M(\vect{f},W):=\{(m_1,\dots,m_n)\in\ZZ^n|\,\sigma(g)=f_1^{m_1}
\dots f_n^{m_n}\,g\text{ for some }g\in W\setminus\{0\}\}.$$

\noindent In the following, we want to calculate a finite representation of $M(\vect{f},\AA)$. If $\AA$ is a field, i.e., $\AA^*=\AA\setminus\{0\}$, it is immediate that $M(\vect{f},\AA)$ is a submodule of $\ZZ^n$ over $\ZZ$ and there is a basis of $M(\vect{f},\AA)$ with rank $\leq n$; see~\cite{Karr:81}. In the setting of rings, this result carries over if the set of semi-constants (also called semi-invariants~\cite{Bron:00})\index{semi-constant}\index{$\gsconstF{\AA}{\sigma}$, $\gsconst{\AA}{\sigma}$}  
of $\dfield{\AA}{\sigma}$ defined by
$$\gsconstF{\AA}{\sigma}=\{c\in\AA| \sigma(c)=u\,c\text{ for some
$u\in\AA^*$}\}$$
forms a multiplicative group (excluding the $0$ element). Note: if $\AA$ is a field, we have that $\gsconstF{\AA}{\sigma}\setminus\{0\}=\AA\setminus\{0\}=\AA^*$. 
Unfortunately, for a general difference ring the set $\gsconstF{\AA}{\sigma}\setminus\{0\}$ is only a multiplicative monoid~\cite{Bron:00}. In order to gain more flexibility, we introduce the following refinement.
For a given multiplicative subgroup $G$ of $\AA^*$ (in short $G\leq\AA^*)$, we define the set of
semi-constants (semi-invariants)\index{$\sconstF{G}{\AA}{\sigma}$, $\sconst{G}{\AA}{\sigma}$} of $\dfield{\AA}{\sigma}$ over $G$ by
$$\sconstF{G}{\AA}{\sigma}=\{c\in\AA| \sigma(c)=u\,c\text{ for some
$u\in G$}\}.$$
Note that $\sconstF{(\AA^*)}{\AA}{\sigma}=\gsconstF{\AA}{\sigma}$ and $\sconstF{\{1\}}{\AA}{\sigma}=\constF{\AA}{\sigma}$. 
If it is clear from the context, we drop $\sigma$ and just write
$\sconst{G}{\AA}{\sigma}$ and $\gsconst{\AA}{\sigma}$, respectively.

Here is one of the main challenges: For all our considerations we will choose $G$ such that $\sconst{G}{\AA}{\sigma}\setminus\{0\}$ is a subgroup of $\AA^*$ (in short, 
$\gsconst{\AA}{\sigma}\setminus\{0\}\leq\AA^*$).
Then with this careful choice of $G$ we can summarize the above considerations with the following lemma.

\begin{lemma}\label{Lemma:MZModule}
Let $\dfield{\AA}{\sigma}$ be a difference ring and let $G\leq\AA^*$ with 
$\sconst{G}{\AA}{\sigma}\setminus\{0\}\leq\AA^*$; let $\vect{f}\in G^n$. Then
$M(\vect{f},\AA)=M(\vect{f},\sconst{G}{\AA}{\sigma})$.
In particular, $M(\vect{f},\AA)$ is a submodule of $\ZZ^n$ over $\ZZ$, and it has a finite $\ZZ$-basis with rank $\leq n$.
\end{lemma}

In the light of this property, we can state Problem~PMT.

\medskip

\noindent\fbox{\begin{minipage}{12.8cm}
\noindent\index{Problem!PMT}\textbf{Problem~PMT in $\dfield{\AA}{\sigma}$
for $G$.} \textit{Given} a difference ring
$\dfield{\AA}{\sigma}$ with $G\leq\AA^*$ such that
$\sconst{G}{\AA}{\sigma}\setminus\{0\}\leq\AA^*$ holds;  given $\vect{f}\in G^n$.
\textit{Find} a
$\ZZ$-basis of $M(\vect{f},\AA)$.
\end{minipage}}

\medskip

\noindent Observe that Problem~MT can be reduced to Problem~PMT for a group $G$ with $\sconst{G}{\AA}{\sigma}\setminus\{0\}\leq\AA^*$ if we restrict\footnote{Note that this restriction, in particular the choice of $G$, is fundamental: it is the essential step to specify the type of products that one can handle algorithmically; see Definition~\ref{Def:SimpleRPS}.} to the situation that $\alpha\in G$. More precisely, assume that we have calculated $\lambda=\ord(\alpha)$ and succeeded in solving Problem~PMT, i.e., we are given a basis of $M=M((\alpha),\AA)\subseteq \ZZ^1$. If the basis is empty, there cannot be an $m\in\ZZ\setminus\{0\}$ and a $g\in\AA\setminus\{0\}$ with~\eqref{Equ:MultTele}. Otherwise, if the basis is not empty, the rank is $1$. More precisely, we obtain $m>0$ with $M=m\,\ZZ$. Hence $m$ is the smallest positive choice such that there is a $g\in\AA\setminus\{0\}$ with~\eqref{Equ:MultTele}. Therefore we can again decide\footnote{Note: If $\lambda:=\ord(\alpha)>0$, we have that $\lambda\in M$, i.e., the rank of $M$ is $1$. In particular, we can construct an \rE-extension $\dfield{\AA}{\sigma}\leq\dfield{\AA[t]}{\sigma}$ with $\sigma(t)=\alpha\,t$ iff $\lambda=m>0$.} Problem MT.
\medskip

For the generalization of Problems~T and~PT we introduce the following set.
Let $\dfield{\AA}{\sigma}$ be a difference ring with constant field $\KK$, let $W\subseteq\AA$, and let $u\in\AA\setminus\{0\}$ and
$\vect{f}=(f_1,\dots,f_n)\in\AA^n$. Then we define~\cite{Karr:81}\index{$V(u,\vect{f},\AA)$}
$$V(u,\vect{f},\dfield{W}{\sigma})=\{(c_1,\dots,c_n,g)\in\KK^n\times
W|\,\sigma(g)-u\,g=c_1\,f_1+\dots+c_n\,f_n\};$$
if it is clear from the context, we write
$V(u,\vect{f},W)$ and suppress the automorphism $\sigma$.

\noindent As with Lemma~\ref{Lemma:MZModule} the following result will be crucial for  further considerations.

\begin{lemma}\label{Lemma:VBasis}
Let $\dfield{\AA}{\sigma}$ be a difference ring with constant field $\KK$ and
let $G\leq\AA^*$ with $\sconst{G}{\AA}{\sigma}\setminus\{0\}\leq\AA^*$. Let
$W$ be a
$\KK$-subspace of $\AA$. Then for $\vect{f}\in\AA^n$ and $u\in G$ we have
that $V(u,\vect{f},W)$ is a $\KK$-subspace of $\KK^n\times W$ with $\dim
V(u,\vect{f},W)\leq n+1$.
\end{lemma}
\begin{proof}
Suppose that there are $m$ linearly independent solutions with $m>n+1$, say 
$(c_{i,1},\dots,c_{i,n},g_i)$ with $1\leq i\leq m$. Then by row operations over the field $\KK$ we can derive at least two linearly independent vectors, say
$\vect{v_1}=(0,\dots,0,g)$ and $\vect{v_2}=(0,\dots,0,h)$. Hence we have that $\sigma(g)=u\,g$
and $\sigma(h)=u\,h$ where
$g,h\in\sconst{G}{\AA}{\sigma}\setminus\{0\}\leq\AA^*$.
Consequently, $\sigma(\frac{g}{h})=\frac{g}{h}$, thus $c=g/h\in\KK^*$ and
therefore $\vect{v_1}=c\,\vect{v_2}$; a contradiction that the vectors are linearly independent.
\end{proof}

This result gives rise to the following problem specification.

\medskip

\noindent\fbox{\begin{minipage}{12.8cm}
\index{Problem!FPLDE}\noindent\textbf{Problem~PFLDE in $\dfield{\AA}{\sigma}$ for $G$ (with constant field $\KK$).} \textit{Given} a difference ring
$\dfield{\AA}{\sigma}$ with constant field $\KK$ and $G\leq\AA^*$ such that
$\sconst{G}{\AA}{\sigma}\setminus\{0\}\leq\AA^*$ holds; given $u\in G$ and
$\vect{f}\in\AA^n$.
\textit{Find} a
$\KK$-basis of $V(u,\vect{f},\AA)$.
\end{minipage}}
\medskip

\noindent In particular, if we can solve Problem~PFLDE in $\dfield{\AA}{\sigma}$ for $G$, it follows with $1\in G$ that we can solve Problem~T and~PT in $\dfield{\AA}{\sigma}$. Furthermore, we can solve the multiplicative version of telescoping: if $\alpha\in G$, we can determine a $g\in\AA\setminus\{0\}$, in case of existence, such that $\sigma(g)=\alpha\,g$ holds. This feature is illustrated by the following example.

\begin{example}[Cont.\ Ex.~\ref{Exp:QkxtCheckRPiExt}]\label{Exp:ProductExample}
Given $Q(b)=\prod_{k=1}^b\frac{-(\iota^k)}{1+k}$ on the left hand side of~\eqref{Equ:ProductId}, we want to rewrite it in terms of the product $P(b)=\prod_{j=1}^{b-1}j\iota^j$. In a preparation step we constructed already the 
\rpisiSE-extension $\dfield{\KK(k)[x]\lr{t}}{\sigma}$ of $\dfield{\KK(k)}{\sigma}$ with $\KK=\QQ(\iota)$, $\sigma(x)=\iota\,x$ and $\sigma(t)=k\,x\,t$ in
Example~\ref{Exp:QkxtCheckRPiExt}. There we can represent $\frac{-(\iota^k)}{k+1}$ with $u=\frac{-x}{k+1}$ and $P(k)$ with $t$.
Now we search for a $g\in\KK(k)[x]\lr{t}\setminus\{0\}$ such that
$\sigma(g)=u\,g$
holds. 
More precisely, we are interested in a basis of $V=V(u,(0),\KK(k)[x]\lr{t})$.
Activating our machinery, we get the basis $\{(0,g),(1,0)\}$ of $V$ with $g=\frac{x(\iota+x^2)}{k}t^{-1}$. For the chosen group $G$ with $u\in G$, that we use to solve the underlying Problem~PFLDE in $\dfield{\KK(k)[x]\lr{t}}{\sigma}$, and the corresponding calculation steps we refer to Example~\ref{Exp:ProductBound} below.
Since $g$ is a solution of $\sigma(g)=u\,g$, 
$g(k)=\big(\iota+(-1)^k\big) \frac{\iota^k}{k}P(k)^{-1}$ is a solution of
$-\frac{\iota^k}{k+1}=\frac{g(k+1)}{g(k)}.$
Hence by the telescoping trick we get
$\prod_{k=1}^b-\frac{\iota^k}{k+1}=\frac{g(b+1)}{g(1)}$
which produces~\eqref{Equ:ProductId}.
\end{example}

\subsection{The main results}\label{Subsec:MainResults}

Suppose that we are given a difference ring $\dfield{\GG}{\sigma}$ which is computable and we are given a group $G\leq\GG^*$ with $\sconst{G}{\GG}{\sigma}\setminus\{0\}\leq\GG^*$. In this article we will restrict to certain classes of \rpisiSE-extensions $\dfield{\EE}{\sigma}$ of $\dfield{\GG}{\sigma}$ equipped with a group $\tilde{G}$ with $G\leq\tilde{G}\leq\EE^*$ and $\sconst{\tilde{G}}{\EE}{\sigma}\setminus\{0\}\leq\EE^*$ such that we can derive the following algorithmic machinery: 
\begin{enumerate}
 \item Problem O in $\tilde{G}$ can be reduced to Problem O in $G$; 
 \item Problem PMT in $\dfield{\EE}{\sigma}$ for $\tilde{G}$ can be reduced to Problem PMT in $\dfield{\GG}{\sigma}$ for $G$; 
 \item Problem PFLDE in $\dfield{\EE}{\sigma}$ for $\tilde{G}$ can be reduced to Problem PFLDE in $\dfield{\GG}{\sigma}$ for $G$ (see Subsection~\ref{Sec:Single-Rooted}) or to Problem PFLDE in
 $\dfield{\GG}{\sigma^k}$ for $G$ for all $k\geq1$ (see Subsection~\ref{Sec:SimpleConstantStable}).
\end{enumerate}
In a nutshell, if we choose as base case a difference ring $\dfield{\GG}{\sigma}$ and a group $G\leq\GG^*$ in which we can solve Problem O in $G$ and Problems PMT and PFLDE in $\dfield{\GG}{\sigma}$ for $G$ (resp.\ $\dfield{\GG}{\sigma^k}$ for $G$ for all $k\geq1$), we obtain recursive algorithms that solve the corresponding problems in the larger difference ring $\dfield{\EE}{\sigma}$ and larger group $\tilde{G}$.

As it turns out, we will succeed in this task for a subclass of \rpisiSE-extensions $\dfield{\GG}{\sigma}\leq\dfield{\EE}{\sigma}$ and a properly chosen group $\tilde{G}\leq\EE^*$ that can treat all objects (among the general class of \rpisiSE-extensions) that the author has encountered in practical problem solving so far. More precisely, we will restrict to simple \rpisiSE-extensions.

Let
$\dfield{\GG\lr{t_1}\dots\lr{t_e}}{\sigma}$
be a \rpisiSE-extension of
$\dfield{\GG}{\sigma}$ and let $G\leq\GG^*$. Then we define\index{product group}\index{$\dgroup{G}{\AA}{\GG}$}
\begin{equation}\label{Def:dgroup}
\dgroup{G}{\EE}{\GG}=\{g\,t_1^{m_1}\dots
t_e^{m_e}|\,h\in G\text{ and }m_i\in\ZZ\text{ where $m_i=0$ if $t_i$ is a
\sigmaSE-monomial}\}.
\end{equation}
It is easy to see that $\tilde{G}=\dgroup{G}{\EE}{\GG}$ forms a group. More precisely, we obtain the following chain of subgroups:
$G\leq\dgroup{G}{\EE}{\GG}\leq\EE^*.$ We call $\dgroup{G}{\EE}{\GG}$ also the product-group over $G$ for the \rpisiSE-extension $\dfield{\EE}{\sigma}$ of $\dfield{\GG}{\sigma}$. We are now ready to define ($G$--)simple \rpisiSE-extensions. 

\begin{definition}\label{Def:SimpleRPS}
Let $\dfield{\GG}{\sigma}$ be a difference ring and let $G\leq\GG^*$ be a
group.
An \rpisiSE-extension $\dfield{\EE}{\sigma}$ of
$\dfield{\GG}{\sigma}$ with $\EE=\GG\lr{t_1}\lr{t_2}\dots\lr{t_e}$ is called
$G$-simple\index{extension!simple, $G$-simple} if for any \rE\piE-monomial $t_i$ we have that $\sigma(t_i)/t_i\in\dgroup{G}{\EE}{\GG}$. 
Moreover, an \rE\piE, \rE\sigmaSE, \pisiSE-, \rE-, \piE-, and \sigmaSE-extension
of $\dfield{\GG}{\sigma}$ is $G$-simple if it is a $G$-simple \rpisiSE-extension.
We call any such extension simple if it is $\GG^*$-simple. Analogously, we call an \rE\piE, \rE\sigmaSE, \pisiSE-, \rE-, \piE-, and \sigmaSE-monomial $G$-simple (resp.\ simple) if the extension is $G$-simple (resp.\ simple).\index{monomial!simple, $G$-simple}
\end{definition}

\noindent In all our examples the difference rings have been built by a simple \rpisiSE-extension $\dfield{\EE}{\sigma}$ of $\dfield{\GG}{\sigma}$ where $\dfield{\GG}{\sigma}$ is a \pisiSE-field $\dfield{\KK(k)}{\sigma}$ 
over $\KK$ with $\sigma(k)=k+1$. In particular, the Problems~PMT and PFLDE have been considered for the constructed $\dfield{\EE}{\sigma}$ in $G=\dgroup{(\KK(k)^*)}{\EE}{\KK(k)}$. Before we finally turn to the class of simple \rpisiSE-extensions, we present one example which cannot be treated properly with our toolbox under consideration.

\begin{example}\label{Exp:NotSimpleExt}
Take $\dfield{\QQ(k)\ltr{t}}{\sigma}$ from Example~\ref{Exp:FactorialPi} with $\sigma(k)=k+1$ and $\sigma(t)=(k+1)\,t$. Subsequently, we will use our notation
$\QQ(k)\lr{t}=\QQ(k)\ltr{t}$. Then we can construct the \rE-extension $\dfield{\QQ(k)\lr{t}[x]}{\sigma}$ of $\dfield{\QQ(k)\lr{t}}{\sigma}$ with $\sigma(x)=-x$ of order $2$. In this ring we are given the idempotent elements $e_1=(1-x)/2$ and $e_2=(x+1)/2$ with $e_1^2=e_1$ and $e_2^2=e_2$. Finally take $\alpha=e_1+e_2\,t$. Then observe that $\alpha\cdot(e_1+e_2/t)=1$, i.e., $\alpha\in\QQ(k)\lr{t}[x]^*$. Note that $\ord(\alpha)=0$. Otherwise it would follow that $e_2^{\lambda}=0$ with $\lambda=\ord(\alpha)>0$; a contradiction that $e_2$ is idempotent. Consequently, $T$ cannot be an \rE-extension, and we construct the unimonomial extension $\dfield{\QQ(k)\lr{t}[x][T,\frac{1}{T}]}{\sigma}$ of $\dfield{\QQ(k)\lr{t}[x]}{\sigma}$ with $\sigma(T)=\alpha\,T$. It seems non-trivial to derive an (algorithmic) proof (or disproof) that $T$ is a \piE-monomial, and it would be nice to see a solution to this problem.
\end{example}

\noindent Summarizing, we aim at solving Problems PMT and PFLDE 
in a $G$-simple \rpisiSE-extension $\dfield{\GG}{\sigma}\leq\dfield{\EE}{\sigma}$ for $\tilde{G}=\dgroup{G}{\EE}{\GG}$, and we want to solve Problem O in $\tilde{G}$. In order to accomplish this task, we will restrict ourselves further to  the following two  situations.

\subsubsection{A solution for single-rooted \rpisiSE-extensions}\label{Sec:Single-Rooted}

In most applications \rE-extensions are not nested, e.g., only objects like $(-1)^k$ arise. In addition, such objects do not occur in transcendental products, but only in sums, like cyclotomic sums~\cite{ABS:11} or generalized harmonic sums~\cite{ABS:13}. A formal definition of this special, but very practical oriented class of \rpisiSE-extensions is as follows.

\begin{definition}\label{Def:SingleRooted}
An \rpisiSE-extension $\dfield{\EE}{\sigma}$ of $\dfield{\GG}{\sigma}$ is called single-rooted\index{extension!single-rooted} if the generators of the extension can be reordered to 
\begin{equation}\label{Equ:SingleRooted}
\EE=\GG\lr{t_1}\dots\lr{t_r}\lr{x_{1}}\dots\lr{x_u}\lr{s_1}\dots\lr{s_v},
\end{equation}
respecting the recursive nature of the automorphism, such that the $t_i$ are \piE-monomials, the $x_i$ are \rE-monomials with $\sigma(x_i)/x_i\in\GG^*$ and the $s_i$ are \sigmaSE-monomials.
\end{definition}

\noindent Given this class of single-rooted and simple\footnote{Note: If $\GG$ is a field, any single-rooted \rpisiSE-extension is simple by Corollary~\ref{Cor:SingleRootedIsSimple}.} \rpisiSE-extension, we will show the following theorem in Proof~\ref{Proof:sconstIsGroupRing}.

\begin{theorem}\label{Thm:sconstIsGroupRing}
Let $\dfield{\GG}{\sigma}$ be a difference ring and let $G\leq\GG^*$ with $\sconst{G}{\GG}{\sigma}\setminus\{0\}\leq\GG^*$.
Let $\dfield{\EE}{\sigma}$ 
be a simple and single-rooted \rpisiSE-extension of $\dfield{\GG}{\sigma}$ with~\eqref{Equ:SingleRooted} as specified in Definition~\ref{Def:SingleRooted}, and let
$\tilde{G}=\dgroup{G}{\GG\lr{t_1}\dots\lr{t_r}}{\GG}$.
Then $\sconst{\tilde{G}}{\EE}{\sigma}\setminus\{0\}\leq\EE^*$ with
\begin{equation*}
\sconst{\tilde{G}}{\EE}{\sigma}=\{h\,t_1^{m_1}\dots t_r^{m_r}\,x_1^{n_1}\dots x_u^{n_u}|\,h\in\sconst{G}{\GG}{\sigma},
m_i\in\ZZ\text{ and } n_i\in\NN\}.
\end{equation*}
\end{theorem}

\noindent In particular, we obtain the following reduction algorithms summarized in Theorem~\ref{Thm:AlgMainResultRestricted}; for a proof of part~1 see Proof~\ref{Proof:AlgMainResultRestricted1} and of part~2 see Proof~\ref{Proof:AlgMainResultRestricted2}.

\begin{theorem}\label{Thm:AlgMainResultRestricted}
Let $\dfield{\GG}{\sigma}$ be a computable difference ring with $G\leq\GG^*$ and $\sconst{G}{\GG}{\sigma}\setminus\{0\}\leq\GG^*$. Let $\dfield{\EE}{\sigma}$ be a single-rooted and $G$-simple \rpisiSE-extension of $\dfield{\GG}{\sigma}$ with~\eqref{Equ:SingleRooted} as given in Definition~\ref{Def:SingleRooted}, and let
$\tilde{G}=\dgroup{G}{\GG\lr{t_1}\dots\lr{t_r}}{\GG}$. Then the following holds.
\begin{enumerate}
 \item Problem~PMT is solvable in $\dfield{\EE}{\sigma}$ for $\tilde{G}$ if it is solvable in $\dfield{\GG}{\sigma}$ for $G$.
 \item Problem~PFLDE is solvable in $\dfield{\EE}{\sigma}$ for $\tilde{G}$ if Problems~PFLDE and PMT are solvable in $\dfield{\GG}{\sigma}$ for $G$ and if\footnote{Instead of Problem~O it suffices if know the orders of all the \rE-monomials in $\dfield{\GG}{\sigma}\leq\dfield{\EE}{\sigma}$.} Problem~O is solvable in $G$.
\end{enumerate}
\end{theorem}

\noindent All the calculations in~\cite{Schneider:04c,PSW:11,Schneider:07b,Schneider:09a,Physics1,Physics2,Physics3} rely precisely on this machinery. For one of the most important applications we refer to Subsection~\ref{Subsec:dAlembert}.

\subsubsection{A solution for simple \rpisiSE-extensions of a strong constant-stable difference field}\label{Sec:SimpleConstantStable}

In the following we restrict to simple \rpisiSE-extensions where the ground domain $\GG=\FF$ is a field. In this setting, the semi-constants form a multiplicative group. More precisely, we will show the following result in Proof~\ref{Proof:sconstIsGroupField}.

\begin{theorem}\label{Thm:sconstIsGroupField}
Let $\dfield{\EE}{\sigma}$
be a simple \rpisiSE-extension of a difference field $\dfield{\FF}{\sigma}$ and consider its
product-group $\tilde{G}=\dgroup{(\FF^*)}{\EE}{\FF}$.
Then $\sconst{\tilde{G}}{\EE}{\sigma}\setminus\{0\}\leq\EE^*$. 
\end{theorem}

\noindent For a solution of Problems PMT and PFLDE we require in addition that $\dfield{\FF}{\sigma}$ is strong constant-stable.

\begin{definition}\label{Def:ConstantStable}
A difference ring $\dfield{\AA}{\sigma}$ with constant field $\KK$ is called
constant-stable if for all $k>0$ we
have that $\constF{\AA}{\sigma^k}=\KK$. It is called strong constant-stable if
it is constant-stable\index{ring!constant-stable}\index{ring!(strong) constant-stable} and any root of unity of $\AA$ is in $\KK$. 
\end{definition}

\noindent In this setting we can treat products over roots of unity from $\KK$ and, more generally, products that are built recursively over such products; for examples see~\eqref{Equ:Equ:AlgebraicObjects} and for further (algorithmic) properties see Corollary~\ref{Cor:SimpleNestedRBasics} below.
More precisely, given such a tower of \rpisiSE-extensions, we can solve Problems PMT and PFLDE as follows; for the proofs, resp. the underlying algorithms, of part~1 see Proof~\ref{Proof:AlgMainResultFull1}, of part~2 see Proof~\ref{Proof:AlgMainResultFull2} and of part~3  see Proof~\ref{Proof:AlgMainResultFull3}.

\begin{theorem}\label{Thm:AlgMainResultFull}
Let $\dfield{\FF}{\sigma}$ be a computable difference field where Problem~O is solvable in $(\const{\FF}{\sigma})^*$.
Let $\dfield{\EE}{\sigma}$ be a simple \rpisiSE-extension of
$\dfield{\FF}{\sigma}$. Then the following holds.
\begin{enumerate}
\item Problem~O is solvable in $\dgroup{(\FF^*)}{\EE}{\FF}$.

\vspace*{0.1cm}

\item[If ] \hspace*{-0.2cm}$\dfield{\FF}{\sigma}$ is in addition strong constant-stable, then
\item Problem~PMT is solvable in $\dfield{\EE}{\sigma}$ for $\dgroup{(\FF^*)}{\EE}{\FF}$ if it is solvable in $\dfield{\FF}{\sigma}$ for $\FF^*$;
\item Problem~PFLDE is solvable in $\dfield{\EE}{\sigma}$ for
$\dgroup{(\FF^*)}{\EE}{\FF}$ if Problem~PMT is solvable
in
$\dfield{\FF}{\sigma}$ for $\FF^*$ and Problem~PFLDE is solvable\footnote{We emphasize that we will always work with the automorphism $\sigma$ during the reduction process. Only in the base cases we might face the problem to solve instances of Problem PFLDE in $\dfield{\FF}{\sigma^k}$ with $k>1$. In a nutshell, we succeed in avoiding to work with $\sigma^k$ for some $k>1$ as much as possible.
This strategy is of particular advantage, if $\dfield{\FF}{\sigma}$ is built only by few summation objects. Then the typical phenomenon of the expression swell in symbolic summation due to $\sigma^k$ is prevented as much as possible.}
in $\dfield{\FF}{\sigma^k}$ for $\FF^*$ for all $k>0$.
\end{enumerate}
\end{theorem}

\noindent We remark that this reduction machinery has been utilized in Examples~\ref{ExpQkxysS} and~\ref{Exp:ProductExample}
to obtain the identities~\eqref{SumId} and~\eqref{Equ:ProductId}, respectively. Further details will be given below.

\subsubsection{A complete machinery: algorithms for the ground difference rings}\label{Subsec:GroundFieldAlg}

Both, Theorems~\ref{Thm:AlgMainResultRestricted} and~\ref{Thm:AlgMainResultFull} provide algorithms to reduce the Problems~PMT and PFLDE (and thus the Problems T, PT and special cases of Problem MT) from an \rpisiSE-extension $\dfield{\EE}{\sigma}$ of $\dfield{\GG}{\sigma}$ to the ground difference ring $\dfield{\GG}{\sigma}$. Theorem~\ref{Thm:AlgMainResultRestricted} requires less conditions on $\dfield{\GG}{\sigma}$, but considers only single--rooted \rpisiSE-extensions, whereas Theorem~\ref{Thm:AlgMainResultFull} requires more properties on $\dfield{\GG}{\sigma}$ but allows nested \rE-extensions which are of the type as given in Corollary~\ref{Cor:SimpleNestedRBasics} below. Note that the algorithms for the latter case are more demanding, in particular, one has to solve Problem~PFLDE in $\dfield{\GG}{\sigma^k}$ with $k>0$ instead of $k=1$ only.

We emphasize that both theorems are applicable for a rather general class of difference fields $\dfield{\GG}{\sigma}$. Namely, $\dfield{\GG}{\sigma}$ itself can be a \pisiSE-field extension of $\dfield{\HH}{\sigma}$ where certain properties in the difference field $\dfield{\HH}{\sigma}$ hold. Here the following remarks are in place.\\
(a) By~\cite{Karr:81} a \pisiSE-field extension $\dfield{\GG}{\sigma}$ of $\dfield{\HH}{\sigma}$ is constant-stable if $\dfield{\HH}{\sigma}$ is constant-stable. In particular, if we are given a root of unity from $\GG$, it cannot depend on transcendental elements and is therefore from $\HH$. Thus $\dfield{\GG}{\sigma}$ is strong constant-stable if $\dfield{\HH}{\sigma}$ is strong constant-stable.\\ 
(b) It has been shown in~\cite{Schneider:06d} that one can solve Problem~PMT in $\dfield{\GG}{\sigma}$ for $\GG^*$ and Problem~PFLDE in $\dfield{\GG}{\sigma^k}$ for $\GG^*$ for $k>0$ if certain properties hold for the difference field $\dfield{\HH}{\sigma}$. Among others (see Def.~1 and~2 in \cite{Schneider:06d}) Problem~PMT must be solvable in $\dfield{\HH}{\sigma}$ for $\HH^*$ and Problem~PFLDE must be solvable in $\dfield{\HH}{\sigma^k}$ for $\HH^*$.

\noindent Summarizing, if we are given the tower of extensions
$$\dfield{\HH}{\sigma}\stackrel{\text{\pisiSE-field ext.}}{\leq}\dfield{\GG}{\sigma}\stackrel{\text{\rpisiSE-ring ext.}}{\leq}\dfield{\EE}{\sigma}$$
where $\dfield{\HH}{\sigma}$ is strong constant-stable and the properties given in Def.~1 and~2 of \cite{Schneider:06d} hold in $\dfield{\HH}{\sigma}$, then we can solve Problems PMT and PFLDE in $\dfield{\EE}{\sigma}$ for $\dgroup{(\GG^*)}{\EE}{\GG}$.

So far, the required properties have been verified and the necessary algorithms have been worked out for the following difference fields $\dfield{\HH}{\sigma}$ with constant field $\KK$.

\begin{enumerate}
 \item $\KK=\HH$, i.e., $\dfield{\GG}{\sigma}$ is a \pisiSE-field over $\KK$; here  the constant field $\KK$ can be a rational function field over an algebraic number field; see~\cite[Theorem~3.5]{Schneider:05c}.
 
 \item $\dfield{\HH}{\sigma}$ is a free difference field, i.e., $\HH=\KK(\dots,x_{-1},x_0,x_1,\dots)$ with $\sigma(x_i)=x_{i+1}$; here $\KK$ is of the type as given in case (1). Note that in this field one can model unspecified sequences; see~\cite{Schneider:06d,Schneider:06e}.

 \item $\dfield{\HH}{\sigma}$ can be a radical difference field representing objects like $\sqrt[d]{k}$; see~\cite{Schneider:07f}. 
\end{enumerate}

\noindent For simplicity, all our examples are chosen from case (1). More precisely, we always take the \pisiSE-field $\dfield{\HH}{\sigma}=\dfield{\KK(k)}{\sigma}$ over $\KK\in\{\QQ,\QQ(\iota)\}$ with $\sigma(k)=k+1$.

\subsection{Application: representation of d'Alembertian solutions in \rpisiSE-extensions}\label{Subsec:dAlembert}

We illustrate how an important class of d'Alembertian solutions~\cite{Abramov:94}, a subclass of Liouvillian solutions~\cite{Singer:99,Petkov:2013}, of a given linear difference operator, can be represented completely automatically in \rpisiSE-extensions. In order to obtain the d'Alembertian solutions, one starts as follows: first the linear difference operator is factored as much as possible into linear right hand factors. This can be accomplished, e.g., with the algorithms from~\cite{Petkov:92,vanHoeij:99,HK:12} or, within the setting of \pisiSE-fields with the algorithms given in~\cite{ABPS:14} which are based on~\cite{Bron:00,Schneider:01,Schneider:05a}. The latter machinery is available within the summation package \texttt{Sigma}. Then given this factored form of the operator, the d'Alembertian solutions can be read off. They can be given by a finite number of hypergeometric expressions 
and indefinite nested sums defined over such expressions. More precisely, each solution is of the form
\begin{equation}\label{Equ:dAlembertianSol}
\sum_{i_1=\lambda_1}^kh_1(i_1)\sum_{i_2=\lambda_2}^{i_1}h_2(i_2)\dots \sum_{i_r=\lambda_{r-1}}^{i_{r-1}}h_r(i_r)
\end{equation}
where $\lambda_i\in\NN$ and the hypergeometric expression $h_i(k)$ can be written in the form $\prod_{j=\lambda_i}^k \alpha_i(j)$ with $\alpha_i(z)$ being a rational function from $\KK(z)$. 

Subsequently, we restrict ourselves to a field $\KK$ which is a rational function field
$\KK=\QQ(n_1,\dots,n_r)$ over the rational numbers. Now take the \pisiSE-field $\dfield{\KK(k)}{\sigma}$ over $\KK$ with $\sigma(k)=k+1$. Then the solutions, all being of the form~\eqref{Equ:dAlembertianSol}, can be represented in a single-rooted simple \rpisiSE-extension as follows.

(1) In~\cite[Section~6]{Schneider:05c} an algorithm has been presented that calculates a single-rooted simple \rE\piE-extension $\dfield{\GG}{\sigma}$ of $\dfield{\KK(k)}{\sigma}$ in which all hypergeometric expressions occurring in the d'Alembertian solutions are explicitly represented. 

(2) Then the challenging task is to construct a \sigmaSE-extension of $\dfield{\GG}{\sigma}$ and to represent there the arising sums of the d'Alembertian solutions. Given $\dfield{\GG}{\sigma}$ from step 1, this can be accomplished by applying iteratively Theorem~\ref{Thm:RPSCharacterization}.(1). Suppose we represented already an inner summand in a \sigmaSE-extension $\dfield{\AA}{\sigma}$ of $\dfield{\GG}{\sigma}$ with $\beta\in\AA$. Since $\dfield{\AA}{\sigma}$ is a simple \rpisiSE-extension of $\dfield{\KK(k)}{\sigma}$ and $\dfield{\KK(k)}{\sigma}$ is a \pisiSE-field over $\KK$, we can solve Problem~T with $f=\beta$ by using the underlying algorithm of Theorem~\ref{Thm:AlgMainResultRestricted} in combination with the base case algorithms; see Subsection~\ref{Subsec:GroundFieldAlg}. If we find a $g\in\AA$ with $\sigma(g)=g+\beta$, we can represent the sum under consideration with $g+c$ where $c\in\KK$ is determined by the boundary condition (lower summation bound) of the given sum; for further details we refer to Example~\ref{ExpQkxysS}.(4). Otherwise, we construct the \sigmaSE-extension $\dfield{\AA[t]}{\sigma}$ of $\dfield{\AA}{\sigma}$ with $\sigma(t)=t+\beta$ by Theorem~\ref{Thm:RPSCharacterization}.(1) and we succeeded in representing the sum under consideration by $t$ with the appropriate shift behaviour. Note that $\dfield{\AA[t]}{\sigma}$ is again a single-rooted simple \rpisiSE-extension of $\dfield{\KK(k)}{\sigma}$. Proceeding iteratively, all the nested hypergeometric sums are represented in terms of an \rpisiSE-extension over $\dfield{\KK(k)}{\sigma}$.

\smallskip

\noindent Exactly this difference ring machinery is implemented in \texttt{Sigma} and has been used to tackle challenging applications, like~\cite{Schneider:04c,PSW:11,Schneider:07b,Schneider:09a,Physics1,Physics2,Physics3} mentioned already in the introduction. In particular, this toolbox has been combined with the algorithms worked out in~\cite{Schneider:04a,Schneider:05c,Schneider:07d,Schneider:08c,Petkov:10,Schneider:14} in order to find representations of d'Alembertian solutions with certain optimality properties, like minimal nesting depth. For a recent summary of all these features (unfortunately, in the setting of difference fields) we refer to~\cite{Schneider:13a,Schneider:13b}.

\section{Single nested \rpisiSE-extensions}\label{Sec:SingleExt}

This section delivers relevant properties of single nested \rpisiSE-extensions. The characterization of \rpisiSE-extensions (Theorem~\ref{Thm:RPSCharacterization}) will be elaborated. In addition, properties of the semi-constants within \rpisiSE-extensions are derived to gain further insight in the nature of \rpisiSE-extensions and to prove Theorems~\ref{Thm:sconstIsGroupRing} and~\ref{Thm:sconstIsGroupField} in Section~\ref{Sec:NestedExt}.

We start with some general properties which will be essential throughout this article. 
\begin{definition}
A ring $\AA$ is called reduced\index{ring!reduced}\index{ring!connected} if there are no non-zero nilpotent elements, i.e.,
for any $f\in\AA\setminus\{0\}$ and any $n>0$ we have that $f^n\neq0$. $\AA$ is called connected if $0$ and $1$ are the only idempotent elements, i.e., for any $f\in\AA\setminus\{0,1\}$ we have that $f^2\neq f$. 
\end{definition}

\noindent Namely, we rely on the following ring properties. A polynomial $\sum_{i=0}^n a_ix^i\in\AA[t]$ with coefficients from a ring $\AA$ is 
invertible if and only if $a_0\in\AA^*$ and $a_i$ with $i\geq1$ are nilpotent elements. Thus in a reduced ring, i.e., a ring which has no nilpotent non-zero elements, we have that $\AA[t]^*=\AA^*$. Besides, there is a complete characterization of invertible elements in the ring of Laurent polynomials $\AA\ltr{t}$ presented in~\cite[Theorem~1]{Karpilovsky:83} (see also~\cite{Neher:09}).
Based on this work we extract the following crucial result.

\begin{lemma}\label{Lemma:PolyInv}\label{Lemma:LaurentInv}
Let $\AA$ be a commutative ring with 1. If $\AA$ is reduced, then $\AA[t]^*=\AA^*$. If $\AA$ is reduced and connected, then $\AA\ltr{t}^*=\{u\,t^r|\, u\in\AA^*\text{ and }r\in\ZZ\}.$
\end{lemma}

\noindent Since our rings are usually not connected, Lemma~\ref{Lemma:LaurentInv} can be applied only partially.

\begin{example}\label{Exp:NotConnected}
The generators in the ring given in Example~\ref{Exp:NotSimpleExt}
can be reordered to $\QQ(k)[x]\lr{t}$. Since $\QQ(k)[x]$ has the idempotent elements $e_1,e_2$, it is not connected. Therefore we get relations such as $(\e_1+e_2\,t)(e_1+\frac{e_2}{t})=1$ which are predicted in~\cite{Karpilovsky:83,Neher:09}.
\end{example}

Subsequently, we enumerate further definitions and properties in difference rings and fields that will be used throughout the article. 
Let $\dfield{\AA}{\sigma}$ be a difference ring. 
The rising factorial (or $\sigma$-factorial)\index{function!rising factorial} of $f\in\AA^*$ to $k\in\ZZ$ is defined by\index{$\sigmaFac{f}{k,\sigma}$, $\sigmaFac{f}{k}$}
$$\sigmaFac{f}{k,\sigma}=\begin{cases}
                  f\,\sigma(f)\dots\,\sigma^{k-1}(f)&\text{ if }k>0\\
                  1&\text{ if }k=0\\
\sigma^{-1}(f^{-1})\,\sigma^{-2}(f^{-1})\dots\sigma^{k}(f^{-1})&\text{
if }k<0.
                 \end{cases}$$
If the automorphism is clear from the context, we also will write
$\sigmaFac{f}{k}$ instead of $\sigmaFac{f}{k,\sigma}$.         
We will rely on the following simple identities
(compare also~\cite[page 307]{Karr:85}). The proofs are omitted to the reader.

\begin{lemma}\label{Lemma:SigmaFacId}
Let $\dfield{\AA}{\sigma}$ be a difference ring, $f,h\in\AA^*$ and
$n,m\in\ZZ$. Then:
\begin{enumerate}
\item $\sigmaFac{(f\,h)}{n}=\sigmaFac{f}{n}\,\sigmaFac{h}{n}$.
\item $\sigmaFac{f}{n+m}=\sigma^n(\sigmaFac{f}{m})\,\sigmaFac{f}{n}$.
\item $\sigmaFac{f}{n\,m}=\sigmaFac{(\sigmaFac{f}{n,\sigma})}{m,\sigma^n}$.
\item If $\sigma(h)=f\,h$, then $\sigma^n(h)=\sigmaFac{f}{n}\,h$.
\item $\sigma^k(f)\in\AA^*$ and $\sigmaFac{f}{n}\in\AA^*$.
\end{enumerate}
\end{lemma}

\noindent Let $\AA\lr{t}$ be a ring of (Laurent) polynomials. For $f=\sum_i f_i\,t^i\in\AA\lr{t}$ we define\index{$\deg$}\index{$\ldeg$}\index{function!degree}
\begin{equation*}
\deg(f)=\begin{cases}
\max\{i|f_i\neq0\}&\text{ if $f\neq0$}\\
-\infty&\text{ if $f=0$}
\end{cases}\quad\text{ and }\quad\ldeg(f)=\begin{cases}
\min\{i|f_i\neq0\}&\text{ if $f\neq0$}\\
\infty&\text{ if $f=0$}.         
\end{cases}
\end{equation*}
In addition, for $a,b\in\ZZ$ we introduce the set of truncated (Laurent) polynomials by\index{$\AA\lr{t}_{a,b}$}
\begin{equation}\label{Equ:TruncatedRing}
\AA\lr{t}_{a,b}=\{\sum_{i=a}^bf_i\,t^i|f_i\in\AA\}.
\end{equation}

\noindent We conclude this part with the following two lemmas.

\begin{lemma}\label{Lemma:BasicsinDR}
Let $\dfield{\AA\lr{t}}{\sigma}$ be a unimonomial ring extension of
$\dfield{\AA}{\sigma}$ of (Laurent) polynomial type. Then for any
$k\in\ZZ$ and $f\in\AA\lr{t}$ we have that $\deg(\sigma^k(f))=\deg(f)$.
\end{lemma}
\begin{proof}
Let $f=\sum_i f_i\,t_i$. If $f=0$, $\sigma^k(f)=0$ and thus with $\deg(0)=-\infty$ the statement holds. Otherwise, let $m:=\deg(f)\in\ZZ$. Then note that $\sigma^k(f)=\sum_i\sigma^k(f_i)(\sigma^k(t))^i$, i.e., $t^m$ is the largest possible monomial in $\sigma^k(f)$ with the coefficient $h:=\sigmaFac{\alpha}{k}^m\,\sigma^k(f_m)$. Since $\sigma^k(f_m)\neq0$ and $\sigmaFac{\alpha}{k}\in\AA^*$ by Lemma~\ref{Lemma:SigmaFacId}.(5), the coefficient $h$ is non-zero.  
\end{proof}

\begin{lemma}\label{LemmaA}
Let $\dfield{\FF(t)}{\sigma}$ be a unimonomial field extension of $\dfield{\FF}{\sigma}$, and let $p,q\in\FF[t]^*$ with $\gcd(p,q)=1$ and $k\in\ZZ$. Then the following holds. 
\begin{enumerate}
\item If $p\mid q$ then $\sigma^k(p)\mid\sigma^k(q)$. 
\item $\gcd(\sigma^k(p),\sigma^k(q))=1$. 
 \item $\frac{\sigma(p/q)}{p/q}\in\FF$ if and only if $\sigma(p)/p\in\FF$ and $\sigma(q)/q\in\FF$.
\end{enumerate}
 \end{lemma}
\begin{proof}
(1) If $p\mid q$, i.e., $p\,w=q$ for some $w\in\FF[t]\setminus\{0\}$, then $\sigma^k(p)=\sigma^k(w)\,\sigma^k(q)$, and thus $\sigma^k(p)\mid \sigma^k(q)$.
(2) Suppose that $1\neq\gcd(\sigma^k(p),\sigma^k(q))=:u\in\FF[t]\setminus\FF$. Then $\sigma^{-k}(u)\in\FF[t]\setminus\FF$. Since $\sigma^{-k}(u)\mid p$ and $\sigma^{-k}(u)\mid q$ by part~1 of the lemma, $\gcd(p,q)\neq1$, a contradiction to the assumption.
(3) The implication $\Leftarrow$ is immediate. Suppose that
$u:=\sigma(p/q)/(p/q)\in\FF$, i.e., $\sigma(p)\,q=u\,p\,\sigma(q)$. By part~2 of the
lemma, $\sigma(p)\mid p$ and $p\mid\sigma(p)$ which implies that
$\sigma(p)/p\in\FF$. Analogously, it follows that $\sigma(q)/q\in\FF$.
\end{proof}

\subsection{\sigmaSE-extensions}

The essence of all the properties of \sigmaSE-extensions is contained in the following lemma.

\begin{lemma}\label{LemmaB}
Let $\dfield{\AA}{\sigma}$ be a difference ring and let $G\leq \AA^*$ with 
$\sconst{G}{\AA}{\sigma}\setminus\{0\}\leq\AA^*$. Let $\dfield{\AA[t]}{\sigma}$
be a unimonomial ring extension of $\dfield{\AA}{\sigma}$ with
$\sigma(t)=t+\beta$ for some $\beta\in\AA$. If there are a $u\in G$ and a 
$g\in\AA[t]$ with $\deg(g)\geq1$ such that
\begin{equation}\label{Equ:LemmaDegConst}
\deg(\sigma(g)-u\,g)<\deg(g)-1
\end{equation}
holds, then there is a $\gamma\in\AA$ with
$\sigma(\gamma)-\gamma=\beta$.
\end{lemma}
\begin{proof}
Let $g=\sum_{i=0}^ng_it^i\in\AA[t]$ with $\deg(g)=n\geq1$ and $u\in G$ as stated in the lemma, and define $f=\sigma(g)-u\,g\in\AA[t]$. With~\eqref{Equ:LemmaDegConst} it follows that $f=\sum_{i=0}^{n-2}f_i\,t^i$.
Thus comparing the $n$th and $(n-1)$th coefficient in
$\sum_{i=0}^{n-2}f_i\,t^i=f=\sigma(g)-ug=\sum_{i=0}^n\sigma(g_i)(t+\beta)^i-u\sum_{i=0}^ng_it^i$
and using $(t+\beta)^i=\sum_{j=0}^i\binom{i}{j}t^{i-j}\beta^j$ for
$0\leq i\leq n$ yield
\begin{align*}
\sigma(g_n)-ug_n=0&&\text{ and }&&
\sigma(g_{n-1})+\sigma(g_n)\tbinom{n}{1}\beta-ug_{n-1}=0.
\end{align*}
The first equation shows that $g_n\in\sconst{G}{\AA}{\sigma}\setminus\{0\}\leq\AA^*$. Hence we get
$u=\sigma(g_n)/g_n$. Substituting $u$ for $\sigma(g_n)/g_n$ in the second equation gives
$\sigma(g_{n-1})-\tfrac{\sigma(g_n)}{g_n}g_{n-1}=-n\beta\sigma(g_n).$
Dividing this equation by $-n\,\sigma(g_n)\in\AA^*$ yields
$\sigma(\gamma)-\gamma=\beta$ with $\gamma:=\frac{-g_{n-1}}{ng_n}\in\AA$.
\end{proof}

\noindent Lemma~\ref{LemmaB} leads to the following equivalent properties
of \sigmaSE-extensions.

\begin{lemma}\label{Lemma:SigmaLemma}
Let $\dfield{\AA}{\sigma}$ be a difference ring and let $G\leq\AA^*$ with 
$\sconst{G}{\AA}{\sigma}\setminus\{0\}\leq\AA^*$.
Let $\dfield{\AA[t]}{\sigma}$ be a unimonomial ring extension of
$\dfield{\AA}{\sigma}$ with $\sigma(t)=t+\beta$ for some $\beta\in\AA$. Then the
following statements are equivalent.
\begin{enumerate}
\item There is a $g\in\AA[t]\setminus\AA$ and $u\in G$ with $\sigma(g)=u\,g$.
\item There is a $g\in\AA$ with $\sigma(g)=g+\beta$.
\item $\const{\AA[t]}{\sigma}\supsetneq\const{\AA}{\sigma}.$
\end{enumerate}
\end{lemma}
\begin{proof}
$(1)\!\!\Rightarrow\!\!(2)$: Let $g\in\AA[t]\setminus\AA$, $u\in G$ with
$\sigma(g)=u\,g$. Since $\deg(g)\geq1$ and
$\deg(\sigma(g)-u\,g)<0\leq\deg(g)-1$, there is a
$\gamma\in\AA$ with $\sigma(\gamma)=\gamma+\beta$ by Lemma~\ref{LemmaB}.\\
$(2)\Rightarrow(3)$: Let $g\in\AA$ with $\sigma(g)=g+\beta$. Since
$\sigma(t)=t+\beta$, it follows that $\sigma(t-g)=(t-g)$, i.e.,
$t-g\in\const{\AA[t]}{\sigma}$. Since $t-g\notin\AA$,
$t-g\notin\const{\AA}{\sigma}$.\\
$(3)\Rightarrow(1)$: Suppose that 
$\const{\AA}{\sigma}\subsetneq\const{\AA[t]}{\sigma}$ and take
$g\in\const{\AA[t]}{\sigma}\setminus\const{\AA}{\sigma}$. Then $\sigma(g)=u\,g$
with $u=1\in G$. Thus the lemma is proven.
\end{proof}

\noindent As a consequence we can now establish the characterization theorem of \sigmaSE-extensions.

\ExternalProof{(Theorem~\ref{Thm:RPSCharacterization}.(1)\label{Thm:SigmaChar})}{For $G=\{1\}$ we have that $\sconst{G}{\AA}{\sigma}=\const{\AA}{\sigma}=\KK$. By assumption $\KK$ is a field and thus $\sconst{G}{\AA}{\sigma}\setminus\{0\}\leq\AA^*$.
Therefore we can apply Lemma~\ref{Lemma:SigmaLemma} and its equivalence $(2)\Leftrightarrow(3)$ establishes Theorem~\ref{Thm:RPSCharacterization}.(1).
}

\noindent In order to rediscover the difference field version from~\cite{Karr:81,Karr:85}, we specialize Lemma~\ref{Lemma:SigmaLemma} to difference fields by exploiting Lemma~\ref{LemmaA}.(3). 

\begin{lemma}\label{Lemma:SigmaFieldLemma}
Let $\dfield{\FF(t)}{\sigma}$ be a unimonomial field extension of $\dfield{\FF}{\sigma}$ with $\sigma(t)=t+\beta$ for some $\beta\in\FF$. Then the following statements are equivalent.
\begin{enumerate}
\item There is a $g\in\FF(t)\setminus\FF$ with $\frac{\sigma(g)}{g}\in\FF$.
\item There is a $g\in\FF$ with $\sigma(g)=g+\beta$.
\item $\const{\FF(t)}{\sigma}\supsetneq\const{\FF}{\sigma}.$
\end{enumerate}
\end{lemma}
\begin{proof}
$(1)\Rightarrow(2)$: Let $g\in\FF(t)\setminus\FF$ with $\sigma(g)/g\in\FF$. Write
$g=\frac{p}{q}$ with $p,q\in\FF[t]^*$ and $\gcd(p,q)=1$. By Lemma~\ref{LemmaA},
$\sigma(p)/p\in\FF$ and $\sigma(q)/q\in\FF$. Since $g\notin\FF$, we have that $p\notin\FF$ or
$q\notin\FF$. Thus there is a $g'\in\FF[t]$ with $\deg(g')\geq1$ and $\deg(\sigma(g')-g')=\deg(0)=-\infty<0\leq\deg(g')-1$. Hence by
Lemma~\ref{LemmaB} there is a $\gamma\in\FF$ with
$\sigma(\gamma)=\gamma+\beta$.\\ 
$(2)\Rightarrow(3)$ follows by
Lemma~\ref{Lemma:SigmaLemma}. $(3)\Rightarrow(1)$ is analogous to the proof of Lemma~\ref{Lemma:SigmaLemma}.
\end{proof}

\noindent Note that the above lemma is contained in Karr's work by combining Theorems~2.1 and~2.3 from~\cite{Karr:85}. As a consequence, we obtain the following result.

\begin{theorem}\label{Thm:SigmaFieldTheorem}
Let $\dfield{\FF(t)}{\sigma}$ be a unimonomial field extension of
$\dfield{\FF}{\sigma}$ with $\sigma(t)=t+\beta$ for some $\beta\in\FF$. Then
this is a \sigmaSE-extension iff there is no $g\in\FF$ with
$\sigma(g)=g+\beta$.
\end{theorem}

\noindent By the equivalence (3) $\Leftrightarrow$ (1) of Lemma~\ref{Lemma:SigmaLemma} we obtain the following result concerning the semi-constants.

\begin{theorem}\label{Thm:SigmaTheoremPart2}
Let $\dfield{\AA}{\sigma}$ be a difference ring and let $G\leq\AA^*$ with 
$\sconst{G}{\AA}{\sigma}\setminus\{0\}\leq \AA^*$. If 
$\dfield{\AA[t]}{\sigma}$ is a \sigmaSE-extension of $\dfield{\AA}{\sigma}$,
then $\sconst{G}{\AA[t]}{\sigma}=\sconst{G}{\AA}{\sigma}$.
\end{theorem}

\noindent Furthermore, if we specialize to $G=\AA[t]^*$ and assume that $\AA$ is reduced, we get Theorem~\ref{Thm:SigmaTheoremStrongPart2}. For its proof given below we use in addition the following lemma.

\begin{lemma}\label{Lemma:sconstGeneral}
Let $\dfield{\AA}{\sigma}$ be a difference ring: $\gsconst{\AA}{\sigma}\setminus\{0\}\leq\AA^*$ iff $\gsconst{\AA}{\sigma}\setminus\{0\}=\AA^*$.
\end{lemma}
\begin{proof}
Suppose that $\gsconst{\AA}{\sigma}\setminus\{0\}\leq\AA^*$. If $a\in\AA^*$, then $\sigma(a)\in\AA^*$. Thus $u:=\frac{\sigma(a)}{a}\in\AA^*$. With $\sigma(a)=u\,a$ it follows that
$a\in\gsconst{\AA}{\sigma}\setminus\{0\}$. Hence $\gsconst{\AA}{\sigma}\setminus\{0\}\supseteq\AA^*$ and with $\gsconst{\AA}{\sigma}\setminus\{0\}\leq\AA^*$ we have 
$\gsconst{\AA}{\sigma}\setminus\{0\}\subseteq\AA^*$. The other implication is immediate.
\end{proof}

\begin{theorem}\label{Thm:SigmaTheoremStrongPart2}
Let $\dfield{\AA[t]}{\sigma}$ be a \sigmaSE-extension of $\dfield{\AA}{\sigma}$
where $\AA$ is reduced and $\gsconst{\AA}{\sigma}\setminus\{0\}\leq\AA^*$. Then
$\gsconst{\AA[t]}{\sigma}\setminus\{0\}=\gsconst{\AA}{\sigma}\setminus\{0\}=\AA^*$.
\end{theorem}
\begin{proof}
By Lemma~\ref{Lemma:sconstGeneral} it follows that $\gsconst{\AA}{\sigma}\setminus\{0\}=\AA^*$.
Since $\AA$ is reduced, $\AA[t]^*=\AA^*$ by Lemma~\ref{Lemma:PolyInv} and thus   
$\gsconst{\AA[t]}{\sigma}=\sconst{\AA[t]^*}{\AA[t]}{\sigma}=\sconst{\AA^*}{\AA[t]}{\sigma}$. 
Now take $G=\AA^*$ and apply Theorem~\ref{Thm:SigmaTheoremPart2}. Hence $\sconst{\AA^*}{\AA[t]}{\sigma}=\sconst{\AA^*}{\AA}{\sigma}=\gsconst{\AA}{\sigma}$. 
\end{proof}

\subsection{\piE-extensions}

Analogously to Lemma~\ref{LemmaB} we obtain by coefficient comparison the following lemma.

\begin{lemma}\label{LemmaC}
Let $\dfield{\AA\ltr{t}}{\sigma}$ be a unimonomial ring extension of $\dfield{\AA}{\sigma}$ with $\alpha=\frac{\sigma(t)}{t}\in\AA^*$; let $u\in\AA$ and $g=\sum_{i=0}^ng_i\,t^i\in\AA\ltr{t}$. If
$\sigma(g)=u\,g,$
then
$\sigma(g_i)=u\,\alpha^{-i}g_i$ for all $i$.
\end{lemma}
%\begin{proof}
%Comparing coefficients in
%$0=\sigma(g)-u\,g=\sum_{i}(\sigma(g_i)\alpha^i-ug_i)t^i$ it follows for all $i$ that %$\sigma(g_i)\alpha^i-u\,g_i=0$ and thus $\sigma(g_i)=u\,\alpha^{-i}g_i$. which %proves the lemma.
%\end{proof}

\noindent Now we are in the position to obtain the characterization theorem of \piE-extensions.

\ExternalProof{(Theorem~\ref{Thm:RPSCharacterization}.(2)\label{Thm:PiChar})}{``$\Leftarrow$'': Let $m\in\ZZ\setminus\{0\}$ and $g\in\AA\setminus\{0\}$ with  $\sigma(g)=\alpha^m\,g$. Since $\sigma(t^m)=\alpha^m\,t^m$, it follows that
$\sigma(g/t^m)=g/t^m$, i.e., $g/t^m\in\const{\AA\ltr{t}}{\sigma}$. Clearly $g/t^m\notin\AA$ which implies that $g/t^m\notin\const{\AA}{\sigma}$.\\
``$\Rightarrow$'': Let $g=\sum_{i}g_it^i\in\AA\ltr{t}\setminus\AA$ such that $\sigma(g)=g$. Thus $g_m\neq0$ for some $m\neq0$. By Lemma~\ref{LemmaC} we have that $\sigma(g_m)=\alpha^{-m}g_m$.\\
Suppose that $t$ is a \piE-monomial, but $\ord(\alpha)=n>0$. Then $\sigma(t^n)=\alpha^n\,t^n=t^n$, which is a contradiction to the first part of the statement.
}

\noindent Requiring in addition that the semi-constants form a group, this result can be sharpened.

\begin{theorem}\label{Thm:PiCharStrong}
Let $\dfield{\AA}{\sigma}$ be a difference ring and let $G\leq\AA^*$ with $
\sconst{G}{\AA}{\sigma}\setminus\{0\}\leq\AA^*$. Let
$\dfield{\AA\ltr{t}}{\sigma}$ be a unimonomial extension of
$\dfield{\AA}{\sigma}$ with $\sigma(t)=\alpha\,t$ for some $\alpha\in G$.
Then this is a \piE-extension iff there are no
$g\in\sconst{G}{\AA}{\sigma}\setminus\{0\}$ and $m>0$ with
$\sigma(g)=\alpha^m\,g$.
\end{theorem}
\begin{proof}
$\Rightarrow$: Suppose that $t$ is not a \piE-monomial. Then we can take $g\in\AA\setminus\{0\}$ and $m\in\ZZ\setminus\{0\}$ with
$\sigma(g)=\alpha^m\,g$. Hence
$g\in\sconst{G}{\AA}{\sigma}\setminus\{0\}\leq\AA^*$. Thus if
$m<0$, we get $\sigma(\tilde{g})=\alpha^{-m}\tilde{g}$ with $\tilde{g}=\frac1g\in\AA^*$. The other direction
is immediate by Theorem~\ref{Thm:RPSCharacterization}.(2).
\end{proof}

\noindent Together with Lemma~\ref{LemmaA} we rediscover Karr's field version; see~\cite[Theorem~2.2]{Karr:85}

\begin{theorem}\label{Thm:PiCharField}
Let $\dfield{\FF(t)}{\sigma}$ be a unimonomial field extension of
$\dfield{\FF}{\sigma}$ with $\alpha=\frac{\sigma(t)}{t}\in\FF^*$.
Then this is a \piE-extension iff
there are no $g\in\FF^*$ and $m>0$ with $\sigma(g)=\alpha^m\,g$.
\end{theorem}

\begin{proof}
 The direction from right to left follows by Theorem~\ref{Thm:RPSCharacterization}.(2) and the
fact that any \piE-field extension is a \piE-ring extension. Now let $g\in\FF(t)\setminus\FF$ with $\sigma(g)=g$. Write $g=p/q$ with $p,q\in\FF(t)$ where $\gcd(p,q)=1$ and $q$ is monic. W.l.o.g.\ suppose that $\deg(q)\geq\deg(p)$ (otherwise take $1/g$ instead of $g$). By Lemma~\ref{LemmaA},
\begin{equation}\label{Equ:pqQuotient}
\sigma(p)/p\in\FF\quad\quad\text{and}\quad\quad\sigma(q)/q\in\FF.
\end{equation}
We consider two cases. First suppose that $p\in\FF^*$ and $q=t^m$ with $m>0$. Then
$\frac{p}{t^m}=g=\sigma(g)=\frac{\sigma(p)}{\alpha^mt^m}$
which implies that $\sigma(p)=\alpha^mp$.
What remains to consider is the case that $p\notin\FF$ or $q\neq t^m$ for some $m>0$. Define
$$a:=\begin{cases}p&\text{if $q=t^m$ for some $m>0$,}\\
q&\text{otherwise}.
\end{cases}$$
The following holds.
\begin{enumerate}
\item $a\in\FF[t]\setminus\FF$: If $a=q$, note that $q\notin\FF$ by
$\deg(p)\leq\deg(q)$ and $p/q\notin\FF$; if $a=p$, $q=t^m$ and hence
$p\notin\FF$ by assumption.
\item $u:=\sigma(a)/a\in\FF^*$ by~\eqref{Equ:pqQuotient}.
\item $a\neq ut^m$ for all $u\in\FF^*$ and $m>0$: $a$ could be only of this
form, if $q=t^m$ for some $m>0$. Hence $a=p$. But since $\gcd(p,q)=1$, $t\nmid p$.
\end{enumerate}
By the properties~(1) and~(3), it follows that $a=\sum_{i=k}^na_it^i$ with $a_k\neq0\neq a_n$ where $n>k\geq0$. Property~(2) and Lemma~\ref{LemmaC} yield
$\sigma(a_k)=\frac{u}{\alpha^k}a_k$ and $\sigma(a_n)=\frac{u}{\alpha^n}a_n$
which implies $\sigma(\frac{a_k}{a_n})=\alpha^{n-k}\frac{a_k}{a_n}$. Since $\frac{a_k}{a_n}\in\FF^*$ and $n-k>0$, the theorem is proven.
\end{proof}

\noindent Finally, we characterize the set of semi-constants for \piE-extensions.

\begin{proposition}\label{Prop:PiPart2Weak}
Let $\dfield{\AA}{\sigma}$ be a difference ring with $G\leq\AA^*$ and
$\sconst{G}{\AA}{\sigma}\setminus\{0\}\leq\AA^*$. Let
$\dfield{\AA\ltr{t}}{\sigma}$ be \piE-extension of $\dfield{\AA}{\sigma}$ with
$\sigma(t)=\alpha\,t$ for some $\alpha\in G$. Then 
$\sconst{G}{\AA\ltr{t}}{\sigma}=\{h\,t^m|h\in\sconst{G}{\AA}{\sigma}\text{ and }m\in\ZZ\}$
and $\sconst{G}{\AA\ltr{t}}{\sigma}\setminus\{0\}\leq\AA\ltr{t}^*$. 
\end{proposition}
\begin{proof}
``$\subseteq$'': Let $g\in\sconst{G}{\AA\ltr{t}}{\sigma}$, i.e., $g=\sum_i
g_it^i\in\AA\ltr{t}$ with $\sigma(g)=u\,g$ for some $u\in G$. By
Lemma~\ref{LemmaC} we get $\sigma(g_i)\alpha^i=u\,g_i$ and thus $\sigma(g_i)=\frac{u}{\alpha^i}\,g_i$ 
Now suppose that there
are $r,s\in\ZZ$ with $s>r$ and $g_r\neq 0\neq g_s$.
As
$\frac{u}{\alpha^s}\in G$, it follows that
$g_s\in\sconst{G}{\AA}{\sigma}\setminus\{0\}\leq\AA^*$. Thus we conclude that
$\sigma(\frac{g_r}{g_s})=\alpha^{s-r}\frac{g_r}{g_s}$
with $s-r>0$; 
a contradiction to Theorem~\ref{Thm:RPSCharacterization}.(2). Hence $g=h\,t^m$ for some $h\in\sconst{G}{\AA}{\sigma}$, $m\in\ZZ$.\\
``$\supseteq$'': Let $g=h\,t^m$ with $h\in\sconst{G}{\AA}{\sigma}$, $m\in\ZZ$. Then there is a $u\in G$ with $\sigma(h)=u\, h$. Hence
$\sigma(g)=\sigma(h)\,\alpha^m t^m=u\,\alpha^m h\,t^m=u\,\alpha^m\,g$ with $u\,\alpha^m\in G$. Thus $g\in\sconst{G}{\AA\ltr{t}}{\sigma}$.\\
Summarizing, we proved equality which implies that
$\sconst{G}{\AA\ltr{t}}{\sigma}\setminus\{0\}\leq\AA\ltr{t}^*$.
\end{proof}

\noindent So far we obtained a description of the semi-constants for a subgroup $G$ of $\AA^*$. Now we will lift this result to the group 
$$\tilde{G}=\dgroup{G}{\AA\lr{t}}{\AA}=\{h\,t^m|\,h\in G\text{ and }m\in\ZZ\}\leq\AA\lr{t}^*\}.$$

\begin{theorem}\label{Thm:PiPart2Strong}
Let $\dfield{\AA}{\sigma}$ be a difference ring and let $G\leq\AA^*$ with 
$\sconst{G}{\AA}{\sigma}\setminus\{0\}\leq\AA^*$. Let
$\dfield{\AA\ltr{t}}{\sigma}$ be \piE-extension of $\dfield{\AA}{\sigma}$ with
$\sigma(t)=\alpha\,t$ for some $\alpha\in G$ and let $\tilde{G}=\dgroup{G}{\AA\lr{t}}{\AA}$.
Then
$\sconst{\tilde{G}}{\AA\ltr{t}}{\sigma}=\sconst{G}{\AA\ltr{t}}{\sigma}=\{h\,t^m|h\in\sconst{G}{\AA}{\sigma}\text{ and }m\in\ZZ\}$.
\end{theorem}
\begin{proof}
We show that
$\sconst{\tilde{G}}{\AA\ltr{t}}{\sigma}=\sconst{G}{\AA\ltr{t}}{\sigma}$. Then by
Proposition~\ref{Prop:PiPart2Weak} the theorem is proven. Since $G\leq \tilde{G}$,
the inclusion
$\sconst{\tilde{G}}{\AA\ltr{t}}{\sigma}\supseteq\sconst{G}{\AA\ltr{t}}{\sigma}$
is immediate. Now suppose that $g=\sum_i
g_i\,t^i\in\sconst{\tilde{G}}{\AA\ltr{t}}{\sigma}$. Hence there are an $m\in\ZZ$ and
an $h\in G$ with $\sigma(g)=h\,t^mg$. By coefficient comparison it follows
that
$\sigma(g_i)\alpha^i=h g_{i-m}.$
If $m\geq1$, take $s$ minimal such that $g_s\neq0$. Then $\sigma(g_s)\alpha^s\neq0$. But by the choice of $s$, we get $g_{s-m}=0$ and thus $h\,g_{s-m}=0$, a contradiction. Otherwise, if $m<0$, take $s$ maximal such that $g_{s-m}\neq0$. Then $h\,g_{s-m}\neq0$. But by the choice of $s$, we get $\sigma(g_s)\alpha^s=0$, again a contradiction. Thus $m=0$ and consequently, $g\in\sconst{G}{\AA\ltr{t}}{\sigma}$.
\end{proof}

\noindent We close this subsection with Theorem~\ref{Thm:PiEfullsconst}. It provides a description of $\gsconst{\AA\ltr{t}}{\sigma}$ under the assumption that $\AA$ is reduced and connected. This result is not applicable if general \rE-extensions pop up; see Example~\ref{Exp:NotConnected}. But, it will be used for further insights summarized in Corollary~\ref{Cor:pisiSCONST}.(2), Proposition~\ref{Prop:SimpleIsKarrVersion} and Corollary~\ref{Cor:SingleRootedIsSimple} below.

\begin{theorem}\label{Thm:PiEfullsconst}
Let $\dfield{\AA}{\sigma}$ be a difference ring being reduced and connected with  $\gsconst{\AA}{\sigma}\setminus\{0\}=\AA^*$. Let $\dfield{\AA\ltr{t}}{\sigma}$ be \piE-extension of $\dfield{\AA}{\sigma}$ with $\sigma(t)=\alpha\,t$ for some $\alpha\in\AA^*$.  Then
$\gsconst{\AA\ltr{t}}{\sigma}=\{h\,t^m|\,h\in\gsconst{\AA}{\sigma}\text{ and }m\in\ZZ\}.$
\end{theorem}
\begin{proof}
Take $\tilde{G}=\dgroup{(\AA^*)}{\AA\lr{t}}{\AA}$. Then
$\tilde{G}=\AA\ltr{t}^*$ by Lemma~\ref{Lemma:LaurentInv}. Thus
$\gsconst{\AA\ltr{t}}{\sigma}=\sconst{\AA[t,1/t]^*}{\AA\ltr{t}}{\sigma}=\sconst
{\tilde{G}}{\AA\ltr{t}}{\sigma}\stackrel{\text{Thm.~\ref{Thm:PiPart2Strong}}}
=\{h\,t^m|\,h\in\gsconst{\AA}{\sigma}\text{ and }m\in\ZZ\}.$
\end{proof}

\subsection{\rE-extensions}

We start with the proof of the characterization theorem of \rE-extensions.

\ExternalProof{(Theorem~\ref{Thm:RPSCharacterization}.(3)\label{Thm:RChar})}{``$\Leftarrow$'': Let $m\in\{1,\dots,\lambda-1\}$ and $g\in\AA\setminus\{0\}$ with  $\sigma(g)=\alpha^m\,g$. Since $\sigma(t^m)=\alpha^m\,t^m$, it follows that
$\sigma(g\,t^{\lambda-m})=g\,t^{\lambda-m}$, i.e., $g\,t^{\lambda-m}\in\const{\AA[t]}{\sigma}$. Clearly $g\,t^{\lambda-m}\notin\AA$ which implies that $g\,t^{\lambda-m}\notin\const{\AA}{\sigma}$.\\
``$\Rightarrow$'': Let $g=\sum_{i=0}^{\lambda-1}g_i\,t^i\in\AA[t]\setminus\AA$ with $\sigma(g)=g$. Thus $g_r\neq0$ for some $r\in\{1,\dots,\lambda-1\}$. By coefficient comparison we get $\sigma(g_r)=\alpha^{\lambda-r}g_r$ with $\lambda-r\in\{1,\dots,\lambda-1\}$.\\
Let $t$ be an \rE-monomial and let $m:=\ord(\alpha)<\lambda$. Then with $g=1\in\AA\setminus\{0\}$ we have that
$\sigma(g)=1=\alpha^m\,1=\alpha^m\,g$. A contradiction to the first statement.
}

\noindent Finally, we work out properties for the set of semi-constants. Since the proof of the following theorem is completely analogous to the proof of Proposition~\ref{Prop:PiPart2Weak}, it is skipped.

\begin{proposition}\label{Prop:RPart2Weak}
Let $\dfield{\AA}{\sigma}$ be a difference ring with $G\leq\AA^*$ and 
$\sconst{G}{\AA}{\sigma}\setminus\{0\}\leq\AA^*$. Let $\dfield{\AA[x]}{\sigma}$
be an \rE-extension of $\dfield{\AA}{\sigma}$ with $\alpha=\frac{\sigma(x)}{x}\in G$ and $\lambda:=\ord(x)=\ord(\alpha)>1$. Then
$\sconst{G}{\AA[x]}{\sigma}=\{h\,x^m|h\in\sconst{G}{\AA}{\sigma},\, 0\leq m<\lambda\}$
and $\sconst{G}{\AA[x]}{\sigma}\setminus\{0\}\leq\AA[x]^*$.
\end{proposition}

As in Theorem~\ref{Thm:PiPart2Strong} we will lift this result from the group $G\leq\AA^*$ to 
$$\tilde{G}=\dgroup{G}{\AA[x]}{\AA}=\{h\,x^m|\,h\in G\text{ and }m\in\{0,\dots,\lambda-1\}\}\leq\AA[x]^*\}.$$
We remark that there is the following subtlety. We have to assume that $\AA[x]$ is reduced in order to prove the result below. In order to take care of this extra property, further investigations will be necessary in Subsection~\ref{Subsec:NestedRExt}.

\begin{theorem}\label{Thm:RPart2Strong}
Let
$\dfield{\AA[x]}{\sigma}$ be an \rE-extension of $\dfield{\AA}{\sigma}$ and 
let $G\leq\AA^*$
with $\sconst{G}{\AA}{\sigma}\setminus\{0\}\leq\AA^*$. If $\AA[x]$ is reduced,
then $\sconst{\tilde{G}}{\AA[x]}{\sigma}\setminus\{0\}\leq\AA[x]^*$ for $\tilde{G}=\dgroup{G}{\AA[x]}{\AA}.$
\end{theorem}
\begin{proof}
Let $\alpha:=\frac{\sigma(x)}{x}\in\AA^*$ and $n=\ord(\alpha)=\ord(x)$.
Let $g\in\sconst{\tilde{G}}{\AA[x]}{\sigma}\setminus\{0\}$, i.e.,
$\sigma(g)=u\,x^m\,g$ with $u\in G$ and $0\leq m<n$. Since $x^{m\,n}=1$, 
$\sigma(g^n)=u^n\,g^n$
with $u^n\in G$.\\
First suppose that $v:=g^n\in\AA$. Since $\AA[x]$ is reduced, 
$v\neq0$ and thus
$v\in\sconst{G}{\AA}{\sigma}\setminus\{0\}\leq\AA^*$, i.e., $g\,(g^{n-1}/v)=1$.
Hence $g$ is invertible, i.e., $g\in\AA[x]^*$.\\
Otherwise, suppose that $v:=g^n\notin\AA$. Define $a:=u^n\in G$. We consider
two sub-cases. Suppose that there are a $k>0$ and a $w\in\AA\setminus\{0\}$ with
$\sigma(w)=a^k\,w$. Then $w\in\sconst{G}{\AA}{\sigma}\setminus\{0\}\leq\AA^*$.
Hence
$\sigma((g^n)^k/w)=(g^n)^k/w$, i.e, $c:=(g^n)^k/w\in\KK$, and since $\AA[x]$ is reduced, $c\neq0$. Thus (as above) $g\,(g^{k\,n-1}/w/c)=1$ and therefore $g\in\KK[x]^*$.\\
Finally, suppose that there are no $k>0$ and $w\in\AA\setminus\{0\}$ with
$\sigma(w)=a^k\,g$. Hence by Theorem~\ref{Thm:PiCharStrong} there is the
\piE-extension $\dfield{\AA\ltr{t}}{\sigma}$ of $\dfield{\AA}{\sigma}$ with
$\sigma(t)=a\,t$ ($a\in G\leq\AA^*$). Let
$v=g^n=\sum_{i=0}^{n-1}v_i\,x^i\in\AA[x]\setminus\AA$. Then $\sigma(v)=a\,v$ and
thus by coefficient comparison it follows that
$\sigma(v_i)=a\,\alpha^{n-i}\,v_i$ for some $v_i\in\AA\setminus\{0\}$ with
$1\leq i<n$. Hence $\sigma(\frac{v_i}{t})=\alpha^{n-i}\frac{v_i}{t}$, and thus
$\sigma(\frac{v_i^n}{t^n})=\frac{v_i^n}{t^n}$. Since $\AA$ is reduced, we have $v_i^n\neq0$, and consequently $\frac{v_i^n}{t^n}\in\const{\AA\ltr{t}}{\sigma}\setminus\AA$, a contradiction that $t$ is a \piE-monomial. Thus this case can be excluded. Summarizing, any element in $\sconst{\tilde{G}}{\AA[x]}{\sigma}\setminus\{0\}$ is from $\AA[x]^*$. 
\end{proof}

\section{Nested \rpisiSE-extensions and simple \rpisiSE-extensions}\label{Sec:NestedExt}

We explore the set of semi-constants. First we deal with nested \rE-extensions in Subsection~\ref{Subsec:NestedRExt} and with nested \pisiSE-extensions in Subsection~\ref{Subsec:NestedPiSiExt}. Finally, we obtain Theorems~\ref{Thm:sconstIsGroupRing} and~\ref{Thm:sconstIsGroupField} for nested \rpisiSE-extensions in Subsection~\ref{Subsec:NestedRpisiExt}. In addition, we work out further structural properties of (simple) \rpisiSE-extensions. 

\subsection{Nested \rE-extensions}\label{Subsec:NestedRExt}

We derive a first result of the semi-constants by applying iteratively Proposition~\ref{Prop:RPart2Weak}.

\begin{proposition}\label{Prop:RExtSConstSpecial}
Let $\dfield{\AA}{\sigma}$ be a difference ring with $G\leq\AA^*$ and $\sconst{G}{\AA}{\sigma}\setminus\{0\}\leq\AA^*$.
Let $\dfield{\EE}{\sigma}$ with $\EE=\AA\lr{x_1}\dots\lr{x_e}$ be an \rE-extension of $\dfield{\AA}{\sigma}$ with $\frac{\sigma(x_i)}{x_i}\in G$ and $n_i=\ord(x_i)$. Then 
$\sconst{G}{\EE}{\sigma}=\{h\,x_1^{m_1}\dots x_1^{m_e}|\,h\in\sconst{G}{\AA}{\sigma}\text{ and }0\leq m_i<n_i\text{ for }1\leq i\leq e\}$
and $\sconst{G}{\EE}{\sigma}\setminus\{0\}\leq\EE^*$.
\end{proposition}

In order to treat nested \rE-extensions, we proceed as follows. Let $\dfield{\AA\lr{x_1}\dots\lr{x_e}}{\sigma}$ be an \rE-extension of
$\dfield{\AA}{\sigma}$ with $\lambda_i=\ord(x_i)$ and
$\sigma(x_i)=\alpha_i\,x_i$. Moreover, take the polynomial ring
$R=\AA[y_1,\dots,y_e]$ and define
$\alpha'_i=\alpha|_{x_1\to y_1,\dots,x_{i-1}\to y_{i-1}}.$
Then we obtain the automorphism $\fct{\sigma'}{R}{R}$ by
$\sigma'|_{\AA}=\sigma$ and $\sigma(y_i)=\alpha_i y_i$, i.e.,
$\dfield{R}{\sigma'}$ is a difference ring extension of $\dfield{\AA}{\sigma}$.
Thus by iterative application of the construction used for
Lemma~\ref{Lemma:RExistence} it follows that $\AA\lr{x_1}\dots\lr{x_e}$ is
isomorphic to $R/I$ where $I$ is the ideal 
\begin{equation}\label{Equ:RExtIdeal}
I=\langle y_1^{\lambda_1}-1,\dots,y_e^{\lambda_e}-1\rangle
\end{equation}
in $R$. In particular, we obtain the automorphism $\fct{\sigma''}{R/I}{R/I}$ defined by $\sigma''(f+I)=\sigma'(f)+I$ and it follows that the difference ring
 $\dfield{\AA\lr{x_1}\dots\lr{x_e}}{\sigma}$
is isomorphic to $\dfield{R/I}{\sigma''}$; here $f\in\AA\lr{x_1}\dots\lr{x_e}$ is mapped to $f'+I$ 
with $f'=f|_{x_1\to y_1,\dots,x_e\to y_e}$.

Take $G=\dgroup{(\FF^*)}{\EE}{\FF}\leq\EE^*$. In order to show that $\sconst{G}{\EE}{\sigma}\setminus\{0\}\leq\EE^*$ holds as claimed in Corollary~\ref{Cor:SConstRExtension} below, we use Gr\"obner bases theory.

\begin{lemma}\label{Lemma:RadicalIdeal}
Let $\lambda_i\in\NN\setminus\{0\}$. Then the zero-dimensional ideal $I$ given in~\eqref{Equ:RExtIdeal} in the polynomial ring $R=\FF[y_1,\dots,y_e]$ is radical. 
\end{lemma}
\begin{proof}
The ideal $I$ is zero-dimensional. Since $\FF$ has characteristic $0$, it
is perfect. We therefore apply Seidenberg's criterion (algorithm) given
in~\cite[Thm.~8.22]{BeckerWeispfenning:93}. Define $f_i=y_i^{\lambda_i}-1$. Then
for each $i$ ($1\leq i\leq e$) we have that $f_i\in R\cap\FF[y_i]$ and
$\gcd(f_i,\frac{d}{d y_i}f_i)=\gcd(y_i^{\lambda_i}-1,\lambda_i
y^{\lambda_i-1})=1$. Thus~\cite[Thm.~8.22]{BeckerWeispfenning:93} implies that
$\langle f_1,\dots,f_e\rangle$ is radical. 
\end{proof}

\begin{corollary}\label{Cor:SConstRExtension}
Let $\dfield{\EE}{\sigma}$  
be an \rE-extension of a difference field $\dfield{\FF}{\sigma}$ and let
$G=\dgroup{(\FF^*)}{\EE}{\FF}$.
Then $\EE$ is reduced and $\sconst{G}{\EE}{\sigma}\setminus\{0\}\leq\EE^*$.
\end{corollary}
\begin{proof}
The difference ring $\dfield{\EE}{\sigma}$ with $\EE=\FF\lr{x_1}\dots\lr{x_r}$ is isomorphic to
$\dfield{R/I}{\sigma''}$ as defined above with~\eqref{Equ:RExtIdeal} where
$\AA=\FF$. Suppose that $\EE$ is not reduced. Then there are an
$f\in\EE\setminus\{0\}$ and an $n>0$ with $f^n=0$. Hence there is an $h\in R$
with $h+I\neq I$ and $(h+I)^n=h^n+I=I$. This implies that $h\notin I$ and
$h^n\in I$. Therefore $I$ is not radical, a contradiction to
Lemma~\ref{Lemma:RadicalIdeal}. Hence $\EE$ is reduced. Thus we can apply
Theorem~\ref{Thm:RPart2Strong} iteratively and it follows that
$\sconst{G}{\EE}{\sigma}\setminus\{0\}\leq\EE^*$.
\end{proof}

\subsection{Nested \pisiSE-extensions}\label{Subsec:NestedPiSiExt}

In Corollary~\ref{Cor:pisiSCONST} we will characterize the set of semi-constants within \pisiSE-extensions. Part~1 will deal with the general case. Part~2 assumes in addition that the ground ring is reduced and connected. In this setting, we rely on the following two lemmas.

\begin{lemma}\label{Lemma:LiftReducedConnected}
Let $\dfield{\AA\lr{t}}{\sigma}$ be a \pisiSE-extension of
$\dfield{\AA}{\sigma}$. If $\AA$ is reduced, $\AA\lr{t}$ is reduced. If $\AA$ is
reduced and connected, $\AA\lr{t}$ is reduced and connected.
\end{lemma}
\begin{proof}
Let $t$ be a \piE-monomial. Moreover, let $\AA$ be reduced. Now take $f=\sum_i
f_it^i\in\AA\lr{t}=\AA\ltr{t}$ with $f\neq0$ and $f^n=0$ for some $n>0$. Since
$\AA$ is reduced, $f\notin\AA$. Let $m\in\ZZ$ be maximal such that $f_m\neq0$.
Then the coefficient of $t^{n\,m}$ in $f^n$ is $f_m^n$. Hence $f_m^n=0$ and thus
$f_m$ is a nilpotent element in $\AA$, a contradiction.\\
Now let $\AA$ be reduced and connected and take $f=\sum_i
f_it^i\in\AA\lr{t}=\AA\ltr{t}$ with $f^2=f$ and $f\notin\{0,1\}$. Since $\AA$ is
connected, $f\notin\AA$. Let $m$ be maximal such that $f_m\neq0$. If $m>0$, then
the coefficient of $t^{2m}$ in $f^2$ is $f_m^2$ and thus with $f^2=f$ we have
that $f_m^2=0$; a contradiction that $\AA$ is reduced. Otherwise, if $m=0$, we
take $\bar{m}$ minimal with $f_{\bar{m}}\neq0$. Note that $\bar{m}<0$ since $f\notin\AA$. As
above, it follows that $f_{\bar{m}}^2=0$, again a contradiction.
Summarizing, if $\AA$ is reduced (and connected), $\AA\ltr{t}$ is
reduced (and connected). 
For a \sigmaSE-monomial $t$, the same implications hold since
$\AA\lr{t}=\AA[t]\leq\AA\ltr{t}$.
\end{proof}

\noindent If $\AA$ is reduced, the shift behaviour of \piE-monomials does not depend on \sigmaSE-monomials.

\begin{lemma}\label{Lemma:OrderPiSi}
Let $\dfield{\EE}{\sigma}$ be a \pisiSE-ring extension of 
$\dfield{\AA}{\sigma}$ where $\AA$ is reduced. Then the generators can be
reordered such that we get the form
$\EE=\AA\lr{t_1}\dots\lr{t_p}\lr{s_1}\dots\lr{s_e}$ where
the $t_i$ are \piE-monomials and the $s_i$ are
\sigmaSE-monomials.
\end{lemma}
\begin{proof}
Let $\EE=\AA\lr{t_1}\dots\lr{t_e}$. By iterative application of
Lemma~\ref{Lemma:LiftReducedConnected} it follows that $\EE$ is reduced. Let
$t_i$ be a \piE-monomial where
$\alpha=\sigma(t_i)/t_i\in\AA\lr{t_1}\dots\lr{t_{i-1}}$ depends on a
\sigmaSE-monomial $t_j$ with $j<i$. 
Then we can reorder the generators such that
we get
$\HH=\AA\lr{t_1}\dots\lr{t_{j-1}}\lr{t_{j+1}}\dots\lr{t_{i-1}}$; here we forget $\sigma$ and argue purely in the given ring. In particular,
$\alpha\in\HH\lr{t_j}=\HH[t_j]\setminus\HH$. Since $\alpha$ is invertible,
$\alpha\in\HH$ by Lemma~\ref{Lemma:PolyInv}; a contradiction. Summarizing, for
all \piE-monomials $t_j$ we have that $\sigma(t_j)/t_j$ is free of
\sigmaSE-monomials. Thus we can shuffle all \piE-monomials to the left and all
\sigmaSE-monomials to the right and obtain again a \pisiSE-extension.
\end{proof}

\begin{corollary}\label{Cor:pisiSCONST}
Let $\dfield{\EE}{\sigma}$ be a \pisiSE-extension of $\dfield{\AA}{\sigma}$ with
$\EE=\AA\lr{t_1}\lr{t_2}\dots\lr{t_e}$.
\begin{enumerate}
\item Let $G\leq\AA^*$ with $\sconst{G}{\AA}{\sigma}\setminus\{0\}\leq
\AA^*$ and let $\tilde{G}=\dgroup{G}{\EE}{\AA}$.
If $\dfield{\EE}{\sigma}$ is a $G$-simple \pisiSE-extension of $\dfield{\AA}{\sigma}$, then $\sconst{\tilde{G}}{\EE}{\sigma}\setminus\{0\}\leq\EE^*$
where

\vspace*{-0.7cm}

\begin{multline*}
\sconst{\tilde{G}}{\EE}{\sigma}=\{h\,t_1^{m_1}\dots t_e^{m_e}|\,h\in\sconst{G}{\AA}{\sigma}\text{ and}\\[-0.09cm]
m_i\in\ZZ\text{ where $m_i=0$ if $t_i$ is a \sigmaSE-monomial}\}.
\end{multline*}

\vspace*{-0.1cm}

\item If $\AA$ is reduced and connected and $\gsconst{\AA}{\sigma}\setminus\{0\}=\AA^*$, then 

\vspace*{-0.5cm}

\begin{equation}\label{Equ:GSConstGeneralNested}
\begin{split}
\gsconst{\EE}{\sigma}\setminus\{0\}=\{h\,t_1^{m_1}\dots t_e^{m_e}|\,&h\in\AA^*\text{  and }m_i\in\ZZ\\[-0.09cm]
&\text{ where $m_i=0$ if $t_i$ is a \sigmaSE-monomial}
\}=\EE^*.
\end{split}
\end{equation}

\vspace*{-0.1cm}

\item If $\AA$ is a field then we have that~\eqref{Equ:GSConstGeneralNested}. 
\end{enumerate}
\end{corollary}

\begin{proof}
The first part is proven by induction on the number $e$ of extensions. If
$e=0$, nothing has to be shown. Now suppose that the first part holds and
consider one extra $\tilde{G}$-simple \pisiSE-monomial $t_{e+1}$ on top. Define
$\tilde{\tilde{G}}=\dgroup{\tilde{G}}{\EE\lr{t_{e+1}}}{\EE}=\dgroup{G}{\EE\lr{t_
{e+1}
}}{\AA}$. If $t_i$ is a \sigmaSE-monomial, $\tilde{\tilde{G}}=\tilde{G}$. Together with
Theorem~\ref{Thm:SigmaTheoremPart2} it follows that
$\sconst{\tilde{\tilde{G}}}{\EE[t_{e+1}]}{\sigma}=\sconst{\tilde{G}}{\EE[t_{e+1}
]}{\sigma}=\sconst{\tilde{G}}{\EE}{\sigma}$ and
$\sconst{\tilde{\tilde{G}}}{\EE[t_{e+1}]}{\sigma}\setminus\{0\}\leq\EE^*\leq\EE[
t_{e+1}]^*$. If $t_i$ is a \piE-monomial, we have
$\sigma(t_{e+1})/t_{e+1}\in\tilde{G}$. 
Hence Theorem~\ref{Thm:PiPart2Strong} yields  
$\sconst{\tilde{\tilde{G}}}{\EE\ltr{t_{e+1}}}{\sigma}=\{h\,t_{e+1}^m|\,
m\in\ZZ\text{ and }h\in\sconst{\tilde{G}}{\EE}{\sigma}\}$ and thus by the
induction assumption we have that
\begin{align*}
\sconst{\tilde{\tilde{G}}}{\EE\ltr{t_{e+1}}}{\sigma}=\{h\,t_1^{m_1}\dots t_{e+1}^{m_{e+1}}|\,&h\in\sconst{G}{\AA}{\sigma}\text{ and }m_i\in\ZZ\\
&\text{ where $m_i=0$ if $t_i$ is a \sigmaSE-monomial}\}
\end{align*}
and thus
$\sconst{\tilde{\tilde{G}}}{\EE\ltr{t_{e+1}}}{\sigma}\setminus\{0\}\leq\EE\ltr{
t_{e+1}}^*$. This completes the induction step.\\
Similarly, the first equality of part~2 follows by Theorems~\ref{Thm:SigmaTheoremStrongPart2}
and~\ref{Thm:PiEfullsconst}. The second equality follows by Lemmas~\ref{Lemma:PolyInv} and~\ref{Lemma:LiftReducedConnected}.
Since any field is connected and reduced and $\gsconst{\AA}{\sigma}\setminus\{0\}=\AA^*$ by Lemma~\ref{Lemma:sconstGeneral}, part~3 follows by part~2. 
\end{proof}

\noindent Restricting to \sigmaSE-extensions, the above result simplifies as follows.

\begin{corollary}\label{Cor:SemiConstantsForSigmaExt}
Let  
$\dfield{\EE}{\sigma}$ be a \sigmaSE-extension of $\dfield{\AA}{\sigma}$. Then the following holds.
\begin{enumerate}
\item If $G\leq\AA^*$ with $\sconst{G}{\AA}{\sigma}\setminus\{0\}\leq \AA^*$, then $\sconst{G}{\EE}{\sigma}=\sconst{G}{\AA}{\sigma}$.
\item If $\AA$ is reduced and $\gsconst{\AA}{\sigma}\setminus\{0\}=\AA^*$, then
$\EE$ is reduced and $\gsconst{\EE}{\sigma}=\gsconst{\AA}{\sigma}$.
\item If $\AA$ is a field, then
$\gsconst{\EE}{\sigma}\setminus\{0\}=\AA^*=\AA\setminus\{0\}$.
\end{enumerate}
\end{corollary}

\subsection{\rpisiSE-extensions and their simple and single-rooted restrictions}\label{Subsec:NestedRpisiExt}

We turn to the set of semi-constants within nested \rpisiSE-extensions. The case of simple and single-rooted \rpisiSE-extensions is immediate.

\ExternalProof{(Theorem~\ref{Thm:sconstIsGroupRing}\label{Proof:sconstIsGroupRing})}{
This follows by Corollary~\ref{Cor:pisiSCONST}.(1) and Proposition~\ref{Prop:RExtSConstSpecial}.
}

\noindent Likewise, simple \rpisiSE-extension can be treated if they are built in a particular form.

\begin{theorem}\label{Thm:rpisiSCONST}
Let $\dfield{\HH}{\sigma}$ be an \rE-extension of
a difference field $\dfield{\FF}{\sigma}$ and
let $\dfield{\EE}{\sigma}$ with
$\EE=\HH\lr{t_1}\lr{t_2}\dots\lr{t_e}$ be a simple \pisiSE-extension of
$\dfield{\HH}{\sigma}$. Let $G=\dgroup{(\FF^*)}{\HH}{\FF}$
and define $\tilde{G}=\dgroup{G}{\EE}{\HH}$. Then we have $\sconst{\tilde{G}}{\EE}{\sigma}\setminus\{0\}\leq\EE^*$ where
\begin{equation*}
\sconst{\tilde{G}}{\EE}{\sigma}=\{h\,t_1^{m_1}\dots t_e^{m_e}|\,h\in
\sconst{G}{\HH}{\sigma},\,m_i\in\ZZ\text{ where $m_i=0$ if $t_i$ is a \sigmaSE-monomial}\}.
\end{equation*}
\end{theorem}

\begin{proof}
By Cor.~\ref{Cor:SConstRExtension}, $\sconst{G}{\HH}{\sigma}\setminus\{0\}\leq\HH^*$. Hence the result follows by
Cor.~\ref{Cor:pisiSCONST}.(1).
\end{proof}

\noindent Next, we show that simple \rpisiSE-extensions can be always brought to the shape as assumed in Theorem~\ref{Thm:rpisiSCONST}. This will finally produce a proof of Theorem~\ref{Thm:sconstIsGroupField}. 

\begin{lemma}\label{Lemma:ReorderSimpleRPISI}
Let $\dfield{\AA}{\sigma}$ be a difference ring with a group $G\leq\AA^*$ and let $\dfield{\EE}{\sigma}$ be a $G$-simple \rpisiSE-extension of $\dfield{\AA}{\sigma}$.
\begin{enumerate}
\item The \rpisiSE-monomials can be reordered to the form
$\EE=\AA\lr{t_1}\lr{t_2}\dots\lr{t_e}$ with $r,p\in\NN$ ($0\leq r\leq p\leq
e$) such that the following holds.
\begin{itemize}
\item For all $i$ $(1\leq i\leq r)$, $t_i$ is an \rE-monomial with
$\sigma(t_i)/t_i=u_i\,t_1^{z_1}\dots t_{i-1}^{z_{i-1}}$
for some root of unity $u_i\in G$ and $z_i\in\NN$.
\item For all $i$ $(r<i\leq p)$, $t_i$ is a \piE-monomial with
$\sigma(t_i)/t_i=u_i\,t_1^{z_1}\dots t_{i-1}^{z_{i-1}}$
for some $u_i\in G$ and $z_i\in\ZZ$.
\item For all $i$ $(p<i\leq e)$, $t_i$ is a \sigmaSE-monomial with
$\sigma(t_i)-t_i\in\AA\lr{t_1}\lr{t_2}\dots\lr{t_{i-1}}.$
\end{itemize}
\item For any $f\in\dgroup{G}{\EE}{\AA}$ which depends on a \piE-monomial we have that $\ord(f)=0$.
\end{enumerate}
\end{lemma}

\begin{proof}
We show the lemma by induction on the number of \rpisiSE-monomials. Suppose that the lemma holds for $e$ extensions. Now let $\EE=\AA\lr{t_1}\dots\lr{t_{e}}$ and consider the \rpisiSE-monomial $t_{e+1}$ on top of $\EE$. 
By the induction assumption we can reorder $\EE$ such that it has the desired form (all \rE-monomials are on the left, all \piE-monomials are in the middle and all \sigmaSE-monomials are on the right).
If $t_{e+1}$ is a \sigmaSE-monomial, the required shape is fulfilled. 
If $t_{e+1}$ is an \rE-monomial, observe that $\alpha:=\sigma(t_{e+1})/t_{e+1}\in\dgroup{G}{\EE}{\AA}$. Since $\ord(\alpha)=\ord(t_{e+1})>1$ by Theorem~\ref{Thm:RPSCharacterization}.(3), $\alpha$ is free of \piE-monomials by the induction assumption and (by definition) free of \sigmaSE-monomials. Thus we can shuffle $t_{e+1}$ to the left (such that all \pisiSE-monomials are to the right), and the required shape is satisfied. Similarly, if $t_{e+1}$ is a \piE-monomial, $\sigma(t_{e+1})/t_{e+1}\in\dgroup{G}{\EE}{\AA}$ is free of \sigmaSE-monomials by definition and we can shuffle $t_{e+1}$ to the left such that all \sigmaSE-monomials are to the right. This completes the first part of the lemma. Now let $\EE=\AA\lr{x_1}\dots\lr{x_{e+1}}$ be in the desired ordered form. If $x_{e+1}$ is a \sigmaSE-monomial, we have that $\dgroup{G}{\EE\lr{x_{e+1}}}{\AA}=\dgroup{G}{\EE}{\AA}$. Thus the second part holds by the induction assumption. If $x_{e+1}$ is an \rE-extension, also all $x_i$ with $1\leq i\leq e$ are \rE-monomials, and the second statement holds trivially. Finally, let $x_{e+1}$ be a \piE-monomial and take $f\in\dgroup{G}{\EE\lr{x_{e+1}}}{\AA}$. If $f\in\EE$ and $f$ depends on \piE-monomials, we have again that $\ord(f)=0$ by the induction assumption. To this end, suppose that $f$ depends on $x_{e+1}$ and we have that $\ord(f)=n>0$. Then $f=u\,x_{e+1}^{m}$ where $m\neq0$ and $u\in\EE^*$. Since $f^n=1$, $u^n\,x_{e+1}^{m\,n}=1$ where
$u^n\neq0$. Hence $x_{e+1}$ is not transcendental over $\EE$, a contradiction to
the definition of a \piE-monomial. Thus $\ord(f)=n=0$. This completes the proof.
\end{proof}

\ExternalProof{(Theorem~\ref{Thm:sconstIsGroupField}\label{Proof:sconstIsGroupField})}{By Lemma~\ref{Lemma:ReorderSimpleRPISI} we can reorder the simple \rpisiSE-extension such that Theorem~\ref{Thm:rpisiSCONST} is applicable.
}

In the remaining part of this section we deliver insight into the structure of (simple) \rpisiSE-extensions. First observe that a tower of simple \rpisiSE-extensions is again simple.

\begin{lemma}\label{Lemma:SimpleTowerIsSimple}
Let $\dfield{\AA}{\sigma}$ be a difference ring with a group $G\leq\AA^*$ and
let $\dfield{\AA}{\sigma}\leq\dfield{\HH}{\sigma}\leq\dfield{\EE}{\sigma}$ be \rpisiSE-extensions. Then $\dgroup{(\dgroup{G}{\HH}{\AA})}{\EE}{\HH}=\dgroup{G}{\EE}{
\AA}$. Moreover, if $\dfield{\AA}{\sigma}\leq\dfield{\HH}{\sigma}$ is $G$-simple and $\dfield{\HH}{\sigma}\leq\dfield{\EE}{\sigma}$ is $\dgroup{G}{\HH}{\AA}$-simple, then $\dfield{\AA}{\sigma}\leq\dfield{\EE}{\sigma}$ is $G$-simple. 
\end{lemma}

\noindent Further, the reordering as described in Lemma~\ref{Lemma:ReorderSimpleRPISI} is also possible if one relaxes the condition that the \rpisiSE-extension is simple but requires that the ground ring is a field.

\begin{lemma}\label{Lemma:ShuffleOverField}
Let $\dfield{\EE}{\sigma}$ be a \rpisiSE-ring extension of a difference field
$\dfield{\FF}{\sigma}$. Then $\dfield{\EE}{\sigma}$ can be reordered to the form
$\EE=\FF\lr{x_1}\dots\lr{x_r}\lr{t_1}\dots\lr{t_p}\lr{s_1}\dots\lr{s_e}$ where
the $x_i$ are \rE-monomials, the $t_i$ are \piE-monomials and the $s_i$ are
\sigmaSE-monomials.
\end{lemma}
\begin{proof}
First we try to shuffle all \rE-extensions to the front. Suppose that this fails at
the first time. Then there are an \rE-extension $\dfield{\HH}{\sigma}$ of
$\dfield{\FF}{\sigma}$, a \pisiSE-extension $\dfield{\GG}{\sigma}$ of
$\dfield{\HH}{\sigma}$ with $\GG=\HH\lr{y_1}\dots\lr{y_l}$ and an \rE-extension
$\dfield{\GG\lr{x}}{\sigma}$ of $\dfield{\GG}{\sigma}$ with $\alpha=\sigma(x)/x$
in
which $y_l$ occurs. Note that $\HH$ is reduced by
Corollary~\ref{Cor:SConstRExtension}, and $\GG$ is reduced by iterative
application of
Lemma~\ref{Lemma:LiftReducedConnected}. Write $\alpha=\sum_i f_i y_l^i$. Let $m\neq0$ such that $f_m\neq0$ and such that $|m|\geq1$ is maximal (we
remark
that $m<0$ can only happen if $y_l$ is a \piE-monomial).
By the choice of $m$, we have that the coefficient of $y_l^{m\,n}$ in
$\alpha^n$ is $f_m^n$. Hence with $\alpha^n=1$ it follows that $f_m^n=0$, a
contradiction to the assumption that $\GG$ is reduced. Therefore we can shuffle all \rE-monomials to the left and all \pisiSE-monomials to the right. Since the
nested \rE-extension is reduced by Corollary~\ref{Cor:SConstRExtension}, we can apply Lemma~\ref{Lemma:OrderPiSi} to
reorder the \pisiSE-monomials further as claimed in the statement.
\end{proof}

\noindent By definition any (nested) \sigmaSE-extension is a also a simple
\sigmaSE-extension. If the ground ring is reduced
and connected, we obtain the following stronger result.

\begin{proposition}\label{Prop:SimpleIsKarrVersion}
Let $\dfield{\GG}{\sigma}$ be a difference ring  where $\GG$ is reduced and connected and where $\gsconst{\GG}{\sigma}\setminus\{0\}=\GG^*$. Then a \pisiSE-extension $\dfield{\EE}{\sigma}$ of $\dfield{\GG}{\sigma}$ is simple.
\end{proposition}
\begin{proof}
Let $\EE=\GG\lr{t_1}\dots\lr{t_e}$. By Lemma~\ref{Lemma:OrderPiSi} we may suppose that the
generators are ordered such that the $t_1\dots,t_p$ are \piE-monomials and
the $t_{p+1}\dots,t_e$ are \sigmaSE-monomials. 
By Corollary~\ref{Cor:pisiSCONST}.(2)  we have that $\frac{\sigma(t_i)}{t_i}\in\GG\lr{t_1}\dots\lr{t_{i-1}}^*=\dgroup{(\GG^*)}{\GG\lr{t_1}\dots\lr{t_{i-1}}}{\GG}$ with $1\leq i\leq p$. Thus the \piE-monomials $t_i$ are $\GG^*$-simple. Moreover, the \sigmaSE-monomials $t_i$ on top are all $\GG^*$-simple by definition. Summarizing $\dfield{\FF}{\sigma}\leq\dfield{\EE}{\sigma}$ is simple.
\end{proof}

\noindent In other words, for a reduced and connected difference ring $\dfield{\AA}{\sigma}$ (e.g., if $\AA$ is a field) the notions of \pisiSE-ring extension and simple \pisiSE-ring extension are equivalent. The situation becomes rather different if the ring is, e.g., not connected; see Example~\ref{Exp:NotSimpleExt}. But, for single-rooted \rpisiSE-extensions over a difference field, the situation is again tame.

\begin{corollary}\label{Cor:SingleRootedIsSimple}
A single-rooted \rpisiSE-extension $\dfield{\EE}{\sigma}$ of a field $\dfield{\GG}{\sigma}$ is simple.
\end{corollary}
\begin{proof}
By definition the \rpisiSE-extension can be reordered to the form~\eqref{Def:SingleRooted}. Since $\GG$ is a field, $\gsconst{\GG}{\sigma}\setminus\{0\}=\GG^*$. By Proposition~\ref{Prop:SimpleIsKarrVersion} the \piE-extension $\dfield{\GG\lr{t_1}\dots\lr{t_r}}{\sigma}$ of $\dfield{\GG}{\sigma}$ is simple. Since $\frac{\sigma(x_i)}{x_i}\in\GG^*$ for $1\leq i\leq u$, the \rE-monomials $x_i$ are $\GG^*$-simple. Since also the \sigmaSE-monomials $s_i$ are $\GG^*$-simple, we conclude that $\dfield{\GG}{\sigma}\leq\dfield{\EE}{\sigma}$ is simple.
\end{proof}

\section{The algorithmic machinery I: order, period, factorial order}\label{Sec:Period}

An important ingredient for the development of our summation algorithms is the knowledge of the order (see its definition in~\eqref{equ:OrderDef} and the corresponding Problem~O), the period and the factorial order. In $\dfield{\AA}{\sigma}$ we define the period of $h\in\AA^*$ by\index{function!period}\index{$\per(f)$}
\index{function!factorial order}\index{$\ford(f)$}\begin{align*}
\per(h)&=\begin{cases}
0&\text{ if }\nexists n>0\text{ s.t.\ } \sigma^n(h)=h\\
\min\{n>0|\,\sigma^n(h)=h\}&\text{ otherwise};
\end{cases}
\intertext{and the factorial order of $h$ by}
\ford(h)&=\begin{cases}
0&\text{ if }\nexists n>0\text{ s.t.\ } \sigmaFac{h}{n}=1\\
\min\{n>0|\,\sigmaFac{h}{n}=1\}&\text{ otherwise}.
\end{cases}
\end{align*}
Using the properties of the automorphism $\sigma$ and
Lemma~\ref{Lemma:SigmaFacId} it is easy to see that the $\ZZ$-modules generated by $\ord(h)$, $\per(h)$ and $\ford(h)$ are
$\langle\ord(h)\rangle=\ord(h)\,\ZZ=\{k\in\ZZ|\,h^k=1\}$,
$\langle\per(h)\rangle=\per(h)\,\ZZ=\{k\in\ZZ|\,\sigma^k(h)=h\}$, and
$\langle\ford(h)\rangle=\ford(h)\,\ZZ=\{k\in\ZZ|\,\sigmaFac{h}{k}=1\}$, respectively.
In addition, the following basic properties hold.

\begin{lemma}\label{Lemma:BasicsOrdPer}
 Let $\dfield{\AA}{\sigma}$ be a difference ring with $\alpha,h\in\AA^*$. Then the following holds.
 \begin{enumerate}
  \item If $\alpha\in(\const{\AA}{\sigma})^*$, then $\per(\alpha)=1$ and $\ford(\alpha)=\ord(\alpha)$.
  \item If $\sigma(h)=\alpha\,h$, then $\per(h)=\ford(\alpha)$.
  \item If $\ord(\alpha)>0$ and $\per(\alpha)>0$, then $\per(\alpha)\mid \ford(\alpha)\mid \per(\alpha)\,\ord(\alpha)$ and
\begin{equation}\label{Equ:FindSmallestSigmaFactPer}
\ford(\alpha)=\min(i\,\per(\alpha)|\,1\leq i\leq
\ord(\alpha)\text{ and }\sigmaFac{\alpha}{i\,\per(\alpha)}=1)>0.
\end{equation}
 \end{enumerate}
\end{lemma}
\begin{proof}
(1) Since $\sigma(\alpha)=\alpha$, $\per(\alpha)=1$. Since $\sigmaFac{\alpha}{n}=\alpha^n$ for $n\geq0$, $\ford(\alpha)=\ord(\alpha)$.\\
(2) By Lemma~\ref{Lemma:SigmaFacId}.(4) we have that $\sigma^n(h)=h$ iff $\sigmaFac{\alpha}{n}=1$. Hence $\per(h)=\ford(\alpha)$.\\
(3) Take $p=\per(\alpha)>0$ and $v=\ord(\alpha)>0$. Then we have that
\begin{equation*}
\sigmaFac{\alpha}{p\,v}=\alpha\,\sigma(\alpha)\dots\sigma^{p\,v-1}
(\alpha)=(\alpha\,
\sigma(\alpha)\dots\sigma^{p-1}(\alpha))^v=\alpha^v\,
\sigma(\alpha^v)\dots\sigma^{p-1}(\alpha^v)=1.
\end{equation*}
Consequently, we can choose $n=\ord(\alpha)\,\per(\alpha)$ to obtain
$\sigmaFac{\alpha}{n}=1$. In particular, for any $i\geq0$ with
$\sigmaFac{\alpha}{i}=1$ we have that 
$1=\frac{\sigma(1)}{1}=\frac{\sigma(\sigmaFac{\alpha}{i})}{\sigmaFac{\alpha}{i}
} =\frac{\sigma^i(\alpha)}{\alpha}.$
Hence $\per(\alpha)| i$. 
Thus the smallest $\lambda$ with $\sigmaFac{\alpha}{\lambda}=1$ is given
by~\eqref{Equ:FindSmallestSigmaFactPer}. In particular, 
$\per(\alpha)|\ford(\alpha)|\ord(\alpha)\,\per(\alpha)$.
\end{proof}

We will present methods to calculate the order, period and factorial order for the elements of $\dgroup{(\AA^*)}{\EE}{\AA}$ of a simple \rE-extension $\dfield{\EE}{\sigma}\geq\dfield{\GG}{\sigma}$ by recursion. First, we assume that the orders of the \rE-monomials in $\dfield{\EE}{\sigma}\geq\dfield{\GG}{\sigma}$ are already computed and show how the orders of the elements of $\dgroup{(\AA^*)}{\EE}{\AA}$ can be determined.

\begin{lemma}\label{Lemma:DROrder}
Let $\dfield{\EE}{\sigma}$ with $\EE=\AA\lr{x_1}\dots\lr{x_e}$ be an \rE-extension
of $\dfield{\AA}{\sigma}$ and define
\begin{equation}\label{Equ:alphaForXMon}
\alpha:=u\,x_1^{z_1}\dots x_e^{z_e}\in\dgroup{(\AA^*)}{\EE}{\AA}
\end{equation}
with $u\in\AA^*$ and $z_i\in\NN$. Then $\ord(\alpha)>0$ iff $\ord(u)>0$. If $\ord(u)>0$, then
\begin{equation}\label{Equ:GetOrder}
\ord(\alpha)=\lcm(\ord(u),\tfrac{\ord(x_1)}{\gcd(\ord(x_1),z_1)},\dots,\tfrac{
\ord(x_e)}{\gcd(\ord(x_e),z_e)}).
\end{equation}
\end{lemma}
\begin{proof}
If $e=0$, the lemma holds. Now let $n:=\ord(\alpha)>0$.
Suppose that  $1\neq(x_1^{z_1}\dots x_e^{z_e})^n=x_1^{n\,z_1}\dots
x_e^{n\,z_e}$. Let $i$ be maximal such that $\ord(x_i)\nmid z_i\,n$.
Then there is an $s$ with $0<s<\ord(x_i)$ with
$x_i^{\ord{x_i}-s}=u^n\,x_1^{z_1}\dots
x_{i-1}^{z_{i-1}}\in\AA\lr{x_1}\dots\lr{x_{i-1}}$ which contradicts to the
construction that $x_i^{\ord(x_i)}=1$ is the defining relation of the
\rE-monomial. Thus $(x_1^{z_1}\dots x_e^{z_e})^n=1$ and $u^n=1$, i.e.,
$\ord(u)>0$ and $\ord(x_1^{z_1}\dots x_e^{z_e})>0$. In particular, 
$\ord(\alpha)=\lcm(\ord(u),\ord(x_1^{z_1}\dots x_e^{z_e})).$
By similar arguments we can show that $(x_1^{z_1})^{n}=\dots=(x_e^{z_e})^{n}=1$ and consequently
$\ord(x_1^{z_1}\dots x_e^{z_e})=\lcm(\ord(x_1^{z_1}),\dots,\ord(x_e^{z_e}).$
Since also 
$\ord(x_i^{z_i})=\frac{\ord(x_i)}{\gcd(\ord(x_i),z_i)}$ holds, the identity~\eqref{Equ:GetOrder} is proven.\\ 
Conversely, suppose that $\ord(u)>0$. Then the value of the right right hand side of~\eqref{Equ:GetOrder} is positive. Denote it by $n$. Then one can check that $\alpha^n=1$. Therefore $\ord(\alpha)>0$.
\end{proof}

\noindent In the next lemma we set the stage to calculate the period and factorial order.

\begin{lemma}\label{Lemma:DRPeriod}
Let $\dfield{\EE}{\sigma}$ with $\EE=\AA\lr{x_1}\dots\lr{x_e}$ be an \rE-extension
of $\dfield{\AA}{\sigma}$ where we have $\per(x_i)>0$ for $1\leq i\leq e$. Let $\alpha\in\dgroup{(\AA^*)}{\EE}{\AA}$ as in~\eqref{Equ:alphaForXMon} with
$z_1,\dots,z_e\in\NN$ and $u\in\AA^*$.
\begin{enumerate}
\item Then $\per(\alpha)>0$ iff $\per(u)>0$. If $\per(u)>0$, then
\begin{equation}\label{Equ:BoundPerAlpha}
\per(\alpha)=\min(1\leq j\leq\mu|\,\sigma^j(\alpha)=\alpha\text{ and }j\mid \mu)
\end{equation}
with $\mu=\lcm(\per(u),\per(x_{i_1}),\dots,\per(x_{i_k}))$ where $\{i_1,\dots,i_k\}=\{i:\,\ord(x_i)\nmid z_i\}$.
\item We have that $\ford(\alpha)>0$ iff $\ford(u)>0$.
\item If $\per(u),\ord(u)>0$, then $\ford(\alpha)>0$ and $0<\per(\alpha)|\ford(\alpha)|\per(\alpha)\,\ord(\alpha)$.
\item If the values $\ord(x_i)$ and $\per(x_i)$ for $1\leq i\leq e$ and the values $\per(u)>0$ and $\ord(u)>0$ are given explicitly, then $\per(\alpha)$ and $\ford(\alpha)$ can be calculated.
\end{enumerate}
\end{lemma}
\begin{proof}
(1) Suppose that $\per(u)>0$. Then $\mu>0$.
In particular, it follows
that $\sigma^{\mu}(\alpha)=\alpha$. Consequently, $\per(\alpha)>0$ with
$\per(\alpha)|\mu$. Hence we have~\eqref{Equ:BoundPerAlpha}. Conversely, suppose that $\per(\alpha)>0$. Then with $\nu:=\lcm(\per(\alpha),\per(x_1),\dots,\per(x_e))>0$ we get $u\,x_1^{z_1}\dots x_e^{z_e}=\alpha=\sigma^{\nu}(\alpha)=\sigma^{\nu}(u)\,x_1^{z_1}\dots x_e^{z_e}$. Thus $\sigma^{\nu}(u)=u$, and consequently $\ord(u)>0$.\\
(2)  Since $\ord(x_i)$ and $\per(x_i)>0$, it follows that $\ford(x_i)>0$ by Lemma~\ref{Lemma:BasicsOrdPer}.(3) for all $1\leq i\leq e$. If $\ford(u)>0$, take $\nu=\lcm(\ford(u),\ford(x_1),\dots,\ford(x_e))>0$. By Lemma~\ref{Lemma:SigmaFacId}.(1), $\sigmaFac{\alpha}{\nu}=1$ and hence $\ford(\alpha)>0$. Conversely, if $\ford(\alpha)>0$, take $\nu'=\lcm(\ford(\alpha),\ford(x_1),\dots,\ford(x_e))>0$. Then again by Lemma~\ref{Lemma:SigmaFacId}.(1): $1=\sigmaFac{\alpha}{\nu'}=\sigmaFac{(u\,x_1^{z_1}\dots x_e^{z_e})}{\nu'}=\sigmaFac{u}{\nu'}$. Thus $\ford(u)>0$.\\
(3) By part~1, $\per(\alpha)>0$. And with $\ord(u)>0$ and Lemma~\ref{Lemma:DROrder} it follows that $\ord(\alpha)>0$. By Lemma~\ref{Lemma:BasicsOrdPer}.(3), 
$\per(\alpha)|\ford(\alpha)|\ord(\alpha)\,\per(\alpha)$. In particular, $\ford(\alpha)>0$.\\
(4) If $\per(u)$ and the values $\per(x_i)$ are given, $\mu$ from part~1 can be computed. 
In particular, if $\ord(u)$ and $\ord(x_i)$ are given explicitly, $\ord(\alpha)$ can be calculated by Lemma~\ref{Lemma:DROrder}.
Thus $\per(\alpha)$ can be determined by~\eqref{Equ:BoundPerAlpha} and then $\ford(\alpha)$ can be computed by~\eqref{Equ:FindSmallestSigmaFactPer}.
\end{proof}

\begin{example}\label{Exp:GetOrder}
(1) Take $\alpha=u=-1\in\QQ$. We get $\ord(\alpha)=2$. In addition, $\per(-1)=1$. Moreover, $1=\per(-1)|\ford(-1)|\per(-1)\,\ord(-1)=2$. Hence~\eqref{Equ:BoundPerAlpha} yields $\ford(-1)=2$.\\
(2) Consider the \rE-extension $\dfield{\QQ[x]}{\sigma}$ of $\dfield{\QQ}{\sigma}$ with $\sigma(x)=-x$ and $\ord(x)=2$, and take $\alpha=-x$. We get $\ord(\alpha)=\lcm(\ord(-1),\ord(x))=2$ by~\eqref{Equ:GetOrder}.
With $\mu=\lcm(\per(-1),\per(x))=2$ we get $\per(\alpha)=2$ by using~\eqref{Equ:BoundPerAlpha}. Furthermore, we get $2=\per(\alpha)|\ford(\alpha)|\per(\alpha)\,\ord(\alpha)=4$. Hence with~\eqref{Equ:BoundPerAlpha} we get $\ford(\alpha)=4$.\\
(3) Consider the \rE-extension $\dfield{\KK(k)[x]}{\sigma}$ of $\dfield{\KK(k)}{\sigma}$ with $\sigma(x)=\iota\,x$ and $\ord(x)=4$ from Example~\ref{Exp:QkxtCheckRPiExt}. We have $\per(x)=4$. Take $\alpha=x$. 
We obtain the following bounds $4=\per(\alpha)|\ford(\alpha)|\per(\alpha)\,\ord(\alpha)=16$. Thus with~\eqref{Equ:BoundPerAlpha} we determine $\ford(\alpha)=8$.
\end{example}

\noindent Combining the two lemmas from above we arrive at the following result.

\begin{proposition}\label{Prop:DeterminOrderAndPer}
Let $\dfield{\AA\lr{x_1}\dots\lr{x_e}}{\sigma}$ be a simple \rE-extension
of $\dfield{\AA}{\sigma}$ such that for $1\leq i\leq e$ we have that  $\sigma(x_i)/x_i=u_i\,x_1^{m_{i,1}}\dots x_{i-1}^{m_{i,i-1}}$
with $u_i\in\AA^*$ and $m_{i,j}\in\NN$. Then the following holds.
\begin{enumerate}
\item $\ord(u_i)>0$
for $1\leq i\leq e$. In particular, if the values $\ord(u_i)$ are given explicitly (are
computable), then the values $\ord(x_i)$ are computable.
\item If $\per(u_i)>0$ for $1\leq i\leq e$,
then $\per(x_i)>0$ for $1\leq i\leq e$. In particular, if the values of $\ord(u_i)$ and $\per(u_i)$ for $1\leq i\leq e$  are given
explicitly (are computed), the values $\per(x_i)$ for all $1\leq i\leq e$ are computable.
\end{enumerate}
\end{proposition}
\begin{proof}
(1) By iterative application of Lemma~\ref{Lemma:DRPeriod} it follows that
$\ord(u_i)>0$ for all $1\leq i\leq e$. Moreover, suppose that $\ord(u_i)$ is
given for $1\leq i\leq e$. Furthermore, assume that the values $\ord(x_i)$ for $1\leq i\leq
s$ with $s<e$ are already determined. Then define $\alpha=\sigma(x_s)/x_s$. 
By~\eqref{Equ:GetOrder} we obtain $\ord(\alpha)$ and thus
$\ord(x_s)=\ord(\alpha)$ by Theorem~\ref{Thm:RPSCharacterization}.(3). This completes the induction step.\\
(2) Suppose that $\per(u_i)>0$ for $1\leq i\leq e$. In addition, suppose that we have
shown already that $d_i=\per(x_i)>0$ for $1\leq i<s$ with $s\leq e$. Define
$\alpha=\sigma(x_s)/x_s$. By Lemma~\ref{Lemma:DRPeriod} we have $\per(\alpha)>0$ and 
$\ford(\alpha)>0$. By Lemma~\ref{Lemma:BasicsOrdPer}.(2) it follows that $\per(x_s)=\ford(\alpha)>0$. If the values $\ord(u_i)$ are given explicitly, we can compute $\ord(\alpha)$ by part~1. If $\per(u_s)$ is given explicitly and
$d_1,\dots,d_{s-1}$ are given (are already computed), $\per(\alpha)$ can be computed with Lemma~\ref{Lemma:BasicsOrdPer}.(3). Hence 
$\ford(\alpha)$ can be calculated with~\eqref{Equ:FindSmallestSigmaFactPer}. Thus we get $\ord(x_s)=\ford(\alpha)$ by Lemma~\ref{Lemma:BasicsOrdPer}.(2) which completes the induction step.
\end{proof}

\noindent If we restrict to the case that the ground domain is a field $\FF$ and all roots of unity of $\FF$ are constants, we end up at the following properties of \rE-extensions.

\begin{corollary}\label{Cor:SimpleNestedRBasics}
Let $\dfield{\EE}{\sigma}$ with $\EE=\FF\lr{x_1}\dots\lr{x_e}$ 
be a simple \rE-extension of a difference field $\dfield{\FF}{\sigma}$ with constant field $\KK$ such that all roots of unity in
$\FF$ are
constants (e.g., if $\dfield{\FF}{\sigma}$ is strong constant-stable). 
Then the following holds.
\begin{enumerate}
\item For $1\leq i\leq e$ we have that
\begin{equation}\label{Equ:SimpleDef2}
\sigma(x_i)/x_i=u_i\,x_1^{m_{i,1}}\dots x_{i-1}^{m_{i,i-1}}
\end{equation}
for some root of unity $u_i\in\KK^*$ with $\ord(u_i)\mid\ord(x_i)$ and
$m_{i,j}\in\NN$. 
\item 
$\dfield{\KK\lr{x_1}\dots\lr{x_e}}{\sigma}$ is a simple \rE-extension
of $\dfield{\KK}{\sigma}$. 
\item Let $\alpha=u\,x_1^{z_1}\dots
x_e^{z_e}\in\dgroup{(\KK^*)}{\KK\lr{x_1}\dots\lr{x_e}}{\KK}$ with $z_1,\dots,z_e\in\NN$ and $u\in\KK^*$. Then 
$$\ord(u)>0\Leftrightarrow
\ord(\alpha)>0\Leftrightarrow\per(\alpha)>0\Leftrightarrow\ford(\alpha)>0.$$
\item If $\dfield{\KK}{\sigma}$ is computable and Problem~O is solvable in $\KK^*$ then the values of $\ord(\alpha)$, $\per(\alpha)$ and $\ford(\alpha)$ are computable for all $\alpha\in\dgroup{(\KK^*)}{\KK\lr{x_1}\dots\lr{x_e}}{\KK}$.
\item Problem~O is solvable in $\dgroup{(\FF^*)}{\EE}{\FF}$ if it is solvable in $\KK^*$ and $\dfield{\FF}{\sigma}$ is computable.
\end{enumerate}
\end{corollary}
\begin{proof}
(1) By definition we have that~\eqref{Equ:SimpleDef2} with $m_i\in\NN$ and
$u_i\in\FF^*$. By Lemma~\ref{Lemma:DROrder} it follows that $\ord(u_i)>0$
and $\ord(u_i)|\ord(x_i)$. In particular,
$u_i\in\KK^*$ since all roots of unity from $\FF$ are constants by assumption.\\
(2) It is immediate that $\dfield{\HH}{\sigma}$ with $\HH=\KK\lr{x_1}\dots\lr{x_e}$ forms a
difference ring. Since $\const{\EE}{\sigma}=\const{\FF}{\sigma}=\KK$, $\dfield{\HH}{\sigma}$ is a simple
\rE-extension of $\dfield{\KK}{\sigma}$.\\ 
(3) By part~1 we get $u_i\in\KK^*$ and $\ord(u_i)>0$ for $1\leq i\leq e$. In particular, $\per(u_i)=1$. With Proposition~\ref{Prop:DeterminOrderAndPer} we get $\per(x_i)>0$, and by Lemma~\ref{Lemma:BasicsOrdPer}.(1) we obtain $\per(u)=1$ and $\ford(u)=\ord(u)$. Thus the equivalences follow by Lemmas~\ref{Lemma:DROrder} and~\ref{Lemma:DRPeriod} (parts 1,2).\\
(4) Since $u_i\in\KK^*$, the values of $\ord(u_i)>0$ can be determined by solving Problem~O in $\KK^*$. Thus by
Proposition~\ref{Prop:DeterminOrderAndPer} the orders and periods of the $x_i$
can be computed. Let $\alpha:=u\,x_1^{z_1}\dots x_e^{z_e}$ with $u\in\KK^*$ and $z_i\in\NN$.
Then by Lemma~\ref{Lemma:DROrder} and the computation of $\ord(u)$ the order of $\alpha$ can be computed. Moreover, since $\per(u)=1$ and $\ord(u)=\ford(u)$ are given, we can invoke
Lemma~\ref{Lemma:DRPeriod} to calculate the period and factorial order of $\alpha$.\\
(5) Let $\alpha$ be given as in~\eqref{Equ:alphaForXMon} with $u\in\FF^*$ and $m_i\in\NN$. By Lemma~\ref{Lemma:DROrder} $\ord(\alpha)>0$ iff $\ord(u)>0$. By assumption, $\ord(u)>0$ implies $u\in\KK^*$.
Thus, if $u\notin\KK$, $\ord(\alpha)=0$.
Otherwise, if $u\in\KK^*$, we can apply part~4.
\end{proof}

\noindent Finally, we are in the position to prove Theorem~\ref{Thm:AlgMainResultFull}.(1).

\ExternalProof{(Theorem~\ref{Thm:AlgMainResultFull}.(1)\label{Proof:AlgMainResultFull1})}{
Let $\dfield{\EE}{\sigma}$ be a simple \rpisiSE-extension of $\dfield{\FF}{\sigma}$ where $\dfield{\FF}{\sigma}$ is computable and where any root of unity of $\FF$ is from $\KK=\const{\FF}{\sigma}$. Reorder it to the shape as given in Lemma~\ref{Lemma:ReorderSimpleRPISI}. In particular, the \rE-extension $\dfield{\FF\lr{t_1}\dots\lr{t_r}}{\sigma}$ of $\dfield{\FF}{\sigma}$ has the shape as given in Corollary~\ref{Cor:SimpleNestedRBasics}.(1).
Let $f\in\dgroup{(\FF^*)}{\EE}{\FF}$. Suppose first that $f$ depends on a \piE-monomial $t_i$. Now assume that $\ord(f)=n>0$, and let $i$ be maximal such that a \piE-monomial depends on $f$. Then $f=v\,t_i^m$ with $v\in\FF\lr{t_1}\dots\lr{t_{i-1}}^*$ and $m\in\ZZ\setminus\{0\}$. Hence $1=f^n=v^n\,t_i^{m\,n}$ and thus $t_i$ is not algebraically independent over $\FF\lr{t_1}\dots\lr{t_{i-1}}$; a contradiction. Consequently, if $f$ depends on \piE-monomials, $\ord(f)=0$. Otherwise, $f=u\,t_1^{m_1}\dots t_r^{m_r}$ with $u\in\FF^*$ and $m_i\in\NN$ where the $t_i$ are all \rE-monomials. Therefore the value $\ord(f)$ can be computed by Corollary~\ref{Cor:SimpleNestedRBasics}.(5).
}

\section{The algorithmic machinery II: Problem~PMT}\label{Sec:M}

We aim at proving Theorems~\ref{Thm:AlgMainResultRestricted}.(1) and~\ref{Thm:AlgMainResultFull}.(2), i.e., providing recursive algorithms that reduce Problem~PMT from a given \rpisiSE-extension to its ground ring (resp.\ field).
For this reduction we assume that for the given ground ring $\dfield{\GG}{\sigma}$ and given group $G\leq\GG^*$ we have that $\sconst{G}{\GG}{\sigma}\setminus\{0\}\leq\GG^*$. This property guarantees that for any $\vect{f}\in G^n$ a $\ZZ$-basis of $M(\vect{f},\GG)$ with rank $\leq n$ exists; see Lemma~\ref{Lemma:MZModule}. In particular, we rely on the fact that there are algorithms available that solve Problem~PMT in $\dfield{\GG}{\sigma}$ for $G$. For concrete classes of difference fields $\dfield{\GG}{\sigma}$ with these algorithmic properties we refer to Subsection~\ref{Subsec:GroundFieldAlg}.

\subsection{A reduction strategy for \pisiSE-extensions}

First, we treat the reduction for \pisiSE-extensions. More precisely, we will obtain

\begin{theorem}\label{Thm:ProblemMPiSi}
Let $\dfield{\GG}{\sigma}$ be a computable difference ring with $G\leq\GG^*$ where $\sconst{G}{\GG}{\sigma}\setminus\{0\}\leq\GG^*$.
Let $\dfield{\EE}{\sigma}$ be a $G$-simple \pisiSE-extension of $\dfield{\GG}{\sigma}$. Then $\sconst{\dgroup{G}{\EE}{\GG}}{\EE}{\sigma}\setminus\{0\}\leq\EE^*$ and 
Problem~PMT is solvable in $\dfield{\EE}{\sigma}$ for
$\dgroup{G}{\EE}{\GG}$ if it is solvable in $\dfield{\GG}{\sigma}$ for
$G$. 
\end{theorem}

\noindent For the underlying reduction method we use the following two lemmas.

\begin{lemma}\label{Lemma:MProblemRemoveSigmaExt}
Let $\dfield{\AA[t]}{\sigma}$ be a \sigmaSE-extension of
$\dfield{\AA}{\sigma}$ and let $H\leq\AA^*$ be a group with
$\sconst{H}{\AA}{\sigma}\setminus\{0\}\leq\AA^*$. Then for $\vect{f}\in H^n$ we have that
$M(\vect{f},\AA[t])=M(\vect{f},\AA)$.
\end{lemma}
\begin{proof}
``$\subseteq$'': Let $\vect{m}=(m_1,\dots,m_n)\in
M(\vect{f},\AA[t])$ with $\vect{f}=(f_1,\dots,f_n)\in H^n$. Thus take
$g\in\AA[t]\setminus\{0\}$ with 
$\sigma(g)=f_1^{m_1}\dots f_n^{m_n}\,g$. Since
$g\in\sconst{H}{\AA[t]}{\sigma}\setminus\{0\}$, we have 
$g\in\sconst{H}{\AA}{\sigma}$ by
Theorem~\ref{Thm:SigmaTheoremPart2}. Hence $\vect{m}\in M(\vect {f},\AA)$. The inclusion $\supseteq$ is obvious. 
\end{proof}

\begin{lemma}\label{Lemma:MProblemReducePi}
Let $\dfield{\AA\lr{t}}{\sigma}$ be a \piE-extension of
$\dfield{\AA}{\sigma}$ and let $H\leq\AA^*$ with
$\sconst{H}{\AA}{\sigma}\setminus\{0\}\leq\AA^*$ and 
$\alpha:=\sigma(t)/t\in H$. Let
$\vect{f}=(f_1,\dots,f_n)\in(\dgroup{H}{\AA\lr{t}}{\AA})^n$ 
with 
$$f_i=h_i\,t^{e_i},\quad h_i\in H,\,e_i\in\ZZ.$$
Then
$M(\vect{f},\AA\lr{t})=M_1\cap M_2$
where
\begin{align*}
M_1&=\{(m_1,\dots,m_n)|\,(m_1,\dots,m_n,m_{n+1})\in
M((h_1,\dots,h_n,\tfrac{1}{\alpha}),\AA)\},\\
M_2&=\text{Ann}_{\ZZ}((e_1,\dots,e_n))=\{(m_1\dots,m_n)\in\ZZ^n|\,m_1\,
e_1+\dots+m_n\,e_n=0\}.
\end{align*}
\end{lemma}
\begin{proof}
``$\subseteq$'': Let $(m_1,\dots,m_n)\in M(\vect{f},\AA\lr{t})$. Hence we can take
$g\in\AA\lr{t}\setminus\{0\}$ with
$\sigma(g)=f_1^{m_1}\dots f_n^{m_n}\,g,$
i.e., $g\in\sconst{\tilde{H}}{\AA\lr{t}}{\sigma}\setminus\{0\}$ with
$\tilde{H}=\dgroup{H}{\AA\lr{t}}{\AA}$. Thus by Theorem~\ref{Thm:PiPart2Strong}
it follows that $g=\tilde{g}\,t^m$ with $m\in\ZZ$ and
$\tilde{g}\in\sconst{\tilde{H}}{\AA}{\sigma}\setminus\{0\}\leq\AA^*$. Hence
$$\sigma(\tilde{g})=f_1^{m_1}\dots f_n^{m_n}\alpha^{-m}\tilde{g}=
h_1^{m_1}\dots
h_n^{m_n}\,\alpha^{-m}\,\tilde{g}\,t^{m_1\,e_1+\dots+m_n\,e_n}.$$
Since $\tilde{g}\neq0$, we conclude that $\sigma(\tilde{g})\neq0$. By
coefficient comparison it follows then that $m_1\,e_1+\dots+m_n\,e_n=0$, i.e., 
$(m_1,\dots,m_n)\in M_2$.
Thus 
$\sigma(\tilde{g})=h_1^{m_1}\dots
h_n^{m_n}\alpha^{-m}\,\tilde{g}$ and consequently 
$(m_1,\dots,m_n,m)\in M((h_1,\dots,h_n,\tfrac1{\alpha}),\AA)$,
i.e.,
$(m_1,\dots,m_n)\in M_1$.\\
``$\supseteq$'': Let
$(m_1,\dots,m_n)\in M_1\cap M_2$. Thus we can take $\tilde{g}\in\AA\setminus\{0\}$ and $m\in\ZZ$ with 
$\sigma(\tilde{g})=h_1^{m_1}\dots h_n^{m_n}\alpha^{-m}\tilde{g}.$
Moreover, we have that $e_1\,m_1+\dots+e_n\,m_n=0$. Thus
$\sigma(\tilde{g}\,t^m)=(h_1\,t^{e_1})^{m_1}\dots(h_n\,t^{e_n})^{m_n}\,\tilde{g
}$
and therefore $(m_1,\dots,m_n)\in M(\vect{f},\AA\lr{t})$.
\end{proof}

\noindent Now we can deal with the underlying algorithm resp.\ proof of Theorem~\ref{Thm:ProblemMPiSi}.

\ExternalProof{
(Theorem~\ref{Thm:ProblemMPiSi})}{
Let $\dfield{\GG}{\sigma}$ be a difference ring and let $G\leq\GG^*$ such that $\sconst{G}{\GG}{\sigma}\setminus\{0\}\leq\GG^*$ holds. Suppose that Problem~PMT is solvable in $\dfield{\GG}{\sigma}$ for $G$. Now let $\dfield{\EE}{\sigma}$ be a $G$-simple \pisiSE-extension of $\dfield{\GG}{\sigma}$ as in the theorem with $\tilde{G}=\dgroup{G}{\EE}{\GG}$ and let $\vect{f}\in\tilde{G}^n$. 
By Corollary~\ref{Cor:pisiSCONST}.(1) it follows that $\sconst{\tilde{G}}{\EE}{\sigma}\setminus\{0\}\leq\EE^*$ and together with Lemma~\ref{Lemma:MZModule} it follows that
$M(\vect{f},\EE)=M(\vect{f},\sconst{\tilde{G}}{\EE}{\sigma})$
is a $\ZZ$-module. The calculation of a basis of $M(\vect{f},\EE)$ will be accomplished by recursion/induction. If $\EE=\AA$, nothing has to be shown. Otherwise, let
$\dfield{\AA}{\sigma}$ be a $G$-simple \pisiSE-extension of $\dfield{\GG}{\sigma}$ in which we know how one can solve Problem~PMT for $H=\dgroup{G}{\AA}{\GG}$, and let
$\EE=\AA\lr{t}$ where $t$ is a $H$-simple \pisiSE-monomial. We have to treat two cases.
First, suppose that $t$ is a \sigmaSE-monomial. Then it follows that
$\tilde{G}=\dgroup{G}{\EE}{\GG}=\dgroup{G}{\AA}{\GG}=H\leq\AA^*$
and thus $\vect{f}\in H^n$.
Hence we can activate Lemma~\ref{Lemma:MProblemRemoveSigmaExt}
and it follows that $M(\vect{f},\EE)=M(\vect{f},\AA)$. 
Thus by assumption we can compute a basis.
Second, suppose that $t$ is a $H$-simple \piE-monomial.
Then we can utilize Lemma~\ref{Lemma:MProblemReducePi}: We calculate a basis of $M_2$ by linear algebra. Furthermore, we compute a basis of $M((h_1,\dots,h_n,\tfrac{1}{\alpha}),\AA)$ by the induction assumption (by recursion). Hence we can derive a basis of $M_2$ and thus of $M_1\cap M_2=M(\vect{f},\AA\lr{t})$. This completes the proof.
}

\noindent Note that the reduction presented in Lemma~\ref{Lemma:MProblemReducePi} is accomplished by increasing the rank of $M_1$ by one. In general, the more \piE-monomials are involved, the higher the rank will be in the arising Problems~PMT of the recursions.

Looking closer at the reduction algorithm, we can extract the following shortcut, resp.\ a refined version of Theorem~\ref{Thm:RPSCharacterization}.(2).

\begin{corollary}\label{Cor:MShortCutSigma}
Let $\dfield{\AA}{\sigma}$ be a difference ring and let $G\leq\AA^*$ with $\sconst{G}{\AA}{\sigma}\setminus\{0\}\leq\AA^*$.
Let $\dfield{\HH}{\sigma}$ be a $G$-simple \piE-extension of
$\dfield{\AA}{\sigma}$ and let
$\dfield{\EE}{\sigma}$ be a \sigmaSE-extension of $\dfield{\HH}{\sigma}$. 
Then $\dgroup{G}{\EE}{\AA}=\dgroup{G}{\HH}{\AA}$ and the following holds.
\begin{enumerate}
\item $M(\vect{f},\EE)=M(\vect{f},\HH)$ for any
$\vect{f}\in(\dgroup{G}{\EE}{\AA})^n$.
\item Let $\alpha\in\dgroup{G}{\EE}{\AA}$. Then there is a \piE-extension
$\dfield{\EE\lr{t}}{\sigma}$
of $\dfield{\EE}{\sigma}$ with $\sigma(t)=\alpha\,t$ iff there is a
\piE-extension $\dfield{\HH\lr{t}}{\sigma}$ of $\dfield{\HH}{\sigma}$ with $\sigma(t)=\alpha\,t$.
\end{enumerate}
\end{corollary}
\begin{proof}
Note that $\dgroup{G}{\EE}{\AA}\leq\HH^*$. Hence by iterative application of
Lemma~\ref{Lemma:MProblemRemoveSigmaExt} part~1 is proven. Part 2 follows by part~1 and Theorem~\ref{Thm:RPSCharacterization}.(2). 
\end{proof}

If one restricts to the special case that $\dfield{\GG}{\sigma}$ is a \pisiSE-field with $G=\GG^*$, the presented reduction techniques boil down to the reduction presented in~\cite[Theorem~8]{Karr:81}. The major contribution here is that  Theorem~\ref{Thm:ProblemMPiSi} can be applied for any computable difference ring $\dfield{\GG}{\sigma}$ with the properties given in Theorem~\ref{Thm:ProblemMPiSi}.
Subsequently, we utilize this additional flexibility to tackle (nested) \rE-extensions. 

\subsection{A reduction strategy for \rE-extensions and thus for \rpisiSE-extensions}

First, we treat the special case of single-rooted and simple \rE-extensions.

\begin{lemma}\label{Lemma:SingleRootExtension}
Let $\dfield{\AA}{\sigma}$ be a difference ring and let $G\leq\AA^*$ with $\sconst{G}{\AA}{\sigma}\setminus\{0\}\leq\AA^*$.
Let $\dfield{\AA[x]}{\sigma}$ be an \rE-extension of $\dfield{\AA}{\sigma}$
with $\sigma(x)=\alpha\,x$ where $\alpha\in G$; let
$\vect{f}=(f_1,\dots,f_n)\in G^n$. Then
$M(\vect{f},\AA[x])=\{(m_1,\dots,m_n)|\,(m_1,\dots,m_{n+1})\in
M((f_1,\dots,f_n,\tfrac{1}{\alpha}),\AA).$
\end{lemma}
\begin{proof}
Let $(m_1,\dots,m_n)\in M(\vect{f},\AA[x])$. Hence there is a
$g\in\sconst{G}{\AA[x]}{\sigma}\setminus\{0\}$ with $\sigma(g)=f_1^{m_1}\dots
f_n^{m_n}\,g$. By Proposition~\ref{Prop:RPart2Weak} it follows that
$g=\tilde{g}\,x^m$ with $\tilde{g}\in\AA\setminus\{0\}$ and $m\in\NN$. Thus
\begin{equation}\label{Equ:TildegMRelation}
\sigma(\tilde{g})=f_1^{m_1}\dots f_n^{m_n}\,\alpha^{-m}\,\tilde{g} 
\end{equation}
and hence $(m_1,\dots,m_n,m)\in M((f_1,\dots,f_n,\tfrac1\alpha),\AA)$.
Conversely, if $(m_1,\dots,m_n,m)\in M((f_1,\dots,f_n,1/\alpha,\AA)$, there is
a $\tilde{g}\in\AA\setminus\{0\}$ with~\eqref{Equ:TildegMRelation}. Therefore we conclude that
$\sigma(\tilde{g}\,t^m)=f_1^{m_1}\dots
f_n^{m_n}\,\tilde{g}\,t^m$ which implies that $(m_1,\dots,m_n)\in M(\vect{f},\AA[x])$. 
\end{proof}

\noindent As a consequence we obtain the proof of our Theorem~\ref{Thm:AlgMainResultRestricted}.(1). 

\ExternalProof{(Theorem~\ref{Thm:AlgMainResultRestricted}.(1)\label{Proof:AlgMainResultRestricted1})}{
Since Problem~PMT is solvable in $\dfield{\GG}{\sigma}$ for $G$, it follows by Theorem~\ref{Thm:ProblemMPiSi} that Problem~PMT is solvable in $\dfield{\HH}{\sigma}$ for $\tilde{G}$ with $\HH=\GG\lr{t_1}\dots\lr{t_r}$ and that $\sconst{\tilde{G}}{\HH}{\sigma}\setminus\{0\}\leq\HH^*$. Thus by iterative applications of Lemma~\ref{Lemma:SingleRootExtension} and Proposition~\ref{Prop:RPart2Weak} we conclude that Problem~PMT is solvable in $\dfield{\bar{\HH}}{\sigma}$ for $\tilde{G}$ with $\bar{\HH}=\HH\lr{x_1}\dots\lr{x_u}$ and that $\sconst{\tilde{G}}{\bar{\HH}}{\sigma}\setminus\{0\}\leq\bar{\HH}^*$.  Finally, by applying again Theorem~\ref{Thm:ProblemMPiSi} it follows that Problem~PMT is solvable in $\dfield{\EE}{\sigma}$ for $\tilde{G}$. 
}

In order to tackle the more general case that the \rE-extensions are nested and that they might occur also in \piE-extensions (see the underlying algorithms for Theorem~\ref{Thm:AlgMainResultFull}.(2) in Proof~\ref{Proof:AlgMainResultFull2} below), we require additional properties on the
difference rings: they must be strong constant-stable; see Definition~\ref{Def:ConstantStable}. With this extra condition the following structural property of the semi-constants holds. They factor into two parts: a factor which depends only on the \rE-monomials with constant coefficients and a factor which is free of the \rE-monomials.

\begin{lemma}\label{Lemma:RExtStructureForRing}
Let $\dfield{\AA}{\sigma}$ be a difference ring which is 
constant-stable and let $G\leq\AA^*$ be closed under $\sigma$ where $\sconstF{G}{\AA}{\sigma^k}\setminus\{0\}\leq\AA^*$ for any $k>0$.
Let $\dfield{\EE}{\sigma}$ be a simple \rE-extension of
$\dfield{\AA}{\sigma}$ with $\EE=\AA[x_1]\dots[x_e]$ where we
have~\eqref{Equ:SimpleDef2} with $m_{i,j}\in\NN$ and $u_i\in G$ with $\per(u_i)>0$. Define
\begin{equation}\label{Equ:DefineRinMProblem}
r:=\begin{cases}
\lcm(\ford(u_1),\dots,\ford(u_e),\ford(x_1),\dots,\ford(x_e)) &\text{ if
}e>0\\
1&\text{ if }e=0.
   \end{cases}
\end{equation}
Let
$\tilde{G}=\dgroup{G}{\EE}{\AA}$ with
$\sconst{\tilde{G}}{\EE}{\sigma}\setminus\{0\}\leq\EE^*$. Then the following holds.\\ (1) $r>0$.
(2) For any $g\in\sconst{\tilde{G}}{\EE}{\sigma}\setminus\{0\}$ we have that 
\begin{equation}\label{Equ:SeperatexPartFPart}
g=\tilde{g}\,h
\end{equation}
with $\tilde{g}\in\sconstF{G}{\AA}{\sigma^r}\setminus\{0\}\leq\AA^*$ and
$h\in\constF{\KK[x_1,\dots,x_e]}{\sigma^r}^*$.\\ 
(3) If $\sigma(g)=v\,x_1^{m_1}\dots x_e^{m_e}\,g$ with $v\in G$, $m_i\in\NN$,
then 
$\sigma(\tilde{g})=\lambda\,v\,\tilde{g}$ with $\lambda\in\AA^*$, $\lambda^r=1$.
\end{lemma}

\begin{proof}
Let $g\in\sconst{\tilde{G}}{\EE}{\sigma}\setminus\{0\}$, i.e.,
$\sigma(g)=v\,x_1^{m_1}\dots x_e^{m_e}\,g$ with $v\in G$ and $m_i\in\ZZ$. Let
$r$ be given as in~\eqref{Equ:DefineRinMProblem}. If $e=0$, i.e., $r=1$, the
lemma holds by taking $\tilde{g}:=g\in\AA^*$ and $h=1$. Otherwise, we may
suppose that $e>0$.\\
(1) Let $1\leq i\leq e$. By Proposition~\ref{Prop:DeterminOrderAndPer}.(1) it follows that $\ord(u_i)>0$. Together with the assumption that $\per(u_i)>0$ we have that $\ford(u_i)>0$ by Lemma~\ref{Lemma:BasicsOrdPer}.(3). Moreover,
by Proposition~\ref{Prop:DeterminOrderAndPer}.(2) it follows that
$\per(x_i)>0$. Again with $\ord(x_i)>0$ and $\per(x_i)>0$ it follows that $\ford(x_i)>0$ by Lemma~\ref{Lemma:BasicsOrdPer}.(3).
Therefore $r>0$.\\
(2) By the choice of $r$ it follows that for all $1\leq i\leq e$ we have 
\begin{equation}\label{Equ:rProperties}
\sigmaFac{(u_i)}{r}=1,\quad \sigmaFac{(x_i)}{r}=1\text{ and }\sigma^r(x_i)=x_i;
\end{equation}
the last equality follows by Lemma~\ref{Lemma:DRPeriod}.(3).
Moreover, by Lemma~\ref{Lemma:SigmaFacId} we conclude that
$$\sigma^r(g)=\sigmaFac{(v\,x_1^{m_1}\dots
x_e^{m_e})}{r}\,g=\sigmaFac{v}{r}\,\sigmaFac{(x_1)}{r}^{m_1}\dots\sigmaFac{(x_e)}{r}^{m_e}\,
g=\tilde {u}\,g$$ 
with $\tilde{u}:=\sigmaFac{v}{r}$. Since $G$ is closed under $\sigma$, we have
that $\tilde{u}\in G$.
Write
$g=\sum_{\vect{s}\in S}g_{\vect{s}}\vect{x}^{\vect{s}}$
where $S\subseteq\NN^e$ is finite, $g_{\vect{s}}\in\FF^*$ and for
$(s_1,\dots,s_e)\in S$ and $\vect{x}=(x_1,\dots,x_e)$ we use the multi-index
notation $\vect{x}^{\vect{s}}=x_1^{s_1}\dots x_e^{s_e}$. In particular, we
suppose that if $\vect{s},\vect{s'}\in S$ with
$\vect{x}^{\vect{s}}=\vect{x}^{\vect{s'}}$ then $\vect{s}=\vect{s'}$. 
Then by coefficient comparison w.r.t.\ $\vect{x}^{\vect{i}}$ and using~\eqref{Equ:rProperties} we obtain
$\sigma^r(g_{\vect{i}})=\tilde{u}\,g_{\vect{i}}$ for any $\vect{i}\in S$.
Note that $g_{\vect{i}}\in\sconstF{G}{\AA}{\sigma^r}\setminus\{0\}\leq\AA^*$.
Hence for any $\vect{s},\vect{r}\in S$ we have that
$\sigma^r(g_{\vect{s}}/g_{\vect{r}})=g_{\vect{s}}/g_{\vect{r}}$. Thus it
follows that $g_{\vect{s}}/g_{\vect{r}}\in(\constF{\AA}{\sigma^r})^*=\KK^*$,
i.e., for all
$s\in S$ we have that
$g_{\vect{s}}=c_{\vect{s}}\,\tilde{g}$
for some $c_{\vect{s}}\in\KK^*$ and $\tilde{g}\in\sconstF{G}{\AA}{\sigma^r}\setminus\{0\}\leq\AA^*$
with
\begin{equation}\label{Equ:RelationForgTilde}
\sigma^r(\tilde{g})=\tilde{u}\,\tilde{g}.
\end{equation}
Consequently, $g=\tilde{g}\,h$ with $h=\sum_{\vect{s}\in S}
c_{\vect{s}}\vect{x}^{\vect{s}}$. Since $g\in\AA^*$, 
$h\in\KK[x_1,\dots,x_e]^*$. Finally, with~\eqref{Equ:rProperties} we conclude that that $h\in\constF{\KK[x_1,\dots,x_e]}{\sigma^r}^*$.\\ 
(3) Taking $\vect{s}=(s_1,\dots,s_e)\in S$, it is easy to see that there is exactly one $\vect{s'}\in S$ with 
$$\sigma(c_{\vect{s}}\vect{x}^{\vect{s}}\tilde{g})=v\,x_1^{m_1}\dots
x_e^{m_e}c_{\vect{s'}}\vect{x}^{\vect{s'}}.$$
This means that on both
sides the same monomial $\vect{x}^{\vect{s'}+(m_1,\dots,m_e)}$ in reduced form
occurs. By coefficient
comparison this gives
$\sigma(\tilde{g})=v\,u_1^{-s_1}\dots
u_e^{-s_e}\,\tfrac{c_{\vect{s'}}}{c_{\vect{s}}}\tilde{g}.$
Thus with~\eqref{Equ:rProperties} and Lemma~\ref{Lemma:SigmaFacId} we get
$\sigma^r(\tilde{g})=\sigmaFac{v}{r}(\tfrac{c_{\vect{s'}}}{c_{\vect{s}}}
)^r\tilde{g}=\tilde{u}\,(\tfrac{c_{\vect{s'}}}{c_{\vect{s}}}
)^r\,\tilde{g}$. Hence with~\eqref{Equ:RelationForgTilde} we obtain 
$(\tfrac{c_{\vect{s'}}}{c_{\vect{s}}}
)^r=1$. Finally, with $\lambda:=u_1^{-s_1}\dots
u_e^{-s_e}\,\tfrac{c_{\vect{s'}}}{c_{\vect{s}}}$ we have that $\sigma(\tilde{g})=\lambda\,v\,\tilde{g}$ with $\lambda^r=1$ and
$\lambda\in\AA^*$.
\end{proof}

\noindent Specializing $\AA$ to a strong constant-stable difference field, the lemma reads as follows.

\begin{corollary}\label{Cor:RExtStructureForField}
Let $\dfield{\FF}{\sigma}$ be a difference field with $\KK=\const{\FF}{\sigma}$ which is strong
constant-stable.
Let $\dfield{\EE}{\sigma}$ be a simple \rE-extension of
$\dfield{\FF}{\sigma}$ with $\EE=\FF[x_1]\dots[x_e]$ such that~\eqref{Equ:SimpleDef2} holds with $m_{i,j}\in\NN$, $u_i\in\KK^*$, and
define~\eqref{Equ:DefineRinMProblem}.
Let
$\tilde{G}=\dgroup{(\FF^*)}{\EE}{\FF}$. 
Then: (1) $r>0$.\\
(2) For any $g\in\sconst{\tilde{G}}{\EE}{\sigma}\setminus\{0\}$ we have~\eqref{Equ:SeperatexPartFPart}
with $\tilde{g}\in\FF^*$ and $h\in\constF{\KK[x_1,\dots,x_e]}{\sigma^r}^*$.\\ 
(3) If $\sigma(g)=v\,x_1^{m_1}\dots x_e^{m_e}\,g$ with $v\in\FF^*$, $m_i\in\ZZ$,
then 
$\sigma(\tilde{g})=\lambda\,v\,\tilde{g}$ with $\lambda\in\KK^*$, $\lambda^r=1$.
\end{corollary}

\begin{proof}
Since $u_i\in\KK^*$ by Corollary~\ref{Cor:SimpleNestedRBasics}.(1), $\per(u_i)=1$. Define $G=\FF^*$ which is closed
under $\sigma$. In particular, $\sconstF{G}{\FF}{\sigma^k}\setminus\{0\}=\FF^*$ for any $k>0$.
In addition, 
$\sconst{\tilde{G}}{\EE}{\FF}\setminus\{0\}\leq\EE^*$ by Corollary~\ref{Cor:SConstRExtension}. Thus we can apply
Lemma~\ref{Lemma:RExtStructureForRing}. The corollary follows by observing that 
$\lambda\in\FF^*$ with $\lambda^r=1$. Then by our assumption it follows that $\lambda\in\KK^*$.
\end{proof}

\noindent With this result we get the following reduction tactic for simple \rE-extensions.

\begin{lemma}\label{Lemma:MProblemForRExtensions}
Let $\dfield{\FF}{\sigma}$ be a difference field with $\KK=\const{\FF}{\sigma}$
which is strong constant-stable.
Let $\dfield{\EE}{\sigma}$ be a simple \rE-extension of
$\dfield{\FF}{\sigma}$ with $\EE=\FF[x_1]\dots[x_e]$ where we
have~\eqref{Equ:SimpleDef2} with $m_{i,j}\in\NN$ and $u_i\in\KK^*$. Define
$r>0$ as given in~\eqref{Equ:DefineRinMProblem} and choose\footnote{In principal, we could also take one primitive $r$th root of unity $\alpha$. However, if $\alpha\notin\KK$, we have to extend the constant field. By efficiency reasons we prefer to stay in the original field. We remark that extending the constant field would not produce further relations.} a set 
$\{\alpha_1,\dots,\alpha_s\}\subseteq\KK^*$ of $r$-th roots of unity which generate multiplicatively all $r$-th roots of unity of $\KK$. Let
$G=\dgroup{(\FF^*)}{\EE}{\FF}$ and let $\vect{f}=(f_1,\dots,f_n)\in G^n$ with
$f_i=\tilde{f}_i\,h_i$ where $\tilde{f}_i\in\FF^*$ and
$h_i=x_1^{z_{i,1}}\dots x_e^{z_{i,e}}$ with $z_{i,j}\in\NN$. Then
\begin{equation}\label{Equ:GetMRExt}
M(\vect{f},\EE)=\{(m_1,\dots,m_n)|\,(m_1,\dots,m_{n+s})\in
M_1\cap M_2\}
\end{equation}
where 
\begin{align*}
 M_1&=M((\tilde{f}_1,\dots,\tilde{f}_n,\alpha_1,\dots,\alpha_s),\FF),\\
 M_2&=M((h_1,\dots,h_n,\tfrac1{\alpha_1},\dots,\tfrac1{\alpha_s}),\KK[x_1]\dots
[x_e]).
\end{align*}
\end{lemma}
\begin{proof}
Let $(m_1,\dots,m_n)\in M(\vect{f},\EE)$, i.e., there is a
$g\in\sconst{G}{\EE}{\FF}\setminus\{0\}$ with
$\sigma(g)=f_1^{m_1}\dots f_n^{m_n}\,g$.
Hence by Corollary~\ref{Cor:RExtStructureForField} it follows that
$g=\tilde{g}\,h$ with
$\tilde{g}\in\FF^*$ and $h\in\KK[x_1]\dots[x_e]^*$. In particular,
$\sigma(\tilde{g})=\tilde{f}_1^{m_1}\dots\tilde{f}_n^{m_n}\,\lambda\,\tilde{g}$ for $\lambda\in\KK^*$ being an $r$th root of unity. Hence we can take $m_{n+1},\dots,m_{n+s}\in\NN$ such that $\lambda=\alpha_1^{m_{n+1}}
\dots\alpha_s^{m_{n+s}}$. Consequently,
\begin{equation}\label{equ:M1Cond}
\sigma(\tilde{g})=\tilde{f}_1^{m_1}\dots\tilde{f}_n^{m_n}\alpha_1^{m_{n+1}}
\dots\alpha_s^{m_{n+s}}\,\tilde{g},
\end{equation}
which yields 
\begin{equation}\label{equ:M2Cond}
\sigma(h)=h_1^{m_1}\dots
h_n^{m_n}\alpha_1^{-m_{n+1}}\dots\alpha_s^{-m_{n+s}}\,h.
\end{equation}
Then~\eqref{equ:M1Cond} and~\eqref{equ:M2Cond} imply $(m_1,\dots,m_{n+s})\in
M_1\cap M_2$. Conversely, let $(m_1,\dots,m_n)\in M_1\cap M_2$. I.e., there are $m_i\in\NN$,
$\tilde{g}\in\FF^*$ and $h\in\KK[x_1]\dots[x_e]^*$ s.t.~\eqref{equ:M1Cond} and~\eqref{equ:M2Cond} hold. Therefore
$\sigma(\tilde{g}\,h)=f_1^{m_1}\dots f_n^{m_n}\,\tilde{g}\,h$
which implies that $(m_1,\dots,m_n)\in M(\vect{f},\EE)$.
\end{proof}

The following remarks are in place. By Corollary~\ref{Cor:SConstRExtension} it follows that $\sconst{G}{\EE}{\FF}\setminus\{0\}\leq\EE^*$  
and thus $M(\vect{f},\EE)$ in Lemma~\ref{Lemma:MProblemForRExtensions}  has a $\ZZ$-basis with rank $\leq n$. In particular, we can compute such a basis as follows.
First note that both $M_1$ and $M_2$ given in Lemma~\ref{Lemma:MProblemForRExtensions} have $\ZZ$-bases with rank $\leq n+s$: for $M_1$ this follows since $\FF$ is a field. 
Moreover, if one takes $\HH=\KK[x_1]\dots[x_e]\leq\EE$ and $H=\dgroup{(\KK^*)}{\HH}{\KK}$, it follows by Corollary~\ref{Cor:SConstRExtension} that $\sconst{H}{\HH}{\sigma}\setminus\{0\}\leq\HH^*$ and thus a $\ZZ$-basis exists with rank $\leq n+s$. Summarizing, we can determine a $\ZZ$-basis of $M(\vect{f},\EE)$ by using~\eqref{Equ:GetMRExt} if bases of $M_1$ and $M_2$ are available. 

\begin{example}\label{Exp:MBasisForDegreeB}
Take the \pisiSE-field $\dfield{\KK(k)}{\sigma}$ over $\KK=\QQ(\iota)$ with $\sigma(k)=k+1$ and consider the \rE-extension $\dfield{\KK(k)[x]}{\sigma}$ of $\dfield{\KK(k)}{\sigma}$ with $\sigma(x)=\iota\,x$ and $\ord(x)=4$ from Example~\ref{Exp:QkxyIsRExt}. In order to obtain a degree bound in Example~\ref{Exp:ProductBound} below, we need a basis of $M=M(\vect{f},\KK(k)[x])$ with $\vect{f}=(k x,-\frac{x}{k+1})$. Here we will apply Lemma~\ref{Lemma:MProblemForRExtensions}. By Example~\ref{Exp:GetOrder}.(3)
we get $\ford(x)=8$. With $u_1=1$ we determine $r=8$ by~\eqref{Equ:DefineRinMProblem}.
We define $\tilde{f_1}=k$, $\tilde{f}_2=-1/(k+1)$ and $h_1=h_2=x$. All 8th roots of unity of $\KK$ are generated by $\alpha_1=\iota$. For the activation of the above lemma, we have to determine a basis of $M_1=M((\tilde{f}_1,\tilde{f}_2,\alpha_1),\KK(k))=M((k,\frac{-1}{k+1},\iota),\KK(k)$. Here we use, e.g., the algorithms worked out in~\cite{Karr:81} (this is the base case of our machinery, see Subsection~\ref{Subsec:GroundFieldAlg}) and obtain the basis
$\{(1, 1, 2), (0, 0, 4)\}$. Moreover, we compute the basis $\{(1, 1, 0), (0, 2, 0), (0, 0, 1)\}$ of $M_2=M((h_1,h_2,\alpha_1),\KK[x])=M((x,x,\iota),\KK[x])$, for details see Example~\ref{Example:MBasisInKx} below. Thus a basis of $M_1\cap M_2$ is $\{(1,1,2),(0,0,4)\}$ and we get the basis $\{(1,1)\}$ of $M$.
\end{example}

\noindent By assumption (i.e., the base case in our recursion) a basis of $M_1$ can be determined. The calculation of a $\ZZ$-basis of $M_2$ can be accomplished by using the following proposition.

\begin{proposition}\label{Prop:ComputeBasesOfRM}
Let $\dfield{\HH}{\sigma}$ with $\HH=\KK[x_1]\dots[x_e]$ be a simple
\rE-extension of $\dfield{\KK}{\sigma}$ with a computable constant field $\KK$ and given $o_i=\ord(x_i)$ for $1\leq i\leq e$. Define $G=\dgroup{(\KK^*)}{\HH}{\KK}$ and
let
$\vect{f}=(f_1,\dots,f_n)\in G^n$ with given $\lambda_i:=\ord(f_i)>0$ for
$1\leq i\leq n$. Then a basis of $M(\vect{f},\HH)$ can be computed.
\end{proposition}
\begin{proof}
Define the finite sets
$$S:=\{(n_1,\dots,n_e)\in\NN^e|\, 0\leq n_i<o_i\}\text{ and }\tilde{M}:=\{(m_1,\dots,m_n)\in\NN^n|\, 0\leq m_i<\lambda_i\}.$$
Then loop
through all vectors $\vect{m}=(m_1\dots,m_n)\in\tilde{M}$
and check if there is a $g\in\HH^*$ with
$\sigma(g)=f_1^{m_1}\dots f_n^{m_n}\,g.$ 
More precisely, we can make the Ansatz $g=\sum_{\vect{i}\in S}
c_{\vect{i}}\,\vect{x}^{\vect{i}}$ which leads to a linear system of equations in
the $c_{\vect{i}}$ with coefficients from $\KK$. Solving this system gives the
solution space\footnote{By arguments as in the proof of Lemma~\ref{Lemma:MZModule} it follows $\dim(L)\leq 1$.} $L$ and we can check if the considered $\vect{m}$ from
$\tilde{M}$
is contained in $M(\vect{f},\HH)$. In this way we can generate the
subset $M'=\tilde{M}\cap M(\vect{f},\HH)$. Denote by $\vect{b_i}\in\KK^n$ the $i$th
unit vector. We show that
\begin{equation}\label{Equ:MGeneration}
\mspan(M'\cup\{\lambda_1\,\vect{b_1},\dots,\lambda_n\,\vect{b_n}\})=M(\vect{f},\HH).
\end{equation}
Namely, since $M(\vect{f},\HH)$ is a $\ZZ$-module (see the remarks above Example~\ref{Exp:MBasisForDegreeB}) and since $\lambda_i\,\vect{b_i}\in M(\vect{f},\HH)$, the left hand side is contained in
the
right hand side. Conversely, suppose that $(m_1,\dots,m_n)\in M(\vect{f},\HH)$.
Then let $m'_i=m_i\mod\lambda_i$, i.e., $0\leq m'_i< \lambda_i$ with
$m_i=m'_i+z_i\,\lambda_i$ for some $z_i\in\ZZ$. Thus
$(m_1,\dots,m_n)=(m'_1,\dots,m'_n)+(\lambda_1\,z_1,\dots,\lambda_n\,z_n)$ where
$(m'_1,\dots,m'_n)\in\tilde{M}$ and
$(\lambda_1\,z_1,\dots,\lambda_n\,z_n)=z_1\,(\lambda_1\,\vect{b_1})+\dots+z_n\,
(\lambda_n\,\vect{b_n})$. Consequently,
$(m_1,\dots,m_n)$ is an element of the left hand side
of~\eqref{Equ:MGeneration}. Since the number of vectors of the span on the left hand side is finite, we can derive a $\ZZ$-basis of~\eqref{Equ:MGeneration}.
\end{proof}

\begin{remark}\label{Remark:MCalculation}
A basis of $M(\vect{f},\HH)$ can be obtained more
efficiently as follows. We start with the $\ZZ$-module which is given by the
basis $B=\{\lambda_1\,\vect{b_1},\dots,\lambda_n\,\vect{b_n}\}$ where $\vect{b_i}\in\KK^n$ is the
$i$th unit vector. Now go through all elements from $\tilde{M}$.
Take the first element $\vect{m}$ from $\tilde{M}$. If it is in $\mspan(B)$ (this can be easily checked), proceed to the next element. Otherwise, if it is an element from $M(\vect{f},\HH)$
(for the check see the proof of Proposition~\ref{Prop:ComputeBasesOfRM}), put
it in $B$ and transform the set again to a $\ZZ$-basis. More precisely, if we compose the rows $\vect{b_i}$ to a matrix, it should yield a matrix in Hermite normal form. In this way, the membership tests for $\mspan(B)$ can be carried out efficiently within the continuing calculation steps.
We proceed until all elements of $\tilde{M}$ are visited and update step by step $B$ as described above. By construction we have that our $\mspan(B)$ equals the left hand side of~\eqref{Equ:MGeneration} and thus equals $M(\vect{f},\HH)$. We
remark that $B$ consists always of $n$ linearly independent vectors. However,
the $\ZZ$-span is more and more refined.
\end{remark}

\begin{example}[Cont.\ Ex.~\ref{Exp:MBasisForDegreeB}]\label{Example:MBasisInKx}
Take the \rE-extension $\dfield{\KK[x]}{\sigma}$ of $\dfield{\KK}{\sigma}$ with $\KK=\QQ(\iota)$, $\sigma(x)=\iota\,x$ and $\ord(x)=4$. We calculate a basis of $M(\vect{f},\KK[x])$ with $\vect{f}=(x,x,\iota)$ as presented in Remark~\ref{Remark:MCalculation}. We start with
$\{(4, 0, 0), (0, 4, 0), (0, 0, 4)\}$
whose rows form a matrix in Hermite normal form.
Now we go through all elements of $\tilde{M}$, say in the order 
\small
$$\tilde{M}=\{(1,0,0),(2,0,0),(3,0,0),(0,1,0),(0,2,0),(0,3,0),(0,0,1),(0,0,2),(0,0,3),(1,1,0),\dots\}.$$
\normalsize
Since $(1,0,0)\notin \mspan(B)$, we check if there is a $g\in\KK[x]\setminus\{0\}$ with $\sigma(g)=x^1\,x^0\,\iota^0\,g$: this is not the case. We continue with $(2,0,0)$. Here we have that $(2,0,0)\notin\mspan(B)$. Now we check if there is a $g\in\KK[x]\setminus\{0\}$ with $\sigma(g)=x^2\,x^0\,\iota^0\,g$. Plugging in $g=g_0+g_1\,x+g_2\,x^2+g_3\,x^3$ into $\sigma(g)=x^2\,g$ gives the constraint $(g_0-g_2)x^0+\iota x (g_1+\iota
g_3)+x^2(-g_0-g_2)+x^3(-g_1-ig_3)=0$ which leads to the solution $g=x+\iota\,x^3$.  A basis of $\mspan(B\cup\{(2,0,0)\})$ is $\{(2,0,0),(0,4,0),(0,0,4)\}$. Thus we update $B$ to $B=\{(2,0,0),(0,4,0),(0,0,4)\}$. We have $(3,0,0)\notin\mspan(B)$, but there is no $g\in\KK[x]\setminus\{0\}$ with $\sigma(g)=x^3\,g$. Similarly to $(1,0,0)$, also $(0,1,0)$ does not change $B$, and similarly to $(2,0,0)$, $(0,2,0)$ leads to the updated basis $B=\{(2,0,0),(0,2,0),(0,0,4)\}$. $(0,3,0)$ does not change $B$. However, for $(0,0,1)\notin\mspan(B)$ we find $g=x$ with $\sigma(g)=x^0\,x^0\,\iota^1\,g$ which yields $B=\{(2,0,0),(0,2,0),(0,0,1)\}$. We have that $(0,0,2),(0,0,3)\in\mspan(B)$. Now we consider $(1,1,0)\notin\mspan(B)$. We find $g=x+\iota\,x^3$ with $\sigma(g)=x^2\,g$ (as already above). Hence we update $B$ to $B=\{(1,1,0),(0,2,0),(0,1,0)\}$ (where the rows form a matrix in Hermite normal form). As it turns out, no further element from $\tilde{M}$ changes $B$. Thus the found $B$ is a basis of $M(\vect{f},\KK[x])$.
\end{example}

\ExternalProof{(Theorem~\ref{Thm:AlgMainResultFull}.(2)\label{Proof:AlgMainResultFull2})}{
By Lemma~\ref{Lemma:ReorderSimpleRPISI} we can reorder the generators of
the \rpisiSE-extension such that $\dfield{\bar{\EE}}{\sigma}$ is an $\FF^*$-simple
\rE-extension of $\dfield{\FF}{\sigma}$ and
$\dfield{\EE}{\sigma}$ is a $G$-simple \pisiSE-extension of
$\dfield{\bar{\EE}}{\sigma}$ with $G=\dgroup{(\FF^*)}{\bar{\EE}}{\FF}$ . Let 
$\bar{\EE}=\FF[x_1]\dots[x_e]$ with $u_i$, $\alpha_i$ and $\vect{f}\in G^n$ with
$\tilde{f}_i$ and $h_i$ as given in
Lemma~\ref{Lemma:MProblemForRExtensions}.
By assumption we can compute a basis of $M_1$ as given in
Lemma~\ref{Lemma:MProblemForRExtensions}. Since Problem~O is solvable in $\KK^*$, we can compute $o_i=\ord(x_i)$ and $\lambda_i=\ord(u_i)$ by Corollary~\ref{Cor:SimpleNestedRBasics}.(4). Thus we can use
Proposition~\ref{Prop:ComputeBasesOfRM} to compute a basis of $M_2$ as posed in
Lemma~\ref{Lemma:MProblemForRExtensions}, and we get a basis of~\eqref{Equ:GetMRExt}. Summarizing, we can solve Problem~PMT in
$\dfield{\bar{\EE}}{\sigma}$ for $G$. In particular, $\sconst{G}{\bar{\EE}}{\sigma}\setminus\{0\}\leq\bar{\EE}^*$ by Corollary~\ref{Cor:SConstRExtension}.
Hence by Theorem~\ref{Thm:ProblemMPiSi}
we can solve Problem~PMT for $\dfield{\EE}{\sigma}$ in $\dgroup{G}{\EE}{\bar{\EE}}$.
Since $\dgroup{G}{\EE}{\bar{\EE}}=\dgroup{(\FF^*)}{\EE}{\FF}$ by Lemma~\ref{Lemma:SimpleTowerIsSimple}, the theorem is proven. 
}

\noindent To this end, we work out the following shortcut, resp.\ refined version of Theorem~\ref{Thm:RPSCharacterization}.(3).

\begin{corollary}\label{Cor:MShortCutR}
Let $\dfield{\FF}{\sigma}$ be a strong constant-stable difference field with
constant
field $\KK$, and let $G\leq\FF^*$ with $\sconst{G}{\FF}{\sigma}\setminus\{0\}\leq\FF^*$.
Let $\dfield{\HH}{\sigma}$ with $\HH=\FF[x_1]\dots[x_r]$ be a $G$-simple
\rE-extension of
$\dfield{\FF}{\sigma}$ and let
$\dfield{\EE}{\sigma}$ be a $\dgroup{G}{\HH}{\FF}$-simple \pisiSE-extension of
$\dfield{\HH}{\sigma}$. 
(1) If $f\in\dgroup{G}{\HH}{\FF}$ with $\ord(f)>0$, then $f\in\dgroup{(\KK^*\cap G)}{\HH}{\FF}$.\\
(2) $M(\vect{f},\EE)=M(\vect{f},\KK[x_1]\dots[x_r])$ for any
$\vect{f}=(f_1,\dots,f_n)\in (\dgroup{G}{\HH}{\FF})^n$ with $\ord(f_i)>0$.\\
(3) Let $\alpha\in
\dgroup{G}{\HH}{\FF}$ with $\ord(\alpha)>0$. Then there is an \rE-extension
$\dfield{\EE[t]}{\sigma}$
of $\dfield{\EE}{\sigma}$ with $\frac{\sigma(t)}{t}=\alpha$ iff there is an
\rE-extension $\dfield{\KK[x_1]\dots[x_r][t]}{\sigma}$ of
$\dfield{\KK[x_1]\dots[x_r]}{\sigma}$ with $\frac{\sigma(t)}{t}=\alpha$. 
\end{corollary}
\begin{proof}
(1) Let $f\in
\dgroup{G}{\HH}{\FF}$, i.e., 
$f=\alpha\,x_1^{m_{1}}\dots x_r^{m_{r}}$ where $\alpha\in G$ and $m_{i}\in\NN$.
With $\ord(f)>0$ and
Corollary~\ref{Cor:SimpleNestedRBasics}.(3) we have that $\ord(\alpha)>0$. Since $\dfield{\FF}{\sigma}$ is strong constant-stable, $\alpha\in\KK^*$. Thus 
$\alpha\in\KK^*\cap G$ and hence $f\in\dgroup{(\KK^*\cap G)}{\HH}{\FF}$.\\
(2) Let $\vect{f}\in
(\dgroup{G}{\HH}{\FF})^n$ be given as above. By part~1,
$f_i=\alpha_i\,x_1^{m_{i,1}}\dots x_1^{m_{i,r}}$ where the $\alpha_i\in\KK$ are
roots
of unity and $m_{i,j}\in\NN$.
By Lemma~\ref{Lemma:ShuffleOverField} we may suppose that 
$\EE=\HH\lr{t_1}\dots\lr{t_k}[s_1]\dots[s_e]$ where the $t_i$ are \piE-monomials and the $s_i$ are \sigmaSE-monomials. By
Corollary~\ref{Cor:MShortCutSigma} we have
that $M(\vect{f},\EE)=M(\vect{f},\HH\lr{t_1}\dots\lr{t_k})$. Now let $(m_1,\dots,m_n)\in M(\vect{f},\HH\lr{t_1}\dots\lr{t_k})$. Then there is a 
$g\in\HH\lr{t_1}\dots\lr{t_k}\setminus\{0\}$
with 
\begin{equation}\label{Equ:sg=ug}
\sigma(g)=u\,g
\end{equation}
for some $u=a\,x_1^{\mu_1}\dots x_r^{\mu_r}$ with $\mu_i\in\NN$ and with
$a$ being a root of unity from $\KK$. By Corollary~\ref{Cor:SimpleNestedRBasics}.(3)
we get $\mu:=\ord(u)>0$; in addition we have that $\mu'=\ford(u)>0$.
By Theorem~\ref{Thm:rpisiSCONST}
it follows that $g=q\,t_1^{\nu_1}\dots t_k^{\nu_k}$ with
$q\in\sconst{G}{\HH}{\sigma}\setminus\{0\}$ and $\nu_i\in\ZZ$. Since $u^{\mu}=1$, it follows with~\eqref{Equ:sg=ug} that
$\sigma(g^{\mu})=g^{\mu}$. Now suppose that $g$ depends on $t_m$ with $1\leq m\leq k$ being maximal. Then $g^{\mu}$ depends also on $t_m$ which contradicts to
$\const{\HH\lr{t_1}\dots\lr{t_k}}{\sigma}=\const{\HH}{\sigma}$. Consequently $g=q\in\sconst{G}{\HH}{\sigma}\setminus\{0\}$.
By
Corollary~\ref{Cor:RExtStructureForField} it
follows that $g=\tilde{g}\,h$ with $h\in\KK[x_1]\dots[x_r]^*$ and
$\tilde{g}\in\FF^*$
with $\sigma(h)=\lambda\,u\,h$ where $\lambda\in\KK^*$ is a root of unity. 
Recall that $\mu'=\ford(u)>0$ and hence $\mu'':=\lcm(\mu',\ord(\lambda))>0$. Since 
$\sigma^{\mu''}(h)=h$ and $\dfield{\FF}{\sigma}$ is constant-stable, it follows that $h\in\KK^*$. Therefore $g\in\KK[x_1]\dots[x_r]^*$.
Summarizing, 
$(m_1,\dots,m_n)\in M(\vect{f},\KK[x_1]\dots[x_r])$ and we conclude
that
$M(\vect{f},\EE)\subseteq M(\vect{f},\KK[x_1]\dots[x_r])$. The other direction is immediate.\\ 
(3) The third part follows
by parts 1 and 2 of the corollary and Theorem~\ref{Thm:PiCharStrong}.
\end{proof}

\section{The algorithmic machinery III: Problem~PFLDE}\label{Sec:PFLDE}

We aim at proving Theorems~\ref{Thm:AlgMainResultRestricted}.(2) and~\ref{Thm:AlgMainResultFull}.(3), i.e., providing recursive algorithms that reduce Problem~PFLDE from a given \rpisiSE-extension to its ground ring (resp.\ field).
If we are considering single-rooted \rpisiSE-extensions (Theorem~\ref{Thm:AlgMainResultRestricted}.(2)), we rely heavily on the fact that for a given difference ring $\dfield{\GG}{\sigma}$ with constant field $\KK$ and given group $G\leq\GG^*$ we have that $\sconstF{G}{\GG}{\sigma}\setminus\{0\}\leq\GG^*$. This property allows us to assume that for any $\vect{f}\in \GG^n$ and any $u\in G$ the $\KK$-vector space $V=V(u,\vect{f},\dfield{\GG}{\sigma})$ has a basis with dimension $\leq n+1$; see Lemma~\ref{Lemma:VBasis}. In particular, our reduction algorithm is based on the assumption that there are algorithms available that solve Problems~PFLDE and PMT in $\dfield{\GG}{\sigma}$ for $G$. For general simple \rpisiSE-extensions over a strong constant-stable difference field $\dfield{\GG}{\sigma}$ (Theorem~\ref{Thm:AlgMainResultFull}.(3)) we need stronger properties: all what we stated above should hold not only for $\dfield{\GG}{\sigma}$ but must hold for $\dfield{\GG}{\sigma^l}$ with $l\geq1$. For the currently explored difference fields $\dfield{\GG}{\sigma}$ with these properties we refer to Subsection~\ref{Subsec:GroundFieldAlg}.

\subsection{A reduction strategy for \pisiSE-extensions}\label{Subsec:PSReduction}

In this subsection we present a reduction method for \pisiSE-extensions which can be summarized with the following theorem.

\begin{theorem}\label{Thm:PSReduction}
Let $\dfield{\AA}{\sigma}$ be a computable difference ring and let $G\leq\AA^*$ with  $\sconst{G}{\AA}{\sigma}\setminus\{0\}\leq\AA^*$. Let $\dfield{\AA\lr{t}}{\sigma}$ be a $G$-simple \pisiSE-extension of $\dfield{\AA}{\sigma}$.
\begin{enumerate}
 \item If $t$ is a \sigmaSE-monomial and Problem~PFLDE is solvable in $\dfield{\AA}{\sigma}$ for $G$, then Problem~PFLDE is solvable in $\dfield{\AA\lr{t}}{\sigma}$ for $\dgroup{G}{\AA\lr{t}}{\AA}$. 
 \item If $t$ is a \piE-monomial and Problems PFLDE and PMT are solvable in $\dfield{\AA}{\sigma}$ for $G$, then Problem~PFLDE is solvable in $\dfield{\AA\lr{t}}{\sigma}$ for $\dgroup{G}{\AA\lr{t}}{\AA}$.
\end{enumerate}
\end{theorem}

\noindent In the following let $\dfield{\AA\lr{t}}{\sigma}\geq\dfield{\AA}{\sigma}$
be a \pisiSE-extension as given in the theorem with $\sigma(t)=\alpha\,t+\beta$ where $\alpha\in G$ and $\beta=0$, or $\alpha=1$ and $\beta\in\AA$.
Furthermore, we define $\tilde{G}=\dgroup{G}{\AA\lr{t}}{\AA}$ and suppose that we are given a $u\in\tilde{G}$, i.e.,
\begin{equation}\label{Equ:PFLDEu}
u=v\,t^m,\quad\text{with }v\in G,\,m\in\ZZ,
\end{equation}
and an $\vect{f}=(f_1,\dots,f_n)\in\tilde{G}^n$.
By Theorem~\ref{Thm:PiPart2Strong} we have that $\sconst{\tilde{G}}{\AA\lr{t}}{\AA}\setminus\{0\}\leq\AA\lr{t}^*$ and hence by Lemma~\ref{Lemma:VBasis} a basis of $V(u,\vect{f},\AA\lr{t})$ with dimension $\leq n+1$ exists.\\
Subsequently, we will prove Theorem~\ref{Thm:PSReduction}, i.e., we will
work out a reduction strategy that provides a basis of $V(u,\vect{f},\AA\lr{t})$ under the assumption that one can solve Problem PFLDE in $\dfield{\AA}{\sigma}$ for $G$ if $t$ is a \sigmaSE-monomial, resp.\ Problems PMT and PFLDE in $\dfield{\AA}{\sigma}$ for $G$ if $t$ is a \piE-monomial.
The two main steps of this reduction  will be described in the following two Subsections~\ref{Subsection:DegreeBounds} and~\ref{Subsubsection:DegreeReduction}.

\subsubsection{Degree bounds}\label{Subsection:DegreeBounds}

The first essential step is to search for degree bounds: we will determine $a,b\in\ZZ$ such that
\begin{equation}\label{Equ:BoundV}
V(u,\vect{f},\AA\lr{t}_{a,b})=V(u,\vect{f},\AA\lr{t})
\end{equation}
holds; for the definition of the truncated set of (Laurent) polynomials see~\eqref{Equ:TruncatedRing}. For technical reasons we also require that the constraint
\begin{equation}\label{Equ:TechnicalCond}
\max(b,b+m)\geq\tilde{b}
\end{equation}
holds where $m$ and $\tilde{b}$ are given by~\eqref{Equ:PFLDEu} and \begin{equation}\label{Equ:abtilde}
\tilde{b}=\max(\deg(f_1)\dots,\deg(f_n)).
\end{equation}
The recovery of these bounds (see Lemmas~\ref{Lemma:SigmaDegBound} and~\ref{Lemma:PiWithMDegreeBound} below) is based on generalizations of ideas given in~\cite{Karr:81}; for further details and proofs in the setting of difference fields see also~\cite{Schneider:04b,Schneider:05b}.

If $t$ is a \sigmaSE-monomial, then $\AA\lr{t}=\AA[t]$ forms a polynomial ring, $\alpha=1$ and $\tilde{G}=G$; in particular we have $m=0$ in~\eqref{Equ:PFLDEu}. In this case, we can utilize the following lemma.

\begin{lemma}\label{Lemma:SigmaDegBound}
 Let $\dfield{\AA[t]}{\sigma}$ be a \sigmaSE-extension of
$\dfield{\AA}{\sigma}$ and let $G\leq\AA^*$ such that 
$\sconst{G}{\AA}{\sigma}\setminus\{0\}\leq\AA^*$ holds. Let $f\in\AA[t]$ and $u\in G$.
Then any
solution $g\in\AA[t]$ of
$\sigma(g)-u\,g=f$
is bounded by $\deg(g)\leq\max(\deg(f)+1,0)$.
\end{lemma}
\begin{proof}
Suppose there is a $g\in\AA[t]$ with
$\deg(g)>\max(\deg(f)+1,0)$. Thus by Lemma~\ref{Lemma:SigmaLemma} there is
a $\gamma\in\AA$ with $\sigma(\gamma)-\gamma=\sigma(t)-t$ which contradicts to
Theorem~\ref{Thm:RPSCharacterization}.(1).
\end{proof}

\noindent Thus we can set $a=0$ and $b=\max(\tilde{b}+1,0)$ to guarantee that~\eqref{Equ:BoundV} and~\eqref{Equ:TechnicalCond} hold.

\begin{example}[Cont.\ Ex.~\ref{ExpQkxysS}]\label{Exp:TeleSigmaBound}
Consider the \sigmaSE-extension $\dfield{\AA[S]}{\sigma}$ of $\dfield{\AA}{\sigma}$ with $\AA=\QQ(k)[x][y][s]$ and $\sigma(S)=S+\frac{-x\,y}{k+1}$ from Example~\ref{ExpQkxysS}.(2). As stated in Example~\ref{ExpQkxysS}.(3), we want to determine a $g\in\AA[S]$ with $\sigma(g)-g=f$ where $f=y\,k^2\,s$, i.e., we want to find a basis of $V(1,\vect{f},\AA[S])$ with $\vect{f}=(k^2 s y)\in\AA[S]^1$. Using Lemma~\ref{Lemma:SigmaDegBound} it follows that $\deg(g)\leq 1$. Consequently, $V(1,\vect{f},\AA[S])=V(1,\vect{f},\AA[S]_0^{1})$. Using our methods below (see Example~\ref{Exp:TeleSigmaDegRed}) we get the basis $\{(1,g),(0,1)\}$ with $g$ as given in~\eqref{Equ:DFTeleSol}.
\end{example}

If $t$ is a \piE-monomial, then $\AA\lr{t}=\AA\ltr{t}$ is a ring of Laurent polynomials and $\beta=0$.\\  
First suppose that $u\notin\AA$, i.e., $m\in\ZZ\setminus\{0\}$ as given in~\eqref{Equ:PFLDEu}.\\ 
If $f_i=0$ for all $i$, it is easy to see that $V(u,\vect{f},\AA\lr{t})=V(u,\vect{f},\{0\})$, i.e., $a=0$ and $b=-1$ fulfil the properties~\eqref{Equ:BoundV} and~\eqref{Equ:TechnicalCond}.\\ 
Otherwise, if not all $f_i$ are 0,
we can use the following fact; the proof is left to the reader.

\begin{lemma}\label{Lemma:SimpleBound}
Let $\dfield{\AA\lr{t}}{\sigma}$ be a \piE-extension of $\dfield{\AA}{\sigma}$.
Let
$v\in\AA^*$, $m\in\ZZ\setminus\{0\}$, $f=\sum_{i=\lambda}^{\mu} f_i\,t^i\in\AA\lr{t}$ with $\lambda,\mu\in\ZZ$ and
$g=\sum_{i=\tilde{\lambda}}^{\tilde{\mu}}g_i t^i\in\AA\lr{t}$ with $\tilde{\lambda},\tilde{\mu}\in\ZZ$ and
$g_{\tilde{\lambda}}\neq0\neq g_{\tilde{\mu}}$ such that
$\sigma(g)-v\,t^m\,g=f$.
Then $\max(\lambda,\lambda-m)\leq\tilde{\lambda}$ and $\tilde{\mu}\leq\min(\mu,\mu-m)$.
\end{lemma}
\noindent Namely, define
$$\tilde{a}=\min(\ldeg(f_1)\dots,\ldeg(f_n)).$$
Note that in this scenario we have that $\tilde{a},\tilde{b}\in\ZZ$; for the definition of $\tilde{b}$ see~\eqref{Equ:abtilde}.
Hence by setting $a=\tilde{a}$ and $b=\tilde{b}$, we can conclude with  Lemma~\ref{Lemma:SimpleBound} that~\eqref{Equ:BoundV} and~\eqref{Equ:TechnicalCond} hold.  

\noindent What remains to consider is the case $u\in G$ with $m=0$. Here we utilize

\begin{lemma}\label{Lemma:PiWithMDegreeBound}
Let $\dfield{\AA\lr{t}}{\sigma}$ be a \piE-extension of
$\dfield{\AA}{\sigma}$ with $G\leq\AA^*$ where
$\sconst{G}{\AA}{\sigma}\setminus\{0\}\leq\AA^*$ and
$\alpha=\sigma(t)/t\in G$. Let
$u\in
G$, $f=\sum_{i=\lambda}^{\mu} f_i\,t^i\in\AA\lr{t}$ and
$g=\sum_{i=\tilde{\lambda}}^{\tilde{\mu}}g_i t^i\in\AA\lr{t}$ with $g_{\tilde{\lambda}}\neq0\neq g_{\tilde{\mu}}$ and
\begin{equation}\label{Equ:PFLDEt0}
\sigma(g)-u\,g=f.
\end{equation}
If there is a $\nu\in\ZZ$ with
\begin{equation}\label{Equ:nuProperty}
\sigma(\gamma)=\alpha^{-\nu}\,u\,\gamma 
\end{equation}
for
some $\gamma\in\sconst{G}{\AA}{\sigma}\setminus\{0\}$, then $\nu$ is uniquely determined and
we have that $\min(\lambda,\nu)\leq\tilde{\lambda}$ and $\tilde{\mu}\leq\max(\mu,\nu)$. If there is not such a $\nu$, we have that $\lambda\leq\tilde{\lambda}$ and $\tilde{\mu}\leq
\mu$.
\end{lemma}
\begin{proof}
Suppose there is a $\nu\in\ZZ$ with~\eqref{Equ:nuProperty}
for some $\gamma\in\sconst{G}{\AA}{\sigma}\setminus\{0\}$. Take in addition,
$\tilde{\nu}\in\ZZ$ such that
$\sigma(\tilde{\gamma})=\alpha^{-\tilde{\nu}}\,u\tilde{\gamma}$ holds
for some $\tilde{\gamma}\in\sconst{G}{\AA}{\sigma}\setminus\{0\}$. Then
$\sigma(\gamma/\tilde{\gamma})=\alpha^{\tilde{\nu}-\nu}\,\gamma/\tilde{\gamma}
$. Since $t$ is a \piE-monomial
it follows by Theorem~\ref{Thm:RPSCharacterization}.(2) that $\nu=\tilde{\nu}$, i.e., $\nu$ is uniquely determined. 
Now suppose that there is an $i$ with $g_i\neq0$ where we have $i<\min(\lambda,\nu)$ or $i>\max(\mu,\nu)$. Then by
coefficient comparison in~\eqref{Equ:PFLDEt0} we get $\sigma(g_i)=u\,\alpha^{-i}\,g_i$
with $g_i\in\sconst{G}{\AA}{\sigma}\setminus\{0\}$. Consequently $\nu=i$, a contradiction.
Otherwise, suppose that there is not such a $\nu\in\ZZ$. Then by the same
arguments it follows that $\tilde{\lambda}<\lambda$ or $\tilde{\mu}>\mu$ is not possible, i.e.,
$\tilde{\lambda}\geq\lambda$ and $\tilde{\mu}\leq\mu$. This completes the proof.
\end{proof}

\noindent Therefore we derive the desired bounds as follows. First, solve Problem~PMT and compute a basis of 
$M((\alpha,u),\AA)$. Then given a basis, we can decide constructively if there is a $\nu\in\ZZ$ such that~\eqref{Equ:nuProperty} holds. If yes, take the uniquely determined $\nu$ and we can take
$a=\min(\tilde{a},\nu)$ and $b=\max(\tilde{b},\nu)$ to obtain~\eqref{Equ:BoundV} and~\eqref{Equ:TechnicalCond}. Otherwise, if there is not such a $\nu$, we can set $a=\min(\tilde{a},0)$ and $b=\max(\tilde{b},-1)$.

\begin{example}[Cont.\ Ex.~\ref{Exp:ProductExample}]\label{Exp:ProductBound}
Take the difference field $\dfield{\KK(k)[x]\lr{t}}{\sigma}$ with $\alpha=\sigma(t)/t=x\,k$ defined in Example~\ref{Exp:ProductExample}. In order to find the identity~\eqref{Equ:ProductId}, we need a basis of $V=V(u,(0),\KK(k)[x]\lr{t})$ with $u=\frac{-x}{k+1}\in G:=\dgroup{(\KK(k)^*)}{\KK(k)[x]\lr{t}}{\KK(k)}$; note that $G\leq\KK(k)[x]\lr{t}^*$ with $\sconst{G}{\KK(k)[x]\lr{t}}{\sigma}\setminus\{0\}\leq\KK(k)[x]\lr{t}^*$.
In this setting we apply Lemma~\ref{Lemma:PiWithMDegreeBound}. I.e., we compute a basis of $M((\alpha,u),\KK(k)[x])=M((k\,x,\frac{-x}{(k+1)}),\KK(k)[x])$. As worked out in Example~\ref{Exp:MBasisForDegreeB}, a basis is $\{(1,1)\}$. Thus we find $\nu=-1$ such that there is a $g\in\KK(k)[x]\setminus\{0\}$ with~\eqref{Lemma:PiWithMDegreeBound}. We conclude that $V=V(u,(0),\KK(k)[x]\lr{t}_{-1}^{-1})$. Using our methods below (see Example~\ref{Equ:ProductDegreeRed}) we arrive at the basis $(0,x(\iota+x^2)/k/t),(1,0)\}$ of $V$. 
\end{example}

\noindent Summarizing, we obtain bounds $a,b\in\ZZ$ such that~\eqref{Equ:BoundV} and~\eqref{Equ:TechnicalCond} hold. For \piE-monomials we rely on the extra assumption that Problem~PMT is solvable in $\dfield{\AA}{\sigma}$ for $G$. 

\subsubsection{Degree reduction}\label{Subsubsection:DegreeReduction}

The following degree reduction has been introduced in~\cite{Karr:81} in the setting of difference fields. Subsequently, we present the basic ideas in the setting of difference rings; further technical details can be found in~\cite[Thm.~3.2.2]{Schneider:01} and~\cite{Schneider:05a,Schneider:14}.

\noindent We want to determine all $c_1,\dots,c_n\in\KK=\const{\AA}{\sigma}$ and $g_i\in\AA$ in $g=\sum_{i=a}^bg_i\,t_i$ such that the following parameterized equation holds:
\begin{equation}\label{Equ:PFLDEAlg}
\sigma(g)-u\,g=c_1\,f_1+\dots+c_n\,f_n.
\end{equation}
If $b<a$, we are in the base case: $g=0$ and a basis of $V(u,\vect{f},\AA\lr{t})=V(u,\vect{f},\{0\})$ can be determined by linear algebra.\\ 
Otherwise, we continue as follows.
Due to~\eqref{Equ:TechnicalCond}, it follows that $\lambda:=\max(b,b+m)$ is the highest possible exponent in~\eqref{Equ:PFLDEAlg}. Let $\tilde{f}_i$ be the coefficient of the term $t^{\lambda}$ in $f_i$.
Then by coefficient comparison w.r.t.\ $t^{\lambda}$ in~\eqref{Equ:PFLDEAlg} we get the following constraints:
\begin{description}
\item[$\textnormal{If }m>0$\textnormal{,}]
\begin{equation}\label{Equ:PFLDESub1}
-v\,g_b=c_1\,\tilde{f_1}+\dots+c_n\,\tilde{f}_n;
\end{equation}
\item[$\textnormal{if }m=0$\textnormal{,}]
\begin{equation}\label{Equ:PFLDESub2}
\alpha^b\,\sigma(g_b)-v\,g_b=c_1\,\tilde{f_1}+\dots+c_n\,\tilde{f}_n;
\end{equation}
\item[$\textnormal{if }m<0$\textnormal{,}]
\begin{equation}\label{Equ:PFLDESub3}
\alpha^b\,\sigma(g_b)=c_1\,\tilde{f_1}+\dots+c_n\,\tilde{f}_n.
\end{equation}
\end{description}
For the cases $m>0$ and $m<0$ one can easily determine a basis of the $\KK$-vector spaces $\{(c_1,\dots,c_n,g_m)|\,\eqref{Equ:PFLDESub1}\text{ holds}\}$ and $\{(c_1,\dots,c_n,g_m)|\,\eqref{Equ:PFLDESub3}\text{ holds}\}$ by linear algebra. Moreover, if $m=0$, equation~\eqref{Equ:PFLDESub2} can be written in the form
$\sigma(g_b)-v\,\alpha^{-b}\,g_b=c_1\,\tilde{f_1}\alpha^{-b}+\dots+c_n\,\tilde{f}_n\alpha^{-b}$
where $v\,\alpha^{-b}\in G$ and $\tilde{f}_i\alpha^{-b}\in\AA$. Thus a basis of 
\begin{equation}\label{Equ:IncrementalSolSpace}
V(v\,\alpha^{-b},(\tilde{f_1}\,\alpha^{-b},\dots,\tilde{f}_n\,\alpha^{-b}),\AA)
\end{equation}
can be determined under our assumption that one can solve Problem~PFLDE in $\dfield{\AA}{\sigma}$ for $G$. 
Now we plug in this partial solution (i.e., the possible leading coefficient $g_b$ 
with the corresponding linear combinations of the $f_i$), and end up at a new first-
order parameterized difference equation where the highest possible coefficient is $\lambda-1$. In other words, we reduced the problem by \textit{degree reduction}.
We continue to search for the next highest coefficient $g_{b-1}$. Hence we proceed recursively by updating $\lambda\to\lambda-1$ and $b\to b-1$ and determine a basis of 
the reduced problem (with highest degree $\lambda-1$). Finally, given a basis of this solution space and given the basis of~\eqref{Equ:IncrementalSolSpace}, one can determine a basis of $V(u,\vect{f},\AA\lr{t}_{a,b})$; for further technical details we refer to~\cite[Thm~12]{Karr:81} or~\cite[Section~3.1]{Schneider:14}.\\ 
Summarizing, solving various instances of Problem~PFLDE with the degree reductions $b\to b-1\to\dots\to a-1$ leads to the base case and we eventually produce a basis of $V(u,\vect{f},\AA\lr{t})$. This concludes the proof of Theorem~\ref{Thm:PSReduction}.

\begin{example}[Cont.~Ex.~\ref{Exp:ProductBound}]\label{Equ:ProductDegreeRed}
We know that $g=g_{-1}t^{-1}$. Plugging in $g$ into $\sigma(g)+\frac{x}{k+1}=0$ yields 
$\sigma(g_{-1})+\frac{x^2\,k}{k+1}\,g_{-1}=0.$
Therefore we look for a basis of $V(\frac{-x^2\,k}{k+1},(0),\KK(k)[x])$. By using the algorithms presented in Subsection~\ref{Subsec:PFLDERExtension} we get the basis $\{(1,x(\iota+x^2)/k),(0,1)\}$. This finally gives the basis $(0,x(\iota+x^2)/k/t),(1,0)\}$ of $V(\frac{-x}{k+1},(0),\KK(k)[x]\lr{t})$.
\end{example}

\begin{example}[Cont.~Ex.~\ref{Exp:TeleSigmaBound}]\label{Exp:TeleSigmaDegRed}
We want to find a basis of $V=V(1,\vect{f},\AA[S]_0^{1})$ with $\AA=\QQ(k)[x][y][s]$ and $\vect{f}=(y\,k^2\,s)$. Hence we make the Ansatz $(c_1,g_0+g_1\,S)\in V$ with the indeterminates $c_1\in\QQ$ and $g_0,g_1\in\AA$ such that 
\begin{equation}\label{Exp:DegRed1}
\sigma(g_0+g_1\,S)-(g_0+g_1\,S)=c_1\,y\,k^2\,s
\end{equation}
holds. Doing coefficient comparison w.r.t.\ $S^1$ yields the constraint
$\sigma(g_1)-g_1=c_1\,0$; compare~\eqref{Equ:PFLDESub2}. Thus we get all solutions by determining a basis of 
$V(1,\vect{\tilde{f}},\QQ(k)[x][s])$ with $\vect{\tilde{f}}=(0)\in\AA^1$. In this particular instance, the $\QQ$-basis $\{(1,0),(0,1)\}$ is immediate utilizing the fact that the constants are precisely $\QQ$. Summarizing, the solutions are $(c_1, g_1)\in\QQ^2$. 
Consequently, our Ansatz can be refined with $(c_1,g_0+c_2\,S)\in V$ where $c_1\in\QQ$, $c_2(=g_1)\in\QQ$ and $g_0\in\AA$ such that $\sigma(g_0+c_2\,S)-(g_0+c_2\,S)=c_1\,k^2 s y$ holds. Bringing the $c_2\,S$ part to the right hand side yields the new equation\footnote{
Note that we reduced the problem to find a polynomial solution of~\eqref{Exp:DegRed1} with maximal degree $1$ to a polynomial solution of~\eqref{Exp:DegRed1} with maximal degree $0$. This degree reduction has been achieved by introducing an extra parameter $c_2$. In general, the more \sigmaSE-monomials are involved, the more parameters will be introduced within the proposed degree reduction.}
\begin{equation}\label{Exp:DegRed2}
\sigma(g_0)-g_0=c_1\,y\,k^2\,s-c_2\,h
\end{equation}
with $h=\sigma(S)-S=\frac{-x\,y}{k+1}\in\AA$. In other words, we need a basis of $V(1,\vect{h},\AA)$ with $\vect{h}=(y\,k^2\,s,\frac{x\,y}{k+1})\in\AA^2$.
Now we apply again the reduction method, but this time in the smaller ring $\AA$ without the \sigmaSE-monomial $S$. We skip all the details, but refer to a particular subproblem that we will consider in Example~\ref{Exp:TeleSigmaSubProb}.
Finally, we get the basis
$$\{(0 , 0 , 1),(1,\tfrac{1}{2},
   \big(\tfrac{1}{4} (1-2
   k)-\tfrac{1}{4}x\big) y+s
   \big(\tfrac{1}{2} (k-1)
   (k+1) x-\tfrac{1}{2} (k-2)
   k\big) y\}$$
of $V(1,\vect{h},\AA)$.
Thus we can reconstruct the basis 
$\{(1,g),(0,1)\}$ of $V$ with $g$ given in~\eqref{Equ:DFTeleSol}. 
\end{example}

\noindent Note that the reduction of Theorem~\ref{Thm:PSReduction} simplifies to Karr's field version given in~\cite{Karr:81} if one specializes $\AA$ to a field and sets $G=\AA^*=\AA\setminus\{0\}$. 
However, the presented version works not only for a field, but for any difference ring $\dfield{\AA}{\sigma}$ as specified in Theorem~\ref{Thm:PSReduction}.
Subsequently, we will exploit this enhancement in order to treat (nested) \rE-extensions.

\subsection{A reduction strategy for \rE-extensions and thus for \rpisiSE-extensions}\label{Subsec:PFLDERExtension}

In order to treat simple and single-rooted \rpisiSE-extensions (Theorem~\ref{Thm:AlgMainResultRestricted}.(2)), we utilize the following proposition.

\begin{proposition}\label{Prop:SpecialRReduction}
Let $\dfield{\AA}{\sigma}$ be a computable difference ring with $G\leq\AA^*$ and $\sconst{G}{\AA}{\sigma}\setminus\{0\}\leq\AA^*$.
Let $\dfield{\AA[t]}{\sigma}$ be an \rE-extension of $\dfield{\AA}{\sigma}$ of given order $d$ with $\frac{\sigma(t)}{t}\in G$.\\ 
Then Problem~PFLDE is solvable in $\dfield{\AA[t]}{\sigma}$ for $G$ if it is solvable in $\dfield{\AA}{\sigma}$ for $G$
\end{proposition}
\begin{proof}
The proof follows by an adapted degree reduction presented in the proof of Theorem~\ref{Thm:PSReduction}; see Subsection~\ref{Subsubsection:DegreeReduction}. Let $u\in G$ and $\vect{f}=(f_1,\dots,f_n)\in\AA[t]^n$. 
By definition, it follows that a solution $g\in\AA[t]$ and $c_1,\dots,c_n\in\KK=\const{\AA}{\sigma}$ of~\eqref{Equ:PFLDEAlg} is of the form $g=\sum_{i=a}^b g_i t^i$ with $a:=0$ and $b:=d-1$. Thus the bounds are immediate (under the assumption that $d$ has been determined; see Section~\ref{Sec:Period}). Since $\sum_{i=0}^{d-1} h_i\,t^i=\sum_{i=0}^{d-1} \bar{h}_i\,t^i$ iff $h_i=\bar{h}_i$, we can activate the degree reduction as outlined in Subsection~\ref{Subsubsection:DegreeReduction}. Namely, by coefficient comparison of the highest term we always enter in the case~\eqref{Equ:PFLDESub2} (note that $m=0$ in~\eqref{Equ:PFLDEu}). By assumption we can solve Problem PFLDE in $\dfield{\AA}{\sigma}$ for $G$ and thus we can determine a basis of~\eqref{Equ:IncrementalSolSpace}. By recursion we finally obtain a basis of $V(u,\vect{f},\AA[t])$. 
\end{proof}

\ExternalProof{(Theorem~\ref{Thm:AlgMainResultRestricted}.(2)\label{Proof:AlgMainResultRestricted2})}{
Since Problem~PFLDE is solvable in $\dfield{\GG}{\sigma}$ for $G$, it follows by iterative applications of Theorem~\ref{Thm:PSReduction} and Corollary~\ref{Cor:pisiSCONST}.(1) that Problem~PFLDE is solvable in $\dfield{\HH}{\sigma}$ for $\tilde{G}$ with $\HH=\GG\lr{t_1}\dots\lr{t_r}$ and that $\sconst{\tilde{G}}{\HH}{\sigma}\setminus\{0\}\leq\HH^*$.
Thus by iterative applications of Propositions~\ref{Prop:SpecialRReduction} and~\ref{Prop:RPart2Weak} 
we conclude that Problem~PFLDE is solvable in $\dfield{\bar{\HH}}{\sigma}$ for $\tilde{G}$ with $\bar{\HH}=\HH\lr{x_1}\dots\lr{x_u}$ and that $\sconst{\tilde{G}}{\bar{\HH}}{\sigma}\setminus\{0\}\leq\bar{\HH}^*$. Finally, by applying iteratively Theorem~\ref{Thm:PSReduction} and Corollary~\ref{Cor:pisiSCONST}.(1) it follows that Problem~PFLDE is solvable in $\dfield{\EE}{\sigma}$ for $\tilde{G}$. 
Note that in Proposition~\ref{Prop:SpecialRReduction} we have to know the values $\ord(x_i)=\ord(\alpha_i)$ with $\alpha_i\in G$  (either as input or by computing them first by solving instances of Problem~O in $G$).
}

\noindent Finally, we present the underlying reduction method for simple \rpisiSE-extensions (Theorem~\ref{Thm:AlgMainResultFull}.(3)) which is based on the following lemma and proposition.

\begin{lemma}\label{Lemma:PFLDE1ShiftTonShift}
 Let $\dfield{\AA}{\sigma}$ be a difference ring, $f\in\AA$, $u\in\AA^*$ and
$\lambda\in\NN\setminus\{0\}$. Then $\sigma(g)-u\,g=f$
implies that
\begin{equation}\label{Equ:nShiftTelescoping}
\sigma^{\lambda}(g)-\sigmaFac{u}{\lambda}\,g=\sum_{j=0}^{\lambda-1}\frac{\sigmaFac{u}{\lambda}}{
\sigmaFac{u}{j+1}}\,\sigma^j(f).
\end{equation}
\end{lemma}
\begin{proof}
From $\sigma(g)-u\,g=f$ we get
$\sigma^{j+1}(g)-\sigma^j(u)\sigma^j(g)=\sigma^j(f)$ for all $j\in\NN$. Multiplying it with $\sigmaFac{u}{\lambda}/\sigmaFac{u}{j+1}$ yields 
$\frac{\sigmaFac{u}{\lambda}}{\sigmaFac{u}{j+1}}\sigma^{j+1}(g)-\frac{\sigmaFac{u}
{\lambda}}{\sigmaFac{u}{j}}\sigma^j(g)=\frac{\sigmaFac{u}{\lambda}}{\sigmaFac{u}{j+1}}
\sigma^j(f).$
Summing this equation over $j$ from $0$ to $\lambda-1$
produces~\eqref{Equ:nShiftTelescoping}.
\end{proof}

\begin{proposition}\label{Prop:LiftForNestedRExt}
Let $\dfield{\AA}{\sigma}$ be a constant-stable and computable difference ring with constant field $\KK$. Let $G\leq\AA^*$ be closed
under $\sigma$ with $\sconstF{G}{\AA}{\sigma^l}\setminus\{0\}\leq\AA^*$ for all $l>0$.
Let $\dfield{\EE}{\sigma}$ with $\EE=\AA\lr{x_1}\dots\lr{x_r}$ be a $G$-simple
\rE-extension of $\dfield{\AA}{\sigma}$ where $\ord(x_i)>0$ and
$\per(x_i)>0$ for $1\leq i\leq r$
are given and where
$\sconst{(\dgroup{G}{\EE}{\AA})}{\EE}{\sigma}\setminus\{0\}\leq\EE^*$.\\
If Problem~PFLDE is solvable in $\dfield{\GG}{\sigma^l}$ for $G$ for all $l>0$, it is solvable in
$\dfield{\EE}{\sigma}$ for
$\dgroup{G}{\EE}{\AA}$.
\end{proposition}
\begin{proof}
Let $\KK=\const{\AA}{\sigma}$, let $\EE=\AA\lr{x_1}\dots\lr{x_r}$, let
$\vect{f}=(f_1,\dots,f_n)\in\EE^n$ and let
$u=v\,x_1^{m_1}\dots x_r^{m_r}\in\dgroup{G}{\EE}{\AA}$ with $v\in G$ and
$m_i\in\NN$. We will present a reduction method to obtain a basis of
$V(u,\vect{f},\EE)$. Set $\alpha:=x_1^{m_1}\dots x_r^{m_r}$. Then by Lemma~\ref{Lemma:DROrder} it follows that $\ord(\alpha)>0$ can be computed by the given values of $\ord(x_i)$ with $1\leq i\leq r$. Hence we can activate Lemma~\ref{Lemma:DRPeriod}.(4) and can compute $\ford(\alpha)>0$. 
Now take 
\begin{equation}\label{Equ:DefineLambdaForPFLDE}
\lambda=\lcm(\ford(\alpha),\per(x_1),\dots,\per(x_r)).
\end{equation}
Thus we have that\footnote{By a mild modification of the proof it suffices to take a  $\lambda$ such that $\sigmaFac{\alpha}{\lambda}\in\const{\AA}{\sigma}$ holds.} $\sigmaFac{\alpha}{\lambda}=1\text{ and }\sigma^{\lambda}(x_i)=x_i$
for all $1\leq i\leq r$. Finally, define
\begin{equation}\label{Equ:definew}
w:=\sigmaFac{u}{\lambda}=\sigmaFac{(\alpha\,v)}{\lambda}=\sigmaFac{v}{\lambda}\in G.
\end{equation}
Now let $(c_1,\dots,c_n,g)\in V(u,\vect{f},\EE)$, i.e., we have that~\eqref{Equ:PFLDEAlg}. Thus
Lemma~\ref{Lemma:PFLDE1ShiftTonShift} 
yields
\begin{equation}\label{Equ:ModPFLDE}
\sigma^{\lambda}(g)-w\,g=c_1\,\tilde{f}_1+\dots+c_n\,\tilde{f}_n
\end{equation}
with 
\begin{equation}\label{Equ:tildefCalc}
\tilde{f}_i=\sum_{j=0}^{\lambda-1}\frac{\sigmaFac{u}{{\lambda}}}{\sigmaFac{u}{j+1}}
\sigma^j(f_i).
\end{equation}
Hence $V(u,\vect{f},\EE)$ is a subset of 
\begin{equation}\label{equ:tildeV}
\tilde{V}=\{(c_1,\dots,c_n,g)\in\KK^n\times\EE|\,\text{\eqref{Equ:ModPFLDE} holds}\}.
\end{equation}
Note that $\tilde{V}$ is a $\KK$-subspace of $\KK^n\times\EE$. Thus $V(u,\vect{f},\EE)$ is a subspace of $\tilde{V}$ over $\KK$.    
First, we show that $\tilde{V}$ has a finite basis and show how one can compute it. For this task define
$S:=\{(n_1,\dots,n_r)\in\NN^r|\, 0\leq n_i<\ord(x_i)\}.$
Write $g=\sum_{\vect{i}\in S}g_{\vect{i}}\vect{x}^{\vect{i}}$ and
$\tilde{f}_j=\sum_{\vect{i}\in
S}\tilde{f}_{j,\vect{i}}\,\vect{x}^{\vect{i}}$
in multi-index notation. Since
$\sigma^{\lambda}(x_i)=x_i$, it follows by coefficient comparison that for $\vect{i}\in S$ we
have that
$$\sigma^{\lambda}(g_{\vect{i}})-w\,g_{\vect{i}}=c_1\,\tilde{f}_{1,\vect{i}}+\dots+c_n\,
\tilde{f}_{n,\vect{i}}.$$
By assumption, $\sconstF{G}{\AA}{\sigma^{\lambda}}\setminus\{0\}\leq\AA^*$. In particular,
since $\dfield{\AA}{\sigma}$ is constant-stable, we have that $\constF{\AA}{\sigma^{\lambda}}=\KK$. Thus with our $w\in G$ and 
$\vect{\tilde{f}}_\vect{i}=(\tilde{f}_{1,\vect{i}},\dots,\tilde{f}_{n,\vect{i}}
)\in\AA^n$ we can solve
Problem~PFLDE in $\dfield{\AA}{\sigma^{\lambda}}$ with constant field $\KK$. Hence we get for all $\vect{i}\in S$ the bases for
\begin{equation}\label{Equ:Vi}
V_{\vect{i}}=V(w,\vect{\tilde{f}}_\vect{i},\dfield{\AA}{\sigma^{\lambda}}).
\end{equation}
Note that by construction it follows that $\tilde{V}$ from~\eqref{equ:tildeV} is given by
\begin{equation}\label{Equ:CombineVi}
\tilde{V}
=\{(c_1,\dots,c_n,\sum_{\vect{i}\in S} g_{\vect{i}}\,\vect{x}^{\vect{i}})|\,
(c_1,\dots,c_n,g_{\vect{i}})\in
V_{\vect{i}}\}.
\end{equation}
Thus by linear algebra we get a basis of~\eqref{Equ:CombineVi}, say
$\vect{b_1},\dots,\vect{b_s}\in\KK^n\times\EE$. 
Recall that $V(u,\vect{f},\dfield{\EE}{\sigma})$ is a $\KK$-subspace of~\eqref{Equ:CombineVi}. 
To this end, we make the Ansatz
$(c_1,\dots,c_n,g)=d_1\,\vect{b_1}+\dots+d_s\,\vect{b_s}$ for indeterminates
$d_1,\dots,d_s\in\KK$ and plug in the generic solution 
into~\eqref{Equ:PFLDEAlg}. This yields another linear system with
unknowns $(d_1,\dots,d_s)$. Solving this system enables one to derive a basis of  $V(u,\vect{f},\EE)$.
\end{proof}

\begin{example}[Cont.\ Ex.~\ref{Exp:TeleSigmaDegRed}]\label{Exp:TeleSigmaSubProb}
In order to compute a basis of $V(1,\vect{h},\AA)$ in Ex.~\ref{Exp:TeleSigmaDegRed}, the recursive reduction enters in the following subproblem.
We are given the \rE-extension $\dfield{\QQ(k)[x]}{\sigma}$ of $\dfield{\QQ(k)}{\sigma}$ with $\sigma(x)=-x$ and need a basis of $V(x,\vect{f},\QQ(k)[x])$ with  
$\vect{f}=(f_1,f_2,f_3)=(\frac{(k^2-1)x}{2 k}+\frac{-k^2-2 k}{2 k},-\frac{x}{k},0)$.
By Example~\ref{Exp:GetOrder}.(2) we get $\per(x)=2$ and $\ford(x)=4$.
Hence using~\eqref{Equ:DefineLambdaForPFLDE} we determine 
$\lambda=\lcm(\ford(x),\per(x))=4$.
Using~\eqref{Equ:tildefCalc} with $u=x$ yields
$(\tilde{f}_1,\tilde{f}_2,\tilde{f}_3)=(-\frac{2 k^2+4 k+1}{k (k+2)}
-\frac{x}{(k+1) (k+3)}
,-\frac{2 x}{(k+1) (k+3)}
-\frac{2}{k (k+2)}
,0)$. Next we write the entries in multi-index notation. Namely, with $S=\{(0),(1)\}\subseteq\NN^1$ we get
\begin{align*}
\vect{\tilde{f}}_{(0)}&=(\tilde{f}_{1,(0)},\tilde{f}_{2,(0)},\tilde{f}_{3,(0)})=(-\frac{2 k^2+4 k+1}{k (k+2)},-\frac{2}{k (k+2)},0)\\
\vect{\tilde{f}}_{(1)}&=(\tilde{f}_{1,(1)},\tilde{f}_{2,(1)},\tilde{f}_{3,(1)})=(-\frac{1}{(k+1) (k+3)},-\frac{2}{(k+1)(k+3)},0)
\end{align*}
with $\vect{f}=\sum_{(m)\in S}\vect{\tilde{f}}_{(m)}x^{m}=\vect{\tilde{f}}_{(0)}+\vect{\tilde{f}}_{(1)}\,x$. Following~\eqref{Equ:definew} we get $w=1$ and we have to compute bases of the~\eqref{Equ:Vi} with $\vect{i}\in S$. Here we obtain the basis $\{(0 , 0 , 1 , 0), (1 , -\frac{1}{2} , 0 , -k/2)\}$ of $V_{(0)}=V(1,\vect{\tilde{f}}_{(0)},\dfield{\QQ(k)}{\sigma^4})$ and the basis $\{(-1 , \frac{1}{2} , 0 , 0), (0 , 0 , 0 , 1), (0 , 0 , 1 , 0)\}$ of $V_{(1)}=V(1,\vect{\tilde{f}}_{(0)},\dfield{\QQ(k)}{\sigma^4})$.
Therefore a basis of 
\begin{align*}
\tilde{V}=&\{(c_1,c_2,c_3,g)\in\QQ^3\times\QQ(k)[x]|\sigma^4(g)-g=c_1\,\tilde{f_1}+c_2\,\tilde{f_2}+c_3\,\tilde{f_3}\}\\
=&\{(c_1,c_2,c_3,\sum_{(i)\in \{(0),(1)\}} g_{\vect{i}}\,x^i)|\,
(c_1,c_2,c_3,g_{(i)})\in
V_{(i)}\}
\end{align*}
can be read off: $\{(1,-1/2,0,-k/2),(0,0,1,0),(0,0,0,1)\}$. Since $V(x,\vect{f},\QQ(k)[x])$ is a $\QQ$-subspace of $\tilde{V}$, we plug in
$(c_1,c_2,c_3,g)=d_1 (1,-1/2,0,-\frac{k}2)+d_2(0,0,1,0)+d_3 (0,0,0,x)+d_4(0,0,0,1)$ with unknowns $d_1,d_2,d_3,d_4\in\QQ$ into~\eqref{Equ:PFLDEAlg}. Together with our given $f_i$ and $u$ we get the linear constraint $\frac{1}{2} (d_1-2 d_3+2 d_4)+x (-d_3-d_4)=0$ or equivalently the linear constraints $-d_3-d_4=0$ and $\frac{1}{2} (d_1
-2 d_3+2 d_4)$. This yields 
$d_3 =\frac{d_1}{4}$ and $d_4=-\frac{d_1}{4}$. Thus we obtain the generic solution
$d_1 (1,-\frac{1}{2},0,-\frac{k}{2}
+\frac{x}{4}
-\frac{1}{4}) + d_2(0,0,1,0)$ of $V(x,\vect{f},\QQ(k)[x])$, i.e., the basis $\{(1,-\frac{1}{2},0,-\frac{k}{2}
+\frac{x}{4}
-\frac{1}{4}), (0,0,1,0)\}$ of $V(x,\vect{f},\QQ(k)[x])$. 
\end{example}

\begin{remark}
(1) In the underlying algorithm of Proposition~\ref{Prop:LiftForNestedRExt} we construct for all $\vect{i}\in S$ the solution spaces given in~\eqref{Equ:Vi} and combine them in one stroke as proposed in~\eqref{Equ:CombineVi}.
This approach is interesting if one wants to perform calculations in parallel. Another approach is to apply similar tactics as given in Subsection~\ref{Subsec:PSReduction}: compute a basis of one of the~\eqref{Equ:Vi}, plug in the found solutions and continue with an updated Ansatz in terms of the remaining monomials. In this way, one usually shortens step by step the length of the vectors $\vect{\tilde{f}_i}$ in~\eqref{Equ:Vi} and ends up very soon at a trivial situation (shortcut). \\
(2) A different approach is to consider an \rE-extension $\dfield{\FF[x]}{\sigma}$ of $\dfield{\FF}{\sigma}$ of order $d$ as a holonomic expression~\cite{Zeilberger:90a,Chyzak:00,Koutschan:13} over a difference field. Then as worked out in~\cite{Schneider:05d,Erocal:11}, a solution $g=\sum_{i=0}^{d-1}g_i\,x^ i$ and $c_i\in\const{\FF}{\sigma}$ of~\eqref{Equ:ParaTeleDR} leads to a coupled system of first-order difference equations in terms of the $g_i$ that can be uncoupled explicitly. More precisely, there is an explicitly given formula that constitutes a higher-order parameterized linear difference equation in $g_{d-1}$ and the parameters $c_i$. Solving this difference equation in terms of $g_{d-1}$ and the $c_i$ delivers automatically the remaining coefficients $g_i$, i.e., the solution $g$ of~\eqref{Equ:ParaTeleDR}. Here one usually has to solve a general higher-order linear difference equation.
For further details on the holonomic Ansatz
in the context of algebraic ring extensions (also on handling such objects in the basis of idempotent elements~\cite{Singer:97,Singer:08}) we refer to~\cite{Erocal:11}. The advantage of the reduction technique proposed in Proposition~\ref{Prop:LiftForNestedRExt} is that it can be applied in one stroke for nested \rE-extensions. In particular, Problem~PFLDE can be always reduced again to Problem~PFLDE by possibly switching to $\dfield{\FF}{\sigma^k}$ for some $k>1$. In this way, general higher-order linear difference equations can be avoided. \end{remark}

Combining all algorithmic parts of this article we obtain the following result.

\begin{theorem}\label{Thm:PFLDEStrongConstantStable}
Let $\dfield{\EE}{\sigma}$ with $\EE=\FF\lr{t_1}\dots\lr{t_e}$ be a simple
\rpisiSE-extension of a constant-stable and computable field $\dfield{\FF}{\sigma}$. Suppose that for all \rE-monomials the periods are positive, and the orders and periods of the \rE-monomials are given explicitly. Then Problem~PFLDE in $\dfield{\EE}{\sigma}$ for
$\dgroup{(\FF^*)}{\EE}{\FF}$ is solvable if one of the
following holds.
\begin{enumerate}
\item All $t_i$ are \rE\sigmaSE-monomials and 
PFLDE is solvable in
$\dfield{\FF}{\sigma^k}$ for $\FF^*$ for all $k>0$.
\item Problem~PMT is solvable in
$\dfield{\FF}{\sigma}$ for $\FF^*$ and Problem~PFLDE is solvable
in
$\dfield{\FF}{\sigma^k}$ for $\FF^*$ for all $k>0$.
\end{enumerate}
\end{theorem}

\begin{proof}
Let $H=\dgroup{(\FF^*)}{\EE}{\FF}$. Recall that by Theorem~\ref{Thm:sconstIsGroupField} we have that $\sconst{H}{\EE}{\sigma}\setminus\{0\}\leq\EE^*$, i.e., Problem~PFLDE is applicable in $\dfield{\EE}{\sigma}$ for $H$.
By Lemma~\ref{Lemma:ReorderSimpleRPISI} we can reorder the generators of
the \rpisiSE-extension such that $\dfield{\bar{\EE}}{\sigma}$ is an $\FF^*$-simple
\rE-extension of $\dfield{\FF}{\sigma}$ and
$\dfield{\EE}{\sigma}$ is a $G$-simple \pisiSE-extension of
$\dfield{\bar{\EE}}{\sigma}$ with $G=\dgroup{(\FF^*)}{\bar{\EE}}{\FF}$. Note that the multiplicative group $\FF^*$ is closed under $\sigma$, $\sconstF{(\FF^*)}{\FF}{\sigma^l}=\FF^*$ for all $l>0$ and 
$\sconst{(\FF^*)}{\bar{\EE}}{\sigma}\setminus\{0\}\leq\bar{\EE}^*$ by Corollary~\ref{Cor:SConstRExtension}. Thus we can apply Proposition~\ref{Prop:LiftForNestedRExt}. Hence Problem~PFLDE is solvable in $\dfield{\bar{\EE}}{\sigma}$ for $G$. If we are in case (1), i.e., no \piE-monomials occur, we can apply iteratively Theorem~\ref{Thm:PSReduction} and obtain an algorithm to solve Problem PFLDE in $\dfield{\EE}{\sigma}$ for $\dgroup{G}{\EE}{\bar{\EE}}=H$. If we are in case (2), i.e., \piE-monomials may occur, we exploit in addition our assumptions together with Theorem~\ref{Thm:AlgMainResultFull}.(2). This shows that we can solve Problem PMT in $\dfield{\EE}{\sigma}$ for $H$ (and in any sub-difference ring by truncating the tower of extensions). Again the iterative application of Theorem~\ref{Thm:PSReduction} shows that Problem~PFLDE is solvable in $\dfield{\EE}{\sigma}$ for $\dgroup{G}{\EE}{\bar{\EE}}=H$.
\end{proof}

\ExternalProof{(Theorem~\ref{Thm:AlgMainResultFull}.(3)\label{Proof:AlgMainResultFull3})}{
Let $\dfield{\EE}{\sigma}$ be a simple \rpisiSE-extension of $\dfield{\FF}{\sigma}$ where $\dfield{\FF}{\sigma}$ is computable and strong constant-stable. Then by Corollary~\ref{Cor:SimpleNestedRBasics} (parts 3 and 4) the periods and orders of all \rE-monomials are positive and can be computed. Thus Theorem~\ref{Thm:PFLDEStrongConstantStable}.(2) is applicable which completes the proof.}

\noindent We remark that in Theorem~\ref{Thm:AlgMainResultFull}.(3) one can drop the condition that Problem~PMT is solvable in $\dfield{\FF}{\sigma}$ for $\FF^*$ if in the \rpisiSE-extension no \piE-monomials occur, i.e., one applies part one and not part two of Theorem~\ref{Thm:PFLDEStrongConstantStable}.

\section{Conclusion}\label{Sec:Conclusion}

We provided important building blocks that extend the well established difference field theory to a difference ring theory. In this setting one can handle in addition objects such as~\eqref{Equ:Equ:AlgebraicObjects}. We elaborated algorithms for the (multiplicative) telescoping problem (Problems T and MT) and the (multiplicative) parameterized telescoping problem (Problems PT and PMT). In particular, Problem PT enables one to apply Zeilberger's creative telescoping paradigm in the rather general class of simple \rpisiSE-extensions. In order to derive these algorithms we showed that certain semi-constants (resp.\ semi-invariants) of the difference rings under consideration form a multiplicative group.

Currently, the underlying engine of Theorem~\ref{Thm:AlgMainResultRestricted} with the ground field machinery of Subsection~\ref{Subsec:GroundFieldAlg} is fully implemented within the summation package\footnote{The \texttt{Sigma} package can be downloaded from \texttt{www.risc.jku.at/research/combinat/software/Sigma/}} \texttt{Sigma}. In this way one can treat big classes of indefinite nested sums and products involving algebraic objects like $(-1)^k$. In particular, one can treat d'Alembertian solutions of linear recurrences as worked out in Subsection~\ref{Subsec:dAlembert}. We emphasize that these algorithms are enhanced by the refinements described in~\cite{Schneider:04a,Schneider:05c,Schneider:07d,Schneider:08c,Petkov:10,Schneider:14} in order to find sum representations with certain optimality criteria, like optimal nesting depth.\\
The machinery to handle nested \rE-extensions (see Theorem~\ref{Thm:AlgMainResultFull}) is not incorporated in \texttt{Sigma} yet. First, further investigations will be necessary so that the new algorithms can be merged with the difference field enhancements of \texttt{Sigma}. 

Another challenging task is to push forward the difference ring theory and the underlying algorithms in order to relax the requirements in Theorems~\ref{Thm:AlgMainResultRestricted} and~\ref{Thm:AlgMainResultFull} that the \rpisiSE-extensions are simple and/or that the ground difference ring is strong constant-stable. In this regard, we refer to the comments given in Example~\ref{Exp:NotSimpleExt}. 

In any case, the currently developed toolbox widens the class of indefinite nested sums and products in the setting of difference rings. We are looking forward to see new kinds of applications that can be attacked with this machinery.

\section*{Acknowledgement}

\noindent I would like to thank Michael Singer and the anonymous referee for their helpful comments and suggestions to improve the presentation of this article. 

\section*{Appendix: A short index}
\begin{multicols}{2}
%\printindex
\begin{theindex}

  \item $M(\vect{f},\AA)$, 10
  \item $V(u,\vect{f},\AA)$, 11
  \item $\AA\lr{t}$, 7
  \item $\AA\lr{t}_{a,b}$, 17
  \item $\deg$, 17
  \item $\dfield{\AA}{\sigma}\leq\dfield{\tilde{A}}{\tilde{\sigma}}$, 5
  \item $\dgroup{G}{\AA}{\GG}$, 12
  \item $\ford(f)$, 27
  \item $\gsconstF{\AA}{\sigma}$, $\gsconst{\AA}{\sigma}$, 10
  \item $\langle S\rangle$, 4
  \item $\ldeg$, 17
  \item $\ord(f)$, 6
  \item $\per(f)$, 27
  \item $\sconstF{G}{\AA}{\sigma}$, $\sconst{G}{\AA}{\sigma}$, 10
  \item $\sigmaFac{f}{k,\sigma}$, $\sigmaFac{f}{k}$, 17
  \item \pisiSE-field, 7

  \indexspace

  \item constant field/ring, 4

  \indexspace

  \item difference
    \subitem ring/field, 4
    \subitem ring/field extension, 5

  \indexspace

  \item extension
    \subitem (nested) \piE,\sigmaSE,\rE,\rE\piE,\\ \rE\sigmaSE,\pisiSE,\rpisiSE, 
		7
    \subitem \piE, 5
    \subitem \rE, 7
    \subitem \sigmaSE, 5
    \subitem algebraic, 6
    \subitem simple, $G$-simple, 12
    \subitem single-rooted, 13
    \subitem unimonomial, 5

  \indexspace

  \item function
    \subitem degree, 17
    \subitem factorial order, 27
    \subitem order, 6
    \subitem period, 27
    \subitem rising factorial, 17

  \indexspace

  \item monomial
    \subitem \piE,\sigmaSE,\rE,\rE\piE,\rE\sigmaSE,\pisiSE,\rpisiSE, 7
    \subitem simple, $G$-simple, 12

  \indexspace

  \item Problem
    \subitem FPLDE, 11
    \subitem MT, 8
    \subitem O, 8
    \subitem PMT, 11
    \subitem PT, 4
    \subitem T, 4
  \item product group, 12

  \indexspace

  \item ring
    \subitem (strong) constant-stable, 14
    \subitem connected, 16
    \subitem constant-stable, 14
    \subitem reduced, 16

  \indexspace

  \item semi-constant, 10

\end{theindex}

\end{multicols}

\bibliographystyle{plain}

\begin{thebibliography}{10}

\bibitem{Physics3}
J.~Ablinger, A.~Behring, J.~Bl{\"u}mlein, A.~De Freitas, A.~Hasselhuhn, A.~von
  Manteuffel, M.~Round, C.~Schneider, and F.~Wissbrock.
\newblock The 3-loop non-singlet heavy flavor contributions and anomalous
  dimensions for the structure function {$F_2(x,Q^2)$} and transversity.
\newblock {\em Nucl. Phys. B}, 886:733--823, 2014.
\newblock arXiv:1406.4654 [hep-ph].

\bibitem{Physics2}
J.~Ablinger, J.~Bl{\"u}mlein, A.~De Freitas~A. Hasselhuhn, A.~von Manteuffel,
  M.~Round, C.~Schneider, and F.~Wissbrock.
\newblock The transition matrix element {$A_{gq}(N)$} of the variable flavor
  number scheme at {$O(\alpha_s^3)$}.
\newblock {\em Nucl. Phys. B}, 882:263--288, 2014.
\newblock arXiv:1402.0359 [hep-ph].

\bibitem{ABS:11}
J.~Ablinger, J.~Bl\"umlein, and C.~Schneider.
\newblock Harmonic sums and polylogarithms generated by cyclotomic polynomials.
\newblock {\em J. Math. Phys.}, 52(10):1--52, 2011.
\newblock [arXiv:1007.0375 [hep-ph]].

\bibitem{ABS:13}
J.~Ablinger, J.~Bl\"umlein, and C.~Schneider.
\newblock Analytic and algorithmic aspects of generalized harmonic sums and
  polylogarithms.
\newblock {\em J. Math. Phys.}, 54(8):1--74, 2013.
\newblock arXiv:1302.0378 [math-ph].

\bibitem{Abramov:71}
S.~A. Abramov.
\newblock On the summation of rational functions.
\newblock {\em Zh. vychisl. mat. Fiz.}, 11:1071--1074, 1971.

\bibitem{ABPS:14}
S.~A. Abramov, M.~Bronstein, M.~Petkov\v{s}ek, and C.~Schneider.
\newblock {\em In preparation}, 2015.

\bibitem{Abramov:94}
S.~A. Abramov and M.~Petkov{\v s}ek.
\newblock D'{A}lembertian solutions of linear differential and difference
  equations.
\newblock In J.~von~zur Gathen, editor, {\em Proc. ISSAC'94}, pages 169--174.
  ACM Press, 1994.

\bibitem{Petkov:10}
S.~A. Abramov and M.~Petkov{\v{s}}ek.
\newblock Polynomial ring automorphisms, rational {$(w,\sigma)$}-canonical
  forms, and the assignment problem.
\newblock {\em J. Symbolic Comput.}, 45(6):684--708, 2010.

\bibitem{Bauer:99}
A.~Bauer and M.~Petkov{\v{s}}ek.
\newblock Multibasic and mixed hypergeometric {Gosper}-type algorithms.
\newblock {\em J.~Symbolic Comput.}, 28(4--5):711--736, 1999.

\bibitem{BeckerWeispfenning:93}
Thomas Becker and Volker Weispfenning.
\newblock {\em Gr\"obner bases}, volume 141 of {\em Graduate Texts in
  Mathematics}.
\newblock Springer-Verlag, New York, 1993.
\newblock A computational approach to commutative algebra, In cooperation with
  Heinz Kredel.

\bibitem{Physics1}
J.~Bl\"umlein, A.~Hasselhuhn, S.~Klein, and C.~Schneider.
\newblock The {$O(\alpha_s^3 n_f T_F^2 C_{A,F})$} contributions to the gluonic
  massive operator matrix elements.
\newblock {\em Nucl. Phys. B}, 866:196--211, 2013.
\newblock [arXiv:1205.4184 [hep-ph]].

\bibitem{BKSF:12}
J.~Bl\"umlein, S.~Klein, C.~Schneider, and F.~Stan.
\newblock A symbolic summation approach to {F}eynman integral calculus.
\newblock {\em J. Symbolic Comput.}, 47:1267--1289, 2012.

\bibitem{Bron:97}
M.~Bronstein.
\newblock {\em Symbolic Integration {I}, {T}ranscendental functions}.
\newblock Springer, Berlin-Heidelberg, 1997.

\bibitem{Bron:00}
M.~Bronstein.
\newblock On solutions of linear ordinary difference equations in their
  coefficient field.
\newblock {\em J.~Symbolic Comput.}, 29(6):841--877, 2000.

\bibitem{Chyzak:00}
F.~Chyzak.
\newblock An extension of {Z}eilberger's fast algorithm to general holonomic
  functions.
\newblock {\em Discrete Math.}, 217:115--134, 2000.

\bibitem{Cohn:65}
R.~M. Cohn.
\newblock {\em Difference Algebra}.
\newblock Interscience Publishers, John Wiley \& Sons, 1965.

\bibitem{Erocal:11}
B:~Er{\"o}cal.
\newblock {\em Algebraic extensions for summation in finite terms}.
\newblock PhD thesis, RISC, Johannes Kepler University, Linz, February 2011.

\bibitem{Gosper:78}
R.~W. Gosper.
\newblock Decision procedures for indefinite hypergeometric summation.
\newblock {\em Proc. Nat. Acad. Sci. U.S.A.}, 75:40--42, 1978.

\bibitem{Singer:08}
C.~Hardouin and M.F. Singer.
\newblock Differential {G}alois theory of linear difference equations.
\newblock {\em Math. Ann.}, 342(2):333--377, 2008.

\bibitem{Singer:99}
P.~A. Hendriks and M.~F. Singer.
\newblock Solving difference equations in finite terms.
\newblock {\em J.~Symbolic Comput.}, 27(3):239--259, 1999.

\bibitem{vanHoeij:99}
{M.~van} Hoeij.
\newblock Finite singularities and hypergeometric solutions of linear
  recurrence equations.
\newblock {\em J.~Pure Appl. Algebra}, 139(1-3):109--131, 1999.

\bibitem{HK:12}
P.~Horn, W.~Koepf, and T.~Sprenger.
\newblock {$m$}-fold hypergeometric solutions of linear recurrence equations
  revisited.
\newblock {\em Math. Comput. Sci.}, 6(1):61--77, 2012.

\bibitem{Karpilovsky:83}
G.~Karpilovsky.
\newblock On finite generation of unit groups of commutative group rings.
\newblock {\em Arch. Math. (Basel)}, 40(6):503--508, 1983.

\bibitem{Karr:81}
M.~Karr.
\newblock Summation in finite terms.
\newblock {\em J.~ACM}, 28:305--350, 1981.

\bibitem{Karr:85}
M.~Karr.
\newblock Theory of summation in finite terms.
\newblock {\em J.~Symbolic Comput.}, 1:303--315, 1985.

\bibitem{KP:11}
M.~Kauers and P.~Paule.
\newblock {\em The concrete tetrahedron}.
\newblock Texts and Monographs in Symbolic Computation. SpringerWienNewYork,
  Vienna, 2011.
\newblock Symbolic sums, recurrence equations, generating functions, asymptotic
  estimates.

\bibitem{Schneider:06e}
M.~Kauers and C.~Schneider.
\newblock Application of unspecified sequences in symbolic summation.
\newblock In J.G. Dumas, editor, {\em Proc. ISSAC'06.}, pages 177--183. ACM
  Press, 2006.

\bibitem{Schneider:06d}
M.~Kauers and C.~Schneider.
\newblock Indefinite summation with unspecified summands.
\newblock {\em Discrete Math.}, 306(17):2021--2140, 2006.

\bibitem{Schneider:07f}
M.~Kauers and C.~Schneider.
\newblock Symbolic summation with radical expressions.
\newblock In C.W. Brown, editor, {\em {Proc. ISSAC'07}}, pages 219--226, 2007.

\bibitem{Koutschan:13}
C.~Koutschan.
\newblock Creative telescoping for holonomic functions.
\newblock In C.~Schneider and J.~Bl\"umlein, editors, {\em {Computer Algebra in
  Quantum Field Theory: Integration, Summation and Special Functions}}, Texts
  and Monographs in Symbolic Computation, pages 171--194. Springer, 2013.
\newblock arXiv:1307.4554 [cs.SC].

\bibitem{Levin:08}
A.~Levin.
\newblock {\em Difference algebra}, volume~8 of {\em Algebra and Applications}.
\newblock Springer, New York, 2008.

\bibitem{Neher:09}
Erhard Neher.
\newblock Invertible and nilpotent elements in the group algebra of a unique
  product group.
\newblock {\em Acta Appl. Math.}, 108(1):135--139, 2009.

\bibitem{Schneider:09a}
R.~Osburn and C.~Schneider.
\newblock Gaussian hypergeometric series and extensions of supercongruences.
\newblock {\em Math. Comp.}, 78(265):275--292, 2009.

\bibitem{Paule:95}
P.~Paule.
\newblock Greatest factorial factorization and symbolic summation.
\newblock {\em J.~Symbolic Comput.}, 20(3):235--268, 1995.

\bibitem{PauleRiese:97}
P.~Paule and A.~Riese.
\newblock A {M}athematica q-analogue of {Z}eil\-ber\-ger's algorithm based on
  an algebraically motivated aproach to $q$-hypergeometric telescoping.
\newblock In M.~Ismail and M.~Rahman, editors, {\em Special Functions, q-Series
  and Related Topics}, volume~14, pages 179--210. AMS, 1997.

\bibitem{Petkov:92}
M.~Petkov{\v s}ek.
\newblock Hypergeometric solutions of linear recurrences with polynomial
  coefficients.
\newblock {\em J.~Symbolic Comput.}, 14(2-3):243--264, 1992.

\bibitem{AequalB}
M.~Petkov{\v s}ek, H.~S. Wilf, and D.~Zeilberger.
\newblock {\em $A=B$}.
\newblock A. K. Peters, Wellesley, MA, 1996.

\bibitem{Petkov:2013}
M.~Petkov{\v s}ek and H.~Zakraj{\v s}ek.
\newblock Solving linear recurrence equations with polynomial coefficients.
\newblock In C.~Schneider and J.~Bl\"umlein, editors, {\em Computer Algebra in
  Quantum Field Theory: Integration, Summation and Special Functions}, Texts
  and Monographs in Symbolic Computation, pages 259--284. Springer, 2013.

\bibitem{PSW:11}
H.~Prodinger, C.~Schneider, and S.~Wagner.
\newblock {Unfair permutations}.
\newblock {\em Europ. J. Comb.}, 32:1282--1298, 2011.

\bibitem{Risch:69}
R.~Risch.
\newblock The problem of integration in finite terms.
\newblock {\em Trans. Amer. Math. Soc.}, 139:167--189, 1969.

\bibitem{Schneider:00}
C.~Schneider.
\newblock {An Implementation of {K}arr's Summation Algorithm in Mathematica}.
\newblock {\em Sem. Lothar. Combin.}, S43b:1--10, 2000.

\bibitem{Schneider:01}
C.~Schneider.
\newblock Symbolic summation in difference fields.
\newblock Technical Report 01-17, RISC-Linz, J.~Kepler University, November
  2001.
\newblock PhD Thesis.

\bibitem{Schneider:04b}
C.~Schneider.
\newblock A collection of denominator bounds to solve parameterized linear
  difference equations in ${\Pi}{\Sigma}$-extensions.
\newblock {\em An. Univ. Timi\c soara Ser. Mat.-Inform.}, 42(2):163--179, 2004.
\newblock Extended version of Proc. SYNASC'04.

\bibitem{Schneider:04c}
C.~Schneider.
\newblock The summation package {S}igma: Underlying principles and a rhombus
  tiling application.
\newblock {\em Discrete Math. Theor. Comput. Sci.}, 6:365--386, 2004.

\bibitem{Schneider:04a}
C.~Schneider.
\newblock Symbolic summation with single-nested sum extensions.
\newblock In J.~Gutierrez, editor, {\em Proc. ISSAC'04}, pages 282--289. ACM
  Press, 2004.

\bibitem{Schneider:05b}
C.~Schneider.
\newblock Degree bounds to find polynomial solutions of parameterized linear
  difference equations in ${\Pi}{\Sigma}$-fields.
\newblock {\em Appl. Algebra Engrg. Comm. Comput.}, 16(1):1--32, 2005.

\bibitem{Schneider:05d}
C.~Schneider.
\newblock A new {S}igma approach to multi-summation.
\newblock {\em Adv. in Appl. Math.}, 34(4):740--767, 2005.

\bibitem{Schneider:05c}
C.~Schneider.
\newblock Product representations in ${\Pi}{\Sigma}$-fields.
\newblock {\em Ann. Comb.}, 9(1):75--99, 2005.

\bibitem{Schneider:05a}
C.~Schneider.
\newblock Solving parameterized linear difference equations in terms of
  indefinite nested sums and products.
\newblock {\em J. Differ. Equations Appl.}, 11(9):799--821, 2005.

\bibitem{Schneider:07b}
C.~Schneider.
\newblock {A}p\'{e}ry's double sum is plain sailing indeed.
\newblock {\em Electron. J.~Combin.}, 14, 2007.

\bibitem{Schneider:07d}
C.~Schneider.
\newblock Simplifying sums in {$\Pi\Sigma$}-extensions.
\newblock {\em J. Algebra Appl.}, 6(3):415--441, 2007.

\bibitem{Schneider:07a}
C.~Schneider.
\newblock Symbolic summation assists combinatorics.
\newblock {\em S\'em.~Lothar. Combin.}, 56:1--36, 2007.
\newblock Article B56b.

\bibitem{Schneider:08c}
C.~Schneider.
\newblock A refined difference field theory for symbolic summation.
\newblock {\em J. Symbolic Comput.}, 43(9):611--644, 2008.
\newblock [arXiv:0808.2543v1].

\bibitem{Schneider:10c}
C.~Schneider.
\newblock Parameterized telescoping proves algebraic independence of sums.
\newblock {\em Ann. Comb.}, 14(4):533--552, 2010.
\newblock [arXiv:0808.2596].

\bibitem{Schneider:10a}
C.~Schneider.
\newblock Structural theorems for symbolic summation.
\newblock {\em Appl. Algebra Engrg. Comm. Comput.}, 21(1):1--32, 2010.

\bibitem{Schneider:10b}
C.~Schneider.
\newblock A symbolic summation approach to find optimal nested sum
  representations.
\newblock In A.~Carey, D.~Ellwood, S.~Paycha, and S.~Rosenberg, editors, {\em
  {Motives, Quantum Field Theory, and Pseudodifferential Operators}}, volume~12
  of {\em Clay Mathematics Proceedings}, pages 285--308. Amer. Math. Soc, 2010.
\newblock arXiv:0808.2543.

\bibitem{Schneider:13a}
C.~Schneider.
\newblock Simplifying multiple sums in difference fields.
\newblock In C.~Schneider and J.~Bl\"umlein, editors, {\em {Computer Algebra in
  Quantum Field Theory: Integration, Summation and Special Functions}}, Texts
  and Monographs in Symbolic Computation, pages 325--360. Springer, 2013.
\newblock arXiv:1304.4134 [cs.SC].

\bibitem{Schneider:13b}
C.~Schneider.
\newblock Modern summation methods for loop integrals in quantum field theory:
  The packages {S}igma, {EvaluateMultiSums} and {SumProduction}.
\newblock In {\em {Proc. ACAT 2013}}, volume 523 of {\em J. Phys.: Conf. Ser.},
  pages 1--17, 2014.
\newblock arXiv:1310.0160 [cs.SC].

\bibitem{Schneider:14}
C.~Schneider.
\newblock Fast algorithms for refined parameterized telescoping in difference
  fields.
\newblock In M.~Weimann J.~Guitierrez, J.~Schicho, editor, {\em Computer Algebra and Polynomials, Applications of Algebra and Number Theory}, Lecture Notes in Computer Science (LNCS), pages 157-191. Springer, 2015.
\newblock arXiv:1307.7887 [cs.SC].


\bibitem{Singer:97}
M.~van~der Put and M.F. Singer.
\newblock {\em Galois theory of difference equations}, volume 1666 of {\em
  Lecture Notes in Mathematics}.
\newblock Springer-Verlag, Berlin, 1997.

\bibitem{Zeilberger:90a}
D.~Zeilberger.
\newblock A holonomic systems approach to special functions identities.
\newblock {\em J.~Comput. Appl. Math.}, 32:321--368, 1990.

\bibitem{Zeilberger:91}
D.~Zeilberger.
\newblock The method of creative telescoping.
\newblock {\em J.~Symbolic Comput.}, 11:195--204, 1991.

\end{thebibliography}

\end{document}